\documentclass[reprint,amsfonts, amssymb, amsmath,  showkeys, nofootinbib,pra, superscriptaddress, twocolumn,longbibliography]{revtex4-2}
\usepackage{float}
\makeatletter
\let\newfloat\newfloat@ltx
\makeatother
\usepackage[english]{babel}
\usepackage[utf8]{inputenc}
\usepackage{graphics}
\usepackage{selinput}

\usepackage{braket}
\usepackage{amsthm}
\usepackage{mathtools}
\usepackage{physics}
\usepackage[dvipsnames]{xcolor}
\usepackage{graphicx}
\usepackage[left=16mm,right=16mm,top=35mm,columnsep=15pt]{geometry} 
\usepackage{adjustbox}
\usepackage{placeins}
\usepackage[T1]{fontenc}
\usepackage{lipsum}
\usepackage{csquotes}
\usepackage{mathbbol} 
\usepackage[linesnumbered,ruled,vlined]{algorithm2e}
\SetKwInput{kwInit}{Init}

\def\HC{\mathcal{H}}
\def\KC{\mathcal{K}}
\def\LC{\mathcal{L}}

\def\ad{^{\dagger}}

\usepackage[makeroom]{cancel}
\usepackage[toc,page]{appendix}
\usepackage[colorlinks=true,citecolor=blue,linkcolor=magenta]{hyperref}

\usepackage{tikz}
\tikzset{every picture/.style=remember picture}

\newcommand{\dya}[1]{\ket{#1}\!\bra{#1}}

\newcommand{\poly}{\operatorname{poly}}

\newcommand{\Ebb}{\mathbb{E}}

\newcommand{\Ubb}{\mathbb{U}}

\newcommand{\AC}{\mathcal{A}}
\newcommand{\BC}{\mathcal{B}}

\newcommand{\DC}{\mathcal{D}}

\newcommand{\IC}{\mathcal{I}}
\newcommand{\MC}{\mathcal{M}}
\newcommand{\NC}{\mathcal{N}}
\newcommand{\OC}{\mathcal{O}}
\newcommand{\PC}{\mathcal{P}}
\newcommand{\QC}{\mathcal{Q}}

\newcommand{\SC}{\mathcal{S}}

\newcommand{\UC}{\mathcal{U}}
\newcommand{\VC}{\mathcal{V}}
\newcommand{\WC}{\mathcal{W}}
\newcommand{\XC}{\mathcal{X}}

\newcommand{\Var}{{\rm Var}}
\newcommand{\Cov}{{\rm Cov}}

\renewcommand{\geq}{\geqslant}
\renewcommand{\leq}{\leqslant}

\renewcommand{\vec}[1]{\boldsymbol{#1}}  

\newcommand{\bs}{\textsf{BS}}

\def\be{\begin{equation}}
\def\ee{\end{equation}}
\def\bs{\begin{split}}
\def\e{\end{split}}
\def\ba{\begin{eqnarray}}
\def\bea{\begin{eqnarray}}

\def\tea{\end{eqnarray}}
\def\ea{\end{eqnarray}}
\def\eea{\end{eqnarray}}

\newtheorem{theorem}{Theorem}
\newtheorem{lemma}{Supplemental Lemma}
\newtheorem{corollary}{Corollary}

\newtheorem{proposition}{Proposition}
\newtheorem*{proposition*}{Proposition}
\newtheorem{supplemental_proposition}{Supplemental Proposition}
\newtheorem{supplemental_corollary}{Supplemental Corollary}
\newtheorem{supplemental_theorem}{Supplemental Theorem}

\newtheorem{definition}{Definition}

\newcommand{\sth}[1]{{\color{black} #1}}

\usepackage{amssymb}
\usepackage{dsfont}

\usepackage[normalem]{ulem}
\usepackage{hyperref}

\def\be{\begin{equation}}
\def\te{\end{equation}}
\def\ee{\end{equation}}
\def\ba{\begin{eqnarray}}
\def\bea{\begin{eqnarray}}

\def\tea{\end{eqnarray}}
\def\ea{\end{eqnarray}}
\def\eea{\end{eqnarray}}

\begin{document}

\title{Exponential concentration in quantum kernel methods}

\author{Supanut Thanasilp}
\affiliation{Centre for Quantum Technologies, National University of Singapore, 3 Science Drive 2 117543, Singapore.}
\affiliation{Institute of Physics, Ecole Polytechnique F\'{e}d\'{e}rale de Lausanne (EPFL), CH-1015 Lausanne, Switzerland}

\author{Samson Wang}
\affiliation{Imperial College London, London, UK.}

\author{M. Cerezo}
\affiliation{Information Sciences, Los Alamos National Laboratory, Los Alamos, NM, USA.}
\affiliation{Quantum Science Center, Oak Ridge, TN 37931, USA}

\author{Zo\"{e} Holmes}
\affiliation{Information Sciences, Los Alamos National Laboratory, Los Alamos, NM, USA.}
\affiliation{Institute of Physics, Ecole Polytechnique F\'{e}d\'{e}rale de Lausanne (EPFL), CH-1015 Lausanne, Switzerland}

\date{\today}

\begin{abstract}
Kernel methods in Quantum Machine Learning (QML) have recently gained significant attention as a potential candidate for achieving a quantum advantage in data analysis. Among other attractive properties, when training a kernel-based model one is guaranteed to find the optimal model's parameters due to the convexity of the training landscape. However, this is based on the assumption that the quantum kernel can be efficiently obtained from quantum hardware.
In this work we study the performance of quantum kernel models from the perspective of the resources needed to accurately estimate kernel values. 
We show that, under certain conditions, values of quantum kernels over different input data can be exponentially concentrated (in the number of qubits) towards some fixed value. Thus on training with a polynomial number of measurements, one ends up with a trivial model where the predictions on unseen inputs are \textit{independent of the input data}.
We identify four sources that can lead to concentration including: \textit{expressivity of data embedding, global measurements, entanglement} and \textit{noise}. For each source, an associated concentration bound of quantum kernels is analytically derived.
Lastly, we show that when dealing with classical data, 
training a parametrized data embedding with a kernel alignment method is also susceptible to exponential concentration. Our results are verified through numerical simulations for several QML tasks. Altogether, we provide guidelines indicating that certain features should be avoided to ensure the efficient evaluation of quantum kernels and so the performance of quantum kernel methods. 

\end{abstract}

\maketitle

\section{Introduction}

Quantum machine learning (QML) has generated tremendous amounts of excitement, but it is important not to over-hype its potential.
On the one hand, a family of impressive results have recently established a provable separation between the power of classical and quantum machine learning methods in a range of contexts~\cite{biamonte2017quantum,huang2021quantum,huang2021information,aharonov2021quantum,sweke2021on,huang2021power,kubler2021inductive,liu2021rigorous, jager2023universal, wu2023quantum}. On the other, many proposals remain heuristic and there are significant questions yet to be answered on the efficient scalability of QML methods.

Quantum kernel methods, which involve embedding classical data into quantum states and then computing their inner-products (i.e., their kernels), or in the case of quantum data directly computing input state overlaps, are widely viewed as particularly promising family of QML algorithms to achieve a practical quantum advantage. 
To ensure provable quantum speed-up over classical algorithms, the key is to construct the embedding (also called a quantum feature map) that is capable of recognizing classically intractable complex patterns~\cite{huang2021power,kubler2021inductive,liu2021rigorous}. 
Quantum kernels are expected to find use in a mix of scientific and practical applications including classifying types of supernovae in cosmology~\cite{peters2021machine}, probing phase transitions in quantum many-body physics~\cite{sancho2022quantum} and detecting fraud in finance~\cite{kyriienko2022unsupervised}. Moreover, kernel methods are famously said to enjoy trainability guarantees due to the convexity of their loss landscapes~\cite{schuld2022is,schuld2021supervised,gentinetta2022complexity,hofmann2008kernel}.

This is in contrast to Quantum Neural Networks (QNNs) where the loss landscape is generally non-convex~\cite{cerezo2020variationalreview,huembeli2021characterizing} and can exhibit Barren Plateaus (BPs).
A barren plateau is a cost landscape where the magnitudes
of gradients vanish exponentially with growing problem size~\cite{mcclean2018barren,cerezo2020cost,larocca2021diagnosing,marrero2020entanglement,patti2020entanglement,holmes2021connecting,holmes2021barren,zhao2021analyzing,wang2020noise,thanasilp2021subtleties,cerezo2020impact,arrasmith2020effect,wang2021can}. 
There are a number of causes that can lead to barren plateaus, including using variational ansatze that are too expressive~\cite{mcclean2018barren,holmes2020barren,tangpanitanon2020expressibility} or too entangling~\cite{marrero2020entanglement,sharma2020trainability}. However, barren plateaus can even arise for inexpressive and low-entangling QNNs if the cost function relies on measuring global properties of the system~\cite{cerezo2020cost} or if the training dataset is too random~\cite{thanasilp2021subtleties,li2022concentration}. Hardware errors can also wash out landscape features leading to noise-induced barren plateaus~\cite{wang2020noise,franca2020limitations}. 

Here we argue that quantum kernel methods experience a similar barrier to barren plateaus. Crucially, the trainability guarantees enjoyed by kernel methods only become meaningful when the values of the kernel can be efficiently estimated to a sufficient precision such that the statistical estimates contain information about the input data. We show that under certain conditions, the value of quantum kernels can exponentially concentrate (with increasing number of qubits) around a fixed value. In such cases, the number of shots required to resolve the kernels to a sufficiently high accuracy scales exponentially. This indicates that the efficient evaluation of quantum kernels cannot always be taken for granted. Consequently, when the kernel values are estimated with a polynomial number of measurements, the trained model with high probability becomes independent of input data. That is, the predictions of the model on unseen data are the same for any target problem that suffers from exponential concentration and thus the learnt model is, for all intents and purposes, useless. This is summarized in Fig.~\ref{fig:summary}.

\begin{figure*}
    \centering
    \includegraphics[width=0.95\textwidth]{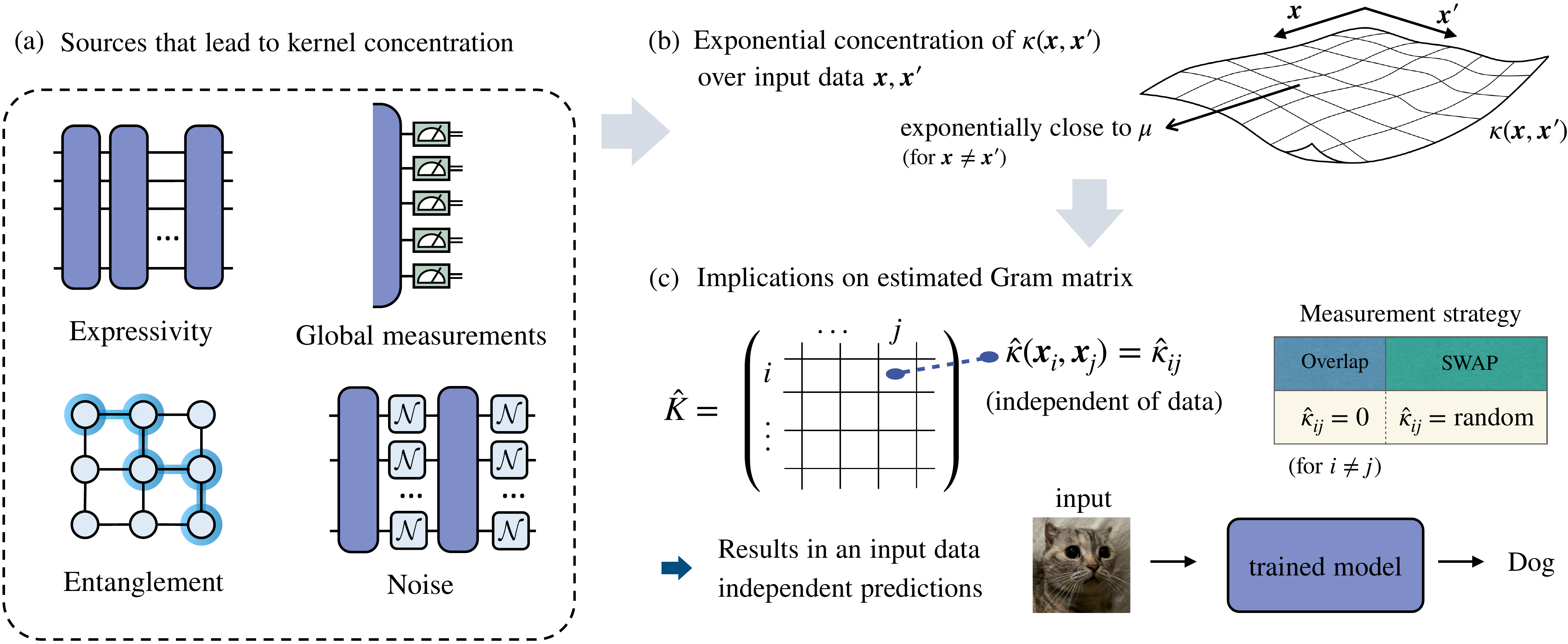}
    \caption{\textbf{Exponential concentration and its implications on kernel methods:} 
    The exponential concentration (in the number of qubits $n$) of quantum kernels $\kappa(\vec{x},\vec{x'})$, over all possible input data pairs $\vec{x},\vec{x'}$, can be seen to stem from the difficulty of information extraction from data quantum states due to various sources (illustrated in panels (a) and (b)). 
    The kernel concentration has a detrimental impact on the performance of quantum kernel-based methods. As shown in panel (c), for a polynomial (in $n$) number of measurement shots, the statistical estimates of the off-diagonal elements in the Gram matrix $\hat{\kappa}(\vec{x}_i,\vec{x}_j)$ contain no information about the input data (with high probability) i.e., each $\hat{\kappa}(\vec{x}_i,\vec{x}_j) = \hat{\kappa}_{ij}$. The exact behaviour of the estimated kernel value depends on the measurement strategy: for the Loschmidt Echo test (i.e., the overlap test), $\hat{\kappa}_{ij}$ concentrates to $0$ for $i \neq j$ (corresponding to the estimated Gram matrix being an identity $\mathbb{1}$) and for the SWAP test $\hat{\kappa}_{i,j}$ for $i \neq j$ is indistinguishable from a data-independent random variable (corresponding to the estimated Gram matrix being a random matrix). Ultimately, this leads to a trivial model where the predictions on unseen inputs are independent of the training data.} 
    \label{fig:summary}
\end{figure*}

\begin{table}[t]
\centering
\begin{tabular}{||c | c | c||}
    \hline
    Sources of concentration  & Kernels & QNNs\\
    \hline \hline
    Expressivity & Theorem~\ref{thm:expressivity-kernel}  & Ref.~\cite{mcclean2018barren,holmes2021connecting}   \\
    \hline
    Global measurements & Proposition~\ref{prop:global-measurement}  & Ref.~\cite{cerezo2020cost} \\
    \hline
    Entanglement & Corollary~\ref{coro:entanglement-induced}  & Ref.~\cite{marrero2020entanglement,sharma2020trainability}, \\
    \hline
    Noise & Theorem~\ref{thm:noise-kernel} & Ref.~\cite{wang2020noise}\\
    \hline
\end{tabular}
\caption{\textbf{Summary of our main results:}
This table summarizes our key analytical results on different sources that lead to the exponential concentration in quantum kernels as compared with BP results of QNNs in the literature.}
\label{table:KernelsVsQML}
\end{table}

This concentration of quantum kernels can in broad terms be viewed as a result of the fact that it can be extremely difficult to extract any useful information from the (necessarily) exponentially large Hilbert space (especially in the presence of noise). 
We show that analogous to the causes of BPs for QNNs there are at least three different mechanisms that can lead to the exponential concentration of the encoded quantum states, including (i) the expressivity of the encoded quantum state ensemble, (ii) the entanglement in encoded quantum states with a local observable and (iii) the effect of noise. We further show that for the case of the commonly used fidelity kernel~\cite{schuld2021supervised,havlivcek2019supervised}, the dependence of global measurements to evaluate the kernel can lead to exponential concentration even when the expressivity of the embedding and the entanglement of the data states are low. 
In all cases, we establish exponential concentration by deriving an analytic bound (summarized in Table~\ref{table:KernelsVsQML}). We further provide numerical results demonstrating these effects for different learning tasks.

Our work on embedding-induced concentration suggests that problem-inspired embeddings should be used over problem-agnostic embeddings (which are typically highly expressive and entangling). 
For instance, one can construct embeddings encoding the geometrical properties of the data~\cite{larocca2022group,meyer2022exploiting,skolik2022equivariant,sauvage2022building,glick2021covariant}. However, additional care should be taken if such embeddings are to be found through optimizing embedding architectures, since we show this training embedding process can also exhibit barren plateaus. 
Furthermore, we consider the projected quantum kernel which is constructed by measuring local subsystems and has been shown to maintain good generalization in a situation where the fidelity kernel fails to generalize~\cite{huang2021power,kubler2021inductive}. 

In contrast to QNNs where the trainability barrier caused by BPs is now common knowledge, the community is generally less aware of the problems posed by exponential concentration for quantum kernel methods. 
\sth{The problem of exponential concentration for the fidelity quantum kernel was first observed in Ref.~\cite{huang2021power} and later analyzed in Ref.~\cite{kubler2021inductive, shaydulin2021importance, canatar2022bandwidth} in the context of generalization. Ref.~\cite{kubler2021inductive} discusses exponential concentration in the context of a projected quantum kernel for a specific example embedding.
On the other hand, Refs.~\cite{gentinetta2022complexity, liu2021rigorous} provide a rigorous study of the number of measurement shots required to successfully train the fidelity kernel but do not address the issue of exponential concentration.} Here we provide a systematic treatment of the causes and effects of exponential concentration in the presence of shot noise. We intend our results to be viewed as a guideline to the types of kernels and embeddings to be avoided for successful training. Moreover, our results on noise-induced kernel concentration serve as a warning against using deep encoding schemes in the near-term. \sth{For a more detailed survey of how our results fit in the context of prior work see Appendix~\ref{ap:priorwork}.} 



\section{Results}
\subsection{Framework}\label{sec:framework-kernels-qml}

Our results apply generally to any method that involves quantum kernels. This includes both supervised learning tasks such as regression and classification tasks, as well as unsupervised learning tasks such as generative modeling and dimensional reduction. However, for concreteness we focus on supervised learning on classical data. Here, 
one is given repeated access to a training dataset $ \mathcal{S} := \{\vec{x}_i,\vec{y}_i\}_{i=1}^{N_s}$, where $\vec{x}_i \in \mathcal{X}$ are input vectors and $\vec{y}_i \in \mathcal{Y}$ are associated labels. The input vectors and labels are related by some unknown target function $f: \mathcal{X} \rightarrow \mathcal{Y}$. Our task is to use the dataset to train a parameterized QML model $h_{\vec{a}}$, i.e. a function $h_{\vec{a}}: \mathcal{X} \rightarrow \mathcal{Y}$ parameterized by $\vec{a}$, to approximate $f$.

The model can be trained by  introducing an empirical loss $\mathcal{L}_{\vec{a}}$ which quantifies the degree to which the model $h_{\vec{a}}$ agrees with the target function $f$ over the training data $\mathcal{S}$. The optimal parameters of the model are given by 
\begin{align}
    \vec{a}_{\rm opt} := {\rm argmin}_{\vec{a}} \mathcal{L}_{\vec{a}}(\mathcal{S}) \,, \label{eq:optimal-params}
\end{align}
and can be obtained by minimizing the empirical loss.
Once trained, the model is tested on some unseen data. The hope is that if the dataset is sufficiently large and appropriately chosen, the optimized function $h_{ \vec{a}_{\rm opt}}$ not only agrees on the training set but also accurately predicts the correct labels on unseen inputs. This is
exactly the question of \textit{generalization}~\cite{huang2021power,kubler2021inductive,liu2021rigorous,heyraud2022noisy,shaydulin2021importance,canatar2022bandwidth,wang2021towards,jerbi2021quantum, caro2020pseudo, bu2021onthestatistical, banchi2021generalization,gyurik2021structural, abbas2020power, du2021efficient, caro2021encodingdependent, chen2021expressibility, popescu2021learning, caro2021generalization, cai2022sample, caro2022outofdistribution, poland2020no, sharma2020reformulation, Volkoff2021Universal, jerbi2021parametrized,wu2023quantum, jager2023universal}: 
does successful training on the training data imply good predictive power on
unseen data?

In what follows we focus on quantum kernel methods. 
Here, each individual input data point $\vec{x}_i$ is encoded into an $n$-qubit data-encoded quantum state $\rho(\vec{x}_i)$ using a data-embedding unitary $U(\vec{x}_i)$, so that
\begin{align}
    \rho(\vec{x}_i) = U(\vec{x}_i) \rho_0 U^\dagger(\vec{x}_i) \;,
\end{align}
for some initial state $\rho_0$.
Consequently, the training input dataset can be seen as an ensemble of data-encoded quantum states. For now, we leave the choice of $U(\vec{x}_i)$ entirely arbitrary, and thus this framework includes all unitary embedding schemes.

For a given input data pair $\vec{x}$ and $\vec{x'}$ we evaluate a similarity measure  $\kappa(\vec{x}, \vec{x'})$ between two encoded quantum states on a quantum computer. Formally, this is a function $\kappa:\XC\times\XC\rightarrow\mathbb{R}$ corresponding to an inner product of data states, and is known as a quantum kernel \cite{schuld2021supervised,havlivcek2019supervised, huang2021power}. Here, we consider two common choices of quantum kernels. First, we study the fidelity quantum kernel~\cite{schuld2021supervised,havlivcek2019supervised}, which is defined as 
\begin{align}
    \kappa^{FQ}(\vec{x},\vec{x'}) & = \Tr[\rho(\vec{x})\rho(\vec{x'})] \;. \label{eq:fidelity-kernel-mt}
\end{align}
Second, we consider the projected quantum kernel~\cite{huang2021power}, given by
\begin{align}
    \kappa^{PQ}(\vec{x},\vec{x'}) = {\rm exp}\left( - \gamma \sum_{k=1}^n \| \rho_k(\vec{x}) - \rho_k(\vec{x'})\|^2_2\right) \;, \label{eq:projected-gaussian-kernel-mt}
\end{align}
where $\rho_k(\vec{x})$ is the reduced state of $\rho(\vec{x})$ on the $k$-the qubit, $\|\cdot\|_2$ is the Schatten 2-norm and $\gamma$ is a positive hyperparameter. 

The power of kernel-based learning methods stems from the fact that 
they map data from $\XC$ to a higher-dimensional feature space (in this case the $2^n$-dimensional Hilbert space) where inner products are taken and a decision boundary such as a support vector machine can be trained~\cite{mohri2018foundations}. Notably, thanks to the Representer Theorem, the optimal kernel-based model is guaranteed to be expressed as a linear combination of the kernels evaluated over the training dataset (see Chapter 5 in~\cite{mohri2018foundations}). More concretely, for a kernel-based QML model $h_{\vec{a}}$ depends on the input data through the inner product between states. We have that the optimal solution is given by
\begin{align}\label{eq:model-prediction}
    h_{\vec{a}_{\rm opt}}(\vec{x}) = \sum_{i = 1}^{N_s} a^{(i)}_{\rm opt} \kappa(\vec{x},\vec{x}_i) \;,
\end{align}
where $\vec{a}_{\rm opt} = (a^{(1)}_{\rm opt}, ... , a^{(N_s)}_{\rm opt})$. Additionally, if the loss $\mathcal{L}$ is appropriately chosen, then the loss landscape can be guaranteed to be convex. It follows that by constructing the Gram matrix $K$ whose entries are kernels over training input pairs,
\begin{align}\label{eq:gram}
    [K]_{ij} = \kappa(\vec{x}_i,\vec{x}_j) \;,
\end{align}
where $ \vec{x}_i, \vec{x}_j \in \mathcal{S}$, the optimal parameters $\vec{a}_{\rm opt}$ can be found by solving the convex optimization problem in Eq.~\eqref{eq:optimal-params}. Thus if the Gram matrix can be calculated exactly, kernel-based methods are perfectly trainable.

As an example, in kernel ridge regression, we consider a square loss function $\mathcal{L}_{\vec{a}}(\mathcal{S})= \frac{1}{2}\sum_{i = 1}^{N_s} (h_{\vec{a}}(\vec{x}_i) - y_i)^2 + \frac{\lambda}{2} \| \vec{a} \|^2_{\HC}$ with a regularization $ \lambda$ and a norm in a feature space $\| \vec{a} \|_{\HC}^2$. The optimal parameters can be analytically shown to be of the form
\begin{align}
    \vec{a}_{\rm opt} = (K - \lambda\mathbb{1})^{-1} \vec{y} \;,
\end{align}
where $\vec{y}$ is a training label vector with the $i^{\rm th}$ component $y_i$. 
Another common example is a support vector machine where we consider a binary classification problem with the corresponding labels $y \in \{-1,+1\}$. Using a hinge loss function with no regularization, i.e.
$\mathcal{L}_{\vec{a}}(\mathcal{S})= \frac{1}{N_s}\sum_{i = 1}^{N_s} {\rm max}(0, 1 - h_{\vec{a}}(\vec{x}_i)y_i)$, the optimization problem in Eq.~\eqref{eq:optimal-params} can be reformulated as
\begin{align}\label{eq:svm-dual-problem}
    \vec{a}_{\rm opt} = {\rm argmax}_{\vec{a}} \left[\sum_{i=1}^{N_s} a^{(i)} - \frac{1}{2}\sum_{i,j=1}^{N_s} a^{(i)}a^{(j)} y_iy_j K_{ij} \right] \; ,
\end{align}
subject to $ 0 \leq a^{(i)}$ for all $i$. Assuming that the Gram matrix $K$ can be accurately and efficiently obtained, solving for the optimal parameters can done with a number of iterations in $\OC(\poly(N_s))$.

\subsection{Why exponential concentration is problematic}\label{sec:concentraion-problem}

By virtue of their convex optimization landscapes, kernel methods are guaranteed to obtain the optimal model from a given Gram matrix. However, due to the probabilistic nature of quantum devices, in practice the entries of the Gram matrix can only be estimated via repeated measurements on a quantum device. Thus the model is only ever trained on a statistical estimate of the Gram matrix, $\hat{K}$, instead of the exact one, $K$. The resulting statistical uncertainty, as we will argue here, inhibits how well quantum kernel methods may perform.

The heart of the problem is that, in a wide range of circumstances, the value of quantum kernels \textit{exponentially concentrate}. That is, as the size of the problem increases, the difference between kernel values become increasingly small and so, more shots are required to distinguish between kernel entries. With a polynomial shot budget this leads to an optimized model which is insensitive to the input data and cannot generalize well. 

More generally, exponential concentration can be formally defined as follows. 
\begin{definition} [Exponential concentration]\label{def:exp-concentration}
Consider a quantity $X(\vec{\alpha})$ that depends on a set of variables $\vec{\alpha}$ and can be measured from a quantum computer as the expectation of some observable. $X(\vec{\alpha})$ is said to be deterministically  exponentially concentrated in the number of qubits $n$ towards a certain \sth{$\vec{\alpha}$-independent} value $\mu$ if
\begin{align}
    |X(\vec{\alpha}) - \mu |\leq \beta \in O(1/b^n) \;,
\end{align}
for some $b>1$ and all $\vec{\alpha}$. Analogously, $X(\vec{\alpha})$ is probabilistically exponentially concentrated if
\begin{align} \label{eq:def-prob-concentration}
    {\rm Pr}_{\vec{\alpha}}[|X(\vec{\alpha}) - \mu| \geq \delta] \leq \frac{\beta}{\delta^2} \;\; , \; \beta \in O(1/b^n) \;,
\end{align}
for $b> 1$. That is, the probability that $X(\vec{\alpha})$ deviates from $\mu$ by a small amount $\delta$ is exponentially small for all $\vec{\alpha}$.

\sth{If $\mu$ additionally exponentially vanishes in the number of qubits i.e., $\mu \in \OC(1/b'^n)$ for some $b' >1$, we say that $X(\vec{\alpha})$ exponentially concentrates towards an exponentially small value.}
\end{definition}

We remark that using Chebyshev's inequality, probabilistic exponential concentration can also be diagnosed by analysing the variance of $X(\vec{\alpha})$. That is, $X(\vec{\alpha})$ is exponentially concentrated towards its mean $\mu = \mathbb{E}_{\vec{\alpha}}[ X(\vec{\alpha})]$ if
\begin{align} \label{eq:def-var-concentration}
   \Var_{\vec{\alpha}} [X(\vec{\alpha})] \in \mathcal{O}(1/b^n) \;,
\end{align}
for $b>1$, thus satisfying Definition \ref{def:exp-concentration}. Here the variance is taken over $\vec{\alpha}$. If $0 \leq X(\vec{\alpha}) \leq 1$ for all $\vec{\alpha}$ (as for quantum kernels) one can demonstrate exponential concentration by showing that the mean $\mu = \Ebb_{\vec{\alpha}}[X(\vec{\alpha})]$ is exponentially small which directly implies that $\Var_{\vec{\alpha}}[ X(\vec{\alpha})] \in \OC(1/b^n)$. \sth{Furthermore, in this context when $\mu$ vanishes exponentially, we can say that the probability of deviating from zero by an arbitrary constant amount is exponentially small.}

Definition~\ref{def:exp-concentration} is rather general and applies to a number of QML frameworks. In the case of quantum neural networks,  $X(\vec{\alpha}) = C(\vec{\theta})$, where $\vec{\alpha}=\vec{\theta}$ and $C(\vec{\theta})$ is a cost function that depends on some variational ansatz parameters $\vec{\theta}$. 
In the context of quantum landscape theory, such concentration is central to studying the BP phenomenon. In particular, the equivalence between exponentially concentrating costs and vanishing gradients cost gradients is demonstrated in Ref.~\cite{arrasmith2021equivalence}.
In the context of quantum kernels, the quantity of interest is the quantum kernel, i.e., $X(\vec{\alpha}) = \kappa(\vec{x},\vec{x'})$ where the set of variables is a pair of input data $\vec{\alpha} = \{ \vec{x}, \vec{x'} \}$. Hence the probability in Eq.~\eqref{eq:def-prob-concentration} and the variance in Eq.~\eqref{eq:def-var-concentration} is now taken over all possible pairs of input data $\{\vec{x}, \vec{x'}\}$.

To understand the problems caused by exponential concentration, let us first consider the fidelity kernel.
In practice, the kernel value is statistically estimated from measuring $N$ samples where (on all but classically simulable quantum devices) we assume we are restricted to $N \in\OC(\poly(n))$. For a given input data pair $\vec{x}$ and $\vec{x'}$, we consider two common measurement strategies to estimate the kernel value: (i) the Loschmidt Echo test (i.e., the overlap test) and (ii) the SWAP test. In either case, the fidelity quantum kernel is equivalent to the expectation value of an observable $O$ for some quantum state $\rho$ with the exact expression for $O$ and $\rho$ depending on the strategy used. If we write the eigendecomposition of the observable as $O = \sum_i o_i |o_i\rangle\langle o_i |$ where $o_i$ and $ |o_i\rangle$ are the eigenvalues and eigenvectors of $O$ respectively, then the statistical estimate after $N$ measurements is given by
\begin{align}\label{eq:estimate-fidelity-kernel}
    \widehat{\kappa}^{\rm FQ}(\vec{x},\vec{x'}) = \frac{1}{N}\sum_{m=1}^{N} \lambda_m \; .
\end{align}
Here $\lambda_m$ is the outcome of the $m^{\rm th}$ measurement and can be treated as a random variable which takes the value $o_i$ with probability $p_i = \Tr[|o_i\rangle\langle o_i | \rho]$.

The behavior of the statistical estimate depends on the measurement strategy taken. When employing the Loschmidt Echo test, the kernel value corresponds to the probability of observing the all-zero bitstring. To estimate this probability, we assign $+1$ to the outcome of obtaining the all-zero bitstring and assign $0$ to other bitstrings. If the kernel value concentrates to \sth{an exponentially small value i.e., $\mu \in \OC(1/b^n)$,} then the chance of \textit{never} obtaining the all-zero bitstring from $N$ samples is $(1 - \mu)^N \approx 1 - N \mu$. That is, with a polynomial number of samples $N \in \OC(\poly(n))$, it is very likely that none are the all-zero bitstring and hence likely that the statistical estimate of the kernel is zero. This is formalized in the following proposition (proven in Appendix~\ref{appendix:fidelity}).

\begin{proposition}\label{prop-stat-kernel-overlap}
Consider the fidelity quantum kernel as defined in Eq.~\eqref{eq:fidelity-kernel-mt}. Assume that the kernel values $\kappa^{\rm FQ}(\vec{x},\vec{x'})$ exponentially concentrate towards \sth{an exponentially small value} as per Definition~\ref{def:exp-concentration}. 
Supposing an $N \in \OC(\poly(n))$ shot Loschmidt Echo test is used to estimate the Gram matrix for a training dataset $\SC = \{\vec{x}_i , y_i\}$ of size $N_s$
then, with a probability exponentially close to $1$, the statistical estimate of the Gram matrix $\widehat{K}$ is equal to the identity matrix. That is,
\begin{align}
    {\rm Pr}[ \widehat{K} =  \mathbb{1} ] \geq 1 - \delta' \; \; , \; \delta' \in \OC(c^{-n}) \; 
\end{align}
for some $c > 1$.
\end{proposition}

In the case of the SWAP test the measurement outcomes are either $+1$ with probability $p_+ = 1/2 + \kappa^{\rm FQ}(\vec{x},\vec{x'})/2$, or $-1$ with probability $1 - p_+$. Thus computing the kernel value amounts to determining the perturbation from the uniform distribution where the $+1$ and $-1$ outcomes occur with equal probabilities. Intuitively, when the kernel value concentrates to \sth{an exponentially small value}, the perturbation cannot be detected with a polynomial number of measurement shots. In other words, a statistical estimate using only a polynomial number of shots does not contain information about the input data pair \sth{with probability exponentially close to $1$}. This is formally stated in the following proposition which is derived by reducing the problem of \sth{distinguishing distributions (i.e., one associated with a kernel value and the uniform distribution)} to a hypothesis testing task. 
\begin{proposition}\label{prop-stat-kernel-swap}
Assume that the fidelity quantum kernel $\kappa^{\rm FQ}(\vec{x},\vec{x'})$ exponentially concentrates towards some \sth{exponentially small value} as per Definition~\ref{def:exp-concentration}. 
Suppose an $N \in \OC(\poly(n))$ shot SWAP test is used to estimate the Gram matrix for a training dataset $\SC = \{\vec{x}_i , y_i\}$ of size $N_s$. Then, with probability exponentially close to $1$ (i.e., probability at least $1-\delta'$ such that $\delta' \in \OC(c^{-n})$ for some $c>1$), 
the estimate of the Gram matrix $\hat K$ is statistically indistinguishable from the matrix $\widehat{K}^{\rm (rand)}_N$ whose diagonal elements are $1$ and off-diagonal elements are instances of 
\begin{align}\label{eq:k0-no-input-mt}
    \widehat{\kappa}^{(\rm rand)}_N = \frac{1}{N} \sum_{m = 1}^N \Tilde{\lambda}_m \;, 
\end{align}
where each $\Tilde{\lambda}_m$ takes either $+1$ or $-1$ with equal probability. We note that $\widehat{\kappa}^{(\rm rand)}_N$ does not contain any information about the input data $\SC = \{\vec{x}_i , y_i\}$.
\end{proposition}
\sth{We refer the readers to Appendix~\ref{appendix:stat-indis-basic} for an introduction to some preliminary tools for a hypothesis testing and Appendix~\ref{appendix:fidelity-swap} for further technical details regarding the SWAP test, which includes formal definitions of statistical indistinguishability (i.e., Definition~\ref{def:StatIndist} for distributions and Definition~\ref{def:StatIndistOutputs} for outputs), and a proof of the proposition.}

\begin{figure}[t]
\includegraphics[width=.85\columnwidth]{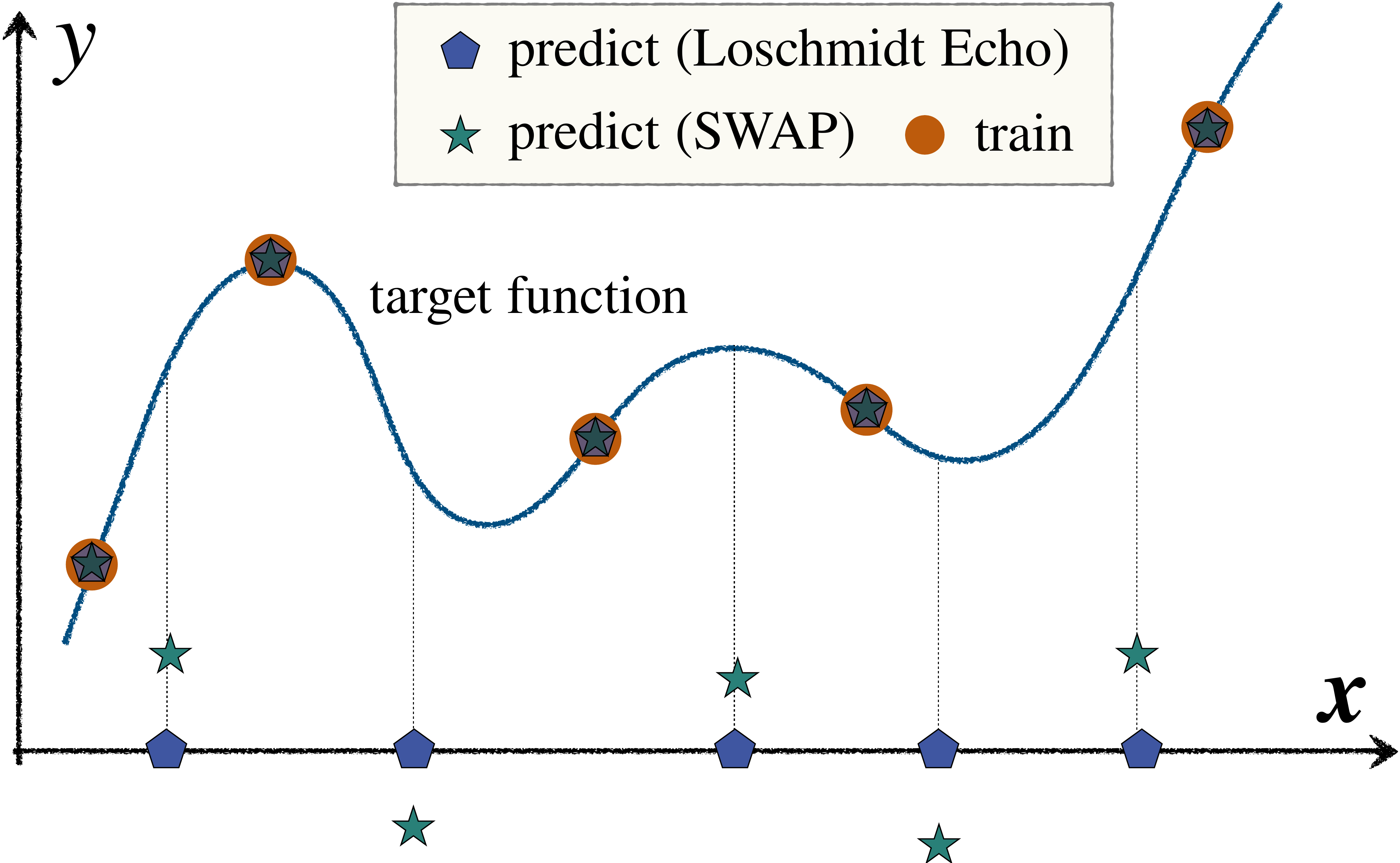}
\caption{\textbf{Schematic of effect of exponential concentration and shot noise on training and generalization performance.} For the unseen (test) data, the behavior depends on how kernel values are statistically estimated. In the case of the Loschmidt Echo test, the model predictions are zero with high probability. On using the SWAP test, the model predictions fluctuate around zero (due to shot noise). On the other hand, for the training data, the training labels are effectively hard-coded by the optimization process. (For simplicity we here consider the limit of no regularization.) %
}\label{fig:effect-exp-con}
\end{figure}

Although statistical estimates of the kernel behave differently depending on the choice of measurement strategy, they are both in effect independent of the input data for large $n$. Thus training with this estimated Gram matrix leads to a model whose predictions are independent of the input training data. 
We present numerical simulations to support our theoretical findings in Appendix~\ref{appendix:fidelity-numerics}. 

Crucially, this conclusion applies generally beyond kernel ridge regression to other kernel methods including both supervised learning tasks and unsupervised learning tasks. As a concrete example, we consider the optimal solution for kernel ridge regression in the presence of exponential concentration.

\begin{corollary}\label{coro:opt-params}
Consider a kernel ridge regression task with a squared loss function and regularization $\lambda$ using the same assumptions as Proposition~\ref{prop-stat-kernel-overlap}. Denote $\vec{y}$ as a vector with its $i^{\rm th}$ elements equal to $y_i$.
\begin{enumerate}
    \item For the Loschimdt Echo test, the optimal parameters are found to be
    \begin{align}
        \vec{a}_{0}(\vec{y}, \lambda) = \frac{\vec{y}}{1 - \lambda} \;,
    \end{align}
    with probability at least $1 - \delta$ with $\delta \in \OC(b^{-n})$ for some $b > 1$. 

    For a test data point $\vec{x} \notin \mathcal{S}$, the model prediction is $0$ with probability at least $1 - \delta'$ such that $\delta' \in \OC(b'^{-n})$ for some $b' > 1$.
    \item For the SWAP test, the optimal parameters are statistically indistinguishable from the vector
    \begin{align}
        \vec{a}_{\rm rand}(\vec{y}, \lambda) = \left( \widehat{K}^{\rm (rand)}_N - \lambda \mathbb{1} \right)^{-1} \vec{y} \;,
    \end{align}
    \sth{with probability at least $1 - \tilde{\delta}$ with $\tilde{\delta} \in \OC(\tilde{b}^{-n})$ for some $\tilde{b} > 1$. Here,} $\widehat{K}^{\rm (rand)}_N$ is a data-independent random matrix whose diagonal elements are $1$ and off-diagonal elements are instances of $\widehat{\kappa}^{(\rm rand)}_N $ in Eq.~\eqref{eq:k0-no-input-mt}.

    In addition, \sth{with probability exponentially close to $1$,} the model prediction \sth{on unseen data} is statistically indistinguishable from the data-independent random variables that result from measuring $\widehat{K}^{\rm (rand)}_N$.
\end{enumerate}
\end{corollary}

Corollary~\ref{coro:opt-params} shows that, regardless of the measurement strategy to estimate the kernel value, exponential concentration leads to a trained model where the predictions on unseen inputs are independent of the training data. A visual illustration of the effect of exponential concentration in the presence of shot noise on model predictions is provided in Fig.~\ref{fig:effect-exp-con}. We note that these propositions and corollary presented here are simplified versions and refer the readers to Appendix~\ref{appendix:fidelity} for the full statements and proofs.

It is natural to ask whether the problems caused by exponential concentration should be viewed as a barrier to training or generalization. Since the estimated Gram matrix is still positive semi-definite, the loss landscape remains convex when the model is trained and the optimal model with respect to this estimated Gram matrix is guaranteed to be obtained. Although this trained model is independent of input data (as explained above), the model can still perform well on the training phase and achieve small training errors in the limit of small regularization. This is because the training output data is trivially `cooked' into the model via the optimization process (independently of the kernel values). 

On the other hand, the data independence of the kernel values means that the predictions of the trained model are completely independent of the training data and so the trained model in general performs trivially on unseen data. That is, the model generalizes terribly. By incorporating the effect of shot noise, this has a different flavor to typical barriers to generalization in that crucially it arises from using not enough shots (rather than not enough training data). Moreover, in our numerics below we concretely see that this barrier cannot be resolved by training on more input data points.

\begin{figure}[t]
\includegraphics[width=.99\columnwidth]{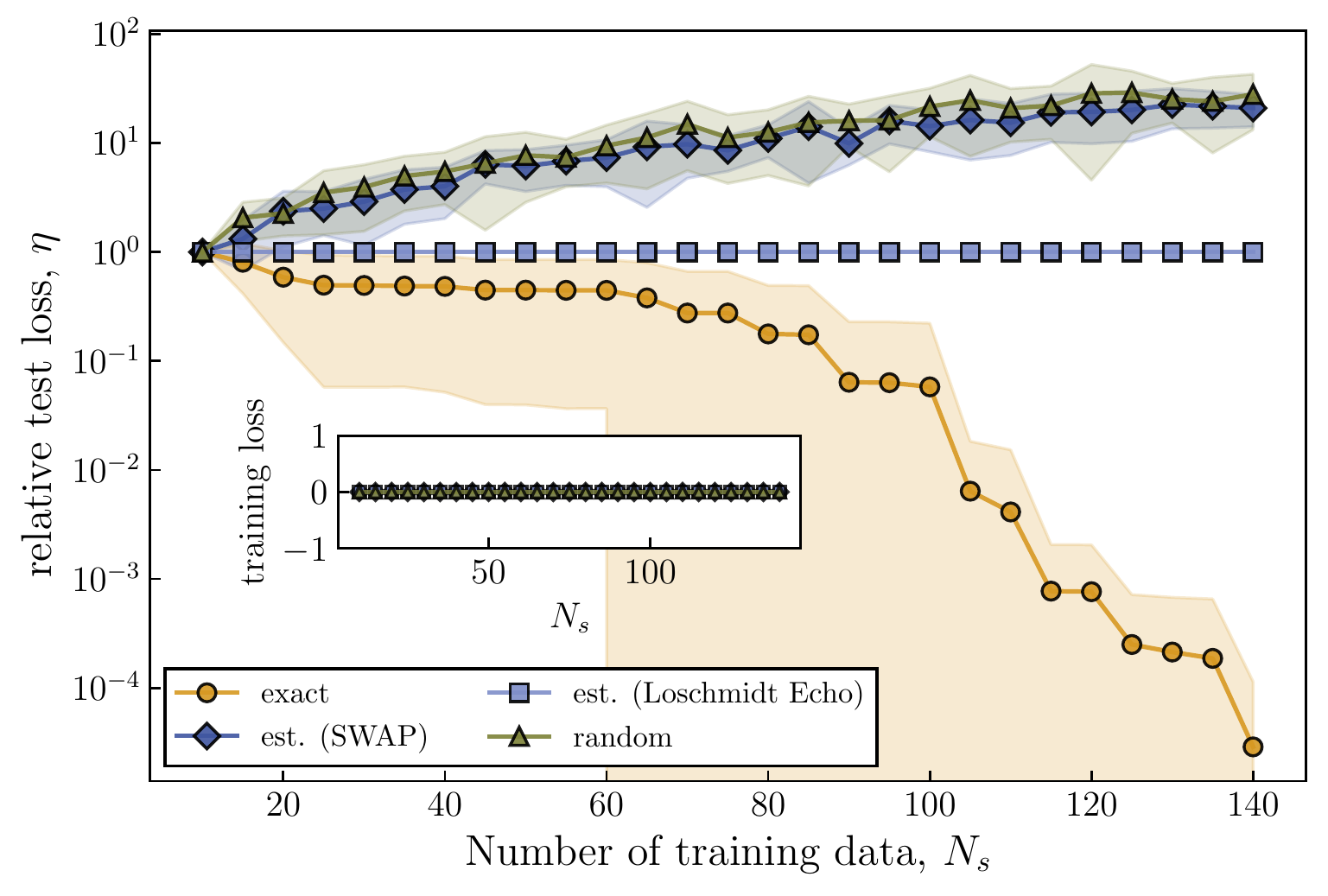}
\caption{\textbf{Effect of exponential concentration on training and generalization performance.} 
We consider a tensor product encoding for an engineered data set where each component is uniformly drawn from $[0, 2\pi]$ and the true label is $y_{\rm true} (\vec{x}) = \sum_{i=1}^{N_s} w_{i} \kappa^{\rm FQ}(\vec{x_i}, \vec{x})$ where $w_i$ is randomly chosen from $[0,1]$. We train on $N_s = 150$ data points. In the main plot, the loss on a test dataset $\SC_{\rm test}$ relative to its initial value (without training) is plotted as a function of increasing training data. In the inset, an \textit{absolute} training error is plotted as a function of the increasing data. We note that each kernel value is estimated with $N=1000$ and the number testing data points is $20$. The training is done with no regularization $\lambda = 0$. We repeat this experiment $10$ times. The solid curves represent averages of respective losses and the shaded areas represent standard deviations. 
}\label{fig:effect-exp-con-gen}
\end{figure}

Fig.~\ref{fig:effect-exp-con-gen} numerically demonstrates this effect on an engineered dataset for a $40$ qubit simulation.
In the main plot, the generalization is studied as a function of increasing training data $\SC_{N_s}$ and whether the training is performed with exact or estimated kernel values. Particularly, to observe the improvement due to increasing data, we plot a \textit{relative} loss on a test dataset $\SC_{\rm test}$ with respect to its initial value ($N_S = 10$) i.e., $\eta(N_s) = \frac{\LC_{\vec{a}}(\SC_{\rm test} | \SC_{N_s})}{\LC_{\vec{a}}(\SC_{\rm test} | \SC_{N_s=10})}$. That is, $\eta(N_s) < 1$ for $N_s > 10$ indicates better generalization with increasing training data. This is observed to be the case for the training on the exact kernel value where the model gradually generalizes better. In fact, this learning task is synthesized such that when training on the whole dataset, with an access to the exact kernel values, the trained model generalizes perfectly. Even with this dataset which is heavily favourable for the fidelity kernel, the performance on unseen data with the estimated kernel values shows no improvement with the increasing training data. Specifically, when the Loschmidt Echo test is used to evaluate kernel values, the statistical estimates accumulate at exactly zero, leading to $\eta(N_s) = 1$ for all $N_s$. In addition, for the SWAP measurement strategy, there is no improvement with increasing data and the behavior of $\eta(N_s)$ aligns with the one where the model is trained on a random matrix where each off-diagonal element is a data-independent random variable in Eq.~\eqref{eq:k0-no-input-mt}. On the other hand, as demonstrated in the inset of Fig~\ref{fig:effect-exp-con-gen}, the trained model performs perfectly on the training set $\SC_{N_s}$ and achieves zero training errors in all cases. This is again because the training label information is hard-coded in the optimization process. These empirical observations are all in good agreement with our theoretical predictions.

Finally, the analysis in the case of the projected quantum kernel is slightly more complicated as estimating the kernel requires us to first obtain the statistical estimates of the 2-norms between the reduced data encoding states on all individual qubits from quantum computers. Two common strategies to to do so include (i) the full tomography of the single qubit reduced density matrices and (ii) the local SWAP tests. In Appendix~\ref{appendix:projected}, we again use a hypothesis testing framework to analyze the effect of exponential concentration on the projected kernel for these strategies. Similarly to the fidelity kernel we find that the final trained model is in effect independent of the training data.

\subsection{Sources of exponential concentration}\label{sec:thm-source}

Given that exponential concentration leads to trivial data-independent models, it is important to determine when kernel values will, or will not, concentrate. In this section, we investigate the causes of exponential concentration for quantum kernels.

In broad terms, the exponential concentration of quantum kernels may be viewed as stemming from the fact in certain situations it can be difficult to extract information from quantum states. In particular, we identify four key features that can severely hinder the information extraction process via kernels. These include: i. \hyperref[sec:expressivity]{the expressivity of the data embedding}, ii.  \hyperref[sec:entanglement]{entanglement}, iii. \hyperref[sec:global]{global measurements} and iv. \hyperref[sec:noise]{noise} (see Fig.~\ref{fig:summary}). For each source, we derive an associated concentration bound. 
As summarised in Table~\ref{table:KernelsVsQML} each of these theorems has an analogue for QNN. All the proofs of our main results are presented in the Appendix.

\subsubsection{Expressivity-induced concentration}\label{sec:expressivity}

In broad terms, the expressivity of an ensemble of unitaries is defined as how close the ensemble uniformly covers the unitary group. To introduce the concept of the expressivity of the data embedding $U(\vec{x})$, we first consider the unitary ensemble generated via the data embedding $U(\vec{x})$ over all possible input data vectors $\vec{x}\in\mathcal{X}$. That is, the data embedding defines a map $U: \mathcal{X} \rightarrow \mathbb{U}_{\vec{x}} \subset \mathcal{U}(d)$, where $\mathcal{U}(d)$ is the total space of unitaries of dimension $d = 2^n$, and
\begin{align}\label{eq:unitaryensemble}
    \mathbb{U}_{\vec{x}} = \{ U(\vec{x})| \vec{x} \in \mathcal{X}\}\;.
\end{align}
In addition, for some initial state $\rho_0$, we can define an ensemble of the data-embedded quantum states $\mathbb{S}_{\vec{x}} = \{ U(\vec{x})\rho_0 U^{\dagger}(\vec{x}) \}$ for all $U(\vec{x}) \in \mathbb{U}_{\vec{x}}$.
Consequently, performing an average over all the input data is equivalent to an average over the ensemble of data-encoded unitaries $\mathbb{U}_{\vec{x}}$, or the data encoded states $\mathbb{S}_{\vec{x}}$.

More concretely, we can measure the expressivity of a given ensemble $\Ubb$ by how close it is from a 2-design (a pseudo-random distribution that agrees with the random distribution up to the second moment). Adopting the definition of expressivity used in Ref.~\cite{sim2019expressibility, nakaji2020expressibility, holmes2021connecting}, the following superoperator formally quantifies
the distance between $\Ubb$ and an ensemble that forms a 2-design,
\begin{align}
    \mathcal{A}_{\Ubb}(\cdot) := \mathcal{V}_{{\rm Haar}}(\cdot) - \int_{\Ubb} dU U^{\otimes 2} (\cdot)^{\otimes 2}(U^\dagger)^{\otimes 2} \; . \label{eq:expressivity-measure-mt}
\end{align}
Here $\mathcal{V}_{{\rm Haar}}(\cdot) = \int_{\mathcal{U}(d)} d\mu (V) V^{\otimes 2}(\cdot)^{\otimes 2} (V^\dagger)^{\otimes 2}$ is an integral over Haar ensemble and the second term is an integral over $\Ubb$. 
In our case, we have the data-encoded ensemble as our ensemble of interest i.e., $\Ubb = \Ubb_{\vec{x}}$.
Given an input state $\rho_0$, the trace norm
\begin{align}
    \varepsilon_{\mathbb{U}_{\vec{x}}} := \|\mathcal{A}_{\mathbb{U}_{\vec{x}}}(\rho_0) \|_1 \, ,\label{eq:expressivity-measure-epsilon}
\end{align}
can be chosen as a data-dependent expressivity measure. 
The data-dependence of $\varepsilon_{\mathbb{U}_{\vec{x}}}$ stems from the dependence of $\mathbb{U}_{\vec{x}}$, Eq.~\eqref{eq:unitaryensemble}, on the input data $\mathcal{X}$. Thus $\varepsilon_{\mathbb{U}_{\vec{x}}}$ takes into account not just the expressivity of the embedding but also the randomness of the input dataset. The measure equals zero, $\varepsilon_{\mathbb{U}_{\vec{x}}} = 0$, only when $\mathbb{U}_{\vec{x}}$ is maximally expressive (i.e., when it agrees with the uniform distribution up to at least the second moment).

To understand why expressivity can be an issue for kernel-based methods, let us consider the fidelity quantum kernel of Eq.~\eqref{eq:fidelity-kernel-mt}. This kernel requires computing the inner product between two vectors in an exponentially large Hilbert space. As such, for highly expressive embeddings we are essentially evaluating the inner product between two approximately random (and hence orthogonal) vectors, thus leading to typical kernel values being  exponentially small. That is, kernel values tend to concentrate with increased expressivity. 
The following theorem establishes the formal relationship between the expressivity of the unitary embedding and the concentration of quantum kernels.

\begin{theorem}[Expressivity-induced concentration]\label{thm:expressivity-kernel}
Consider the fidelity quantum kernel as defined in Eq.~\eqref{eq:fidelity-kernel-mt} and the projected quantum kernel as defined in Eq.~\eqref{eq:projected-gaussian-kernel-mt}. Assume that input data $\vec{x}$ and $\vec{x'}$ are drawn from the same distribution, leading to an ensemble of unitaries $\mathbb{U}_{\vec{x}}$ as defined in Eq.~\eqref{eq:unitaryensemble}. We have
\begin{align}
    {\rm Pr}_{\vec{x},\vec{x'}}[|\kappa(\vec{x},\vec{x'}) -  \mathbb{E}_{\vec{x},\vec{x'}} [\kappa(\vec{x},\vec{x'})]| \geq \delta] \leq \frac{G_n(\varepsilon_{\mathbb{U}_{\vec{x}}})}{\delta^2} \;, 
\end{align}
where $ \varepsilon_{\mathbb{U}_{\vec{x}}} = \|\mathcal{A}_{\mathbb{U}_{\vec{x}}}(\rho_0) \|_1$ is the data-dependent expressivity measure over $\mathbb{U}_{\vec{x}}$ defined in Eq.~\eqref{eq:expressivity-measure-epsilon}, and  $G_n(\varepsilon_{\mathbb{U}_{\vec{x}}})$ is a function of $\varepsilon_{\mathbb{U}_{\vec{x}}}$ defined as below.
\begin{enumerate}
    \item For the fidelity quantum kernel $\kappa(\vec{x},\vec{x'}) = \kappa^{FQ}(\vec{x},\vec{x'})$, we have
\begin{align}
    {G_n(\varepsilon_{\mathbb{U}_{\vec{x}}})} = \beta_{\rm Haar} +  \varepsilon_{\mathbb{U}_{\vec{x}}} ( \varepsilon_{\mathbb{U}_{\vec{x}}} + 2\sqrt{\beta_{\rm Haar}}) \; ,
\end{align}
where $\beta_{\rm Haar} = \frac{1}{2^{n-1}(2^n+1)}$
    \item For the projected quantum kernel $\kappa(\vec{x},\vec{x'}) = \kappa^{PQ}(\vec{x},\vec{x'})$, we have 
\begin{align}
    {G_n(\varepsilon_{\mathbb{U}_{\vec{x}}})} = 4 \gamma n ( \tilde{\beta}_{\rm Haar} + \varepsilon_{\mathbb{U}_{\vec{x}}}) \;,
\end{align}
where $\tilde{\beta}_{\rm Haar} = \frac{3}{2^{n+1}+2}$. 
\end{enumerate}
\end{theorem}

Theorem~\ref{thm:expressivity-kernel} establishes that higher embedding expressivity leads to greater quantum kernel concentration. That is, the upper bound on the kernel concentration becomes smaller when $U(\vec{x})$ is more expressive. In the limit where $\mathbb{U}_{\vec{x}}$ forms an ensemble that is exponentially close to a 2-design (corresponding to $\varepsilon_{\mathbb{U}_{\vec{x}}} \in \mathcal{O}(1/b^n)$ for $b>1$), the kernel exponentially concentrates, and so exponentially many measurement shots are required to evaluate the kernel on a quantum device. \sth{Note that the fidelity kernel exponentially concentrates to some exponentially small value i.e., $\mu = \Ebb_{\vec{x},\vec{x'}\sim {\rm Haar}}[ \kappa^{\rm FQ}(\vec{x},\vec{x'})] = 1/2^n$.}

We stress that the proof of Theorem~\ref{thm:expressivity-kernel} makes no assumptions on the form of $U(\vec{x})$. This means the theorem holds for a wide range of embedding architectures, including both problem-agnostic~\cite{peters2021machine,hubregtsen2021training,wang2021towards,thanasilp2021subtleties,schuld2021supervised,havlivcek2019supervised} and problem-inspired embeddings~\cite{huang2021power,kubler2021inductive,liu2021rigorous}. 
In Appendix~\ref{appendix:extension-expressivity}, we generalize Theorem~\ref{thm:expressivity-kernel} by relaxing the assumption that $\vec{x}$ and $\vec{x'}$ are drawn from the same distribution. This is relevant, for example, in binary classification tasks where one might want to analyze the behavior of the kernel when a pair of inputs are drawn from different training ensembles. Here we find a similar conclusion, where higher expressivity lead to more concentrated kernel values. 

Although Theorem~\ref{thm:expressivity-kernel} is stated in terms of the unitary embedding of classical data, the theorem is also applicable to quantum data. That is, it can also be applied when the input data is a collection of pure quantum states generated directly from some quantum process of interest. This follows from the fact that given a set of quantum data states, there is a (potentially unknown) underlying ensemble of unitaries associated with preparing this set of states. Since we can associate each of these unitaries with a classical label, we can associate the quantum data with an encoding ensemble $\mathbb{U}_{\vec{x}}$ and apply Theorem~\ref{thm:expressivity-kernel} as before. 
For example, consider a Hamiltonian $H(\vec{x})$ where $\vec{x}$ are parameters describing the Hamiltonian (e.g. perhaps on-site energies or interaction strengths). The quantum data generated by an evolution under $H(\vec{x})$ for time $T$ can be expressed as $\{U(\vec{x_i})\ket{0} = e^{-iH(\vec{x_i})T}\ket{0} \}_i$.

\begin{figure}[t]
\includegraphics[width=.5\columnwidth]{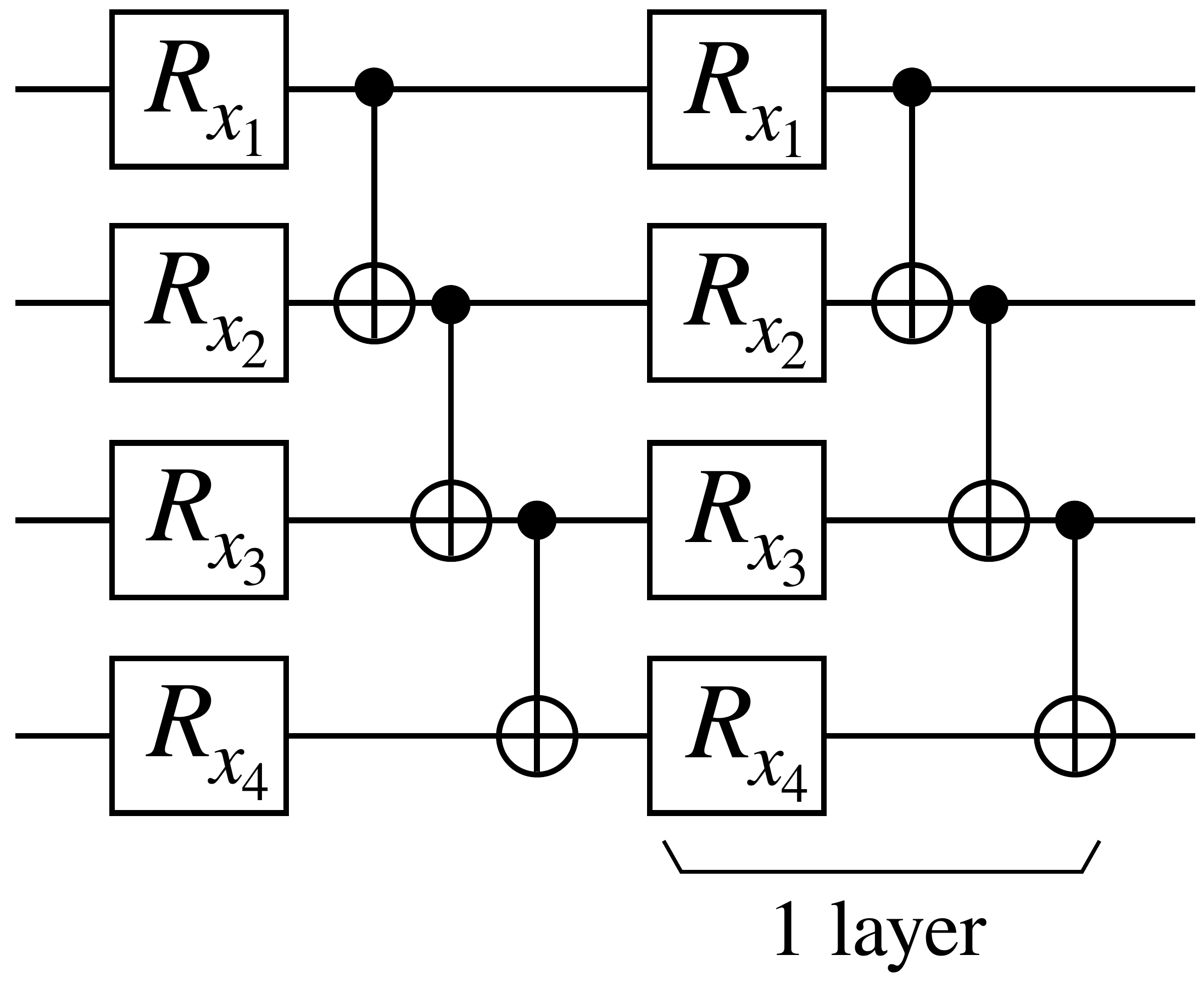}
\caption{\textbf{Hardware Efficient Embedding (HEE).} A layer is composed of single qubit x-rotations where the rotation angle on qubit $k$ is given by the $k_{\rm th}$ component of the input data point $\vec{x}$. After each layer of rotations, one applies entangling gates acting on adjacent pairs of qubits.}\label{fig:HEE}
\end{figure}

\medskip 

We numerically probe the dependence of the concentration of quantum kernel values on the expressivity of the data embedding.
To do so, we consider a Hardware Efficient Embedding (HEE)~\cite{thanasilp2021subtleties}, comprised of $L$ layers of data-dependent single-qubit rotations around the $x$-axis followed by entangling gates (see Fig.~\ref{fig:HEE}). We further consider a data re-uploading strategy where an input data point is repeatedly uploaded into the data embedding~\cite{schuld2021supervised,perez2020data,thanasilp2021subtleties,gan2022fock}. In particular, the $i^{\rm th}$-component of a data point $\vec{x}$ is encoded as the rotation angle of qubit $i$ in every HEE layer.

We choose to focus on the binary classification task of distinguishing handwritten `0' and `1' digits from the MNIST dataset~\cite{lecun1998mnist}. As sketched in Fig.~\ref{fig:dataset}(a), each individual image (i.e., an input data point) is dimensionally reduced to a real-valued vector of length $n$ using principle component analysis. We refer the reader to Appendix~F of Ref.~\cite{thanasilp2021subtleties} for more details.

\begin{figure}[t]
\includegraphics[width=.9\columnwidth]{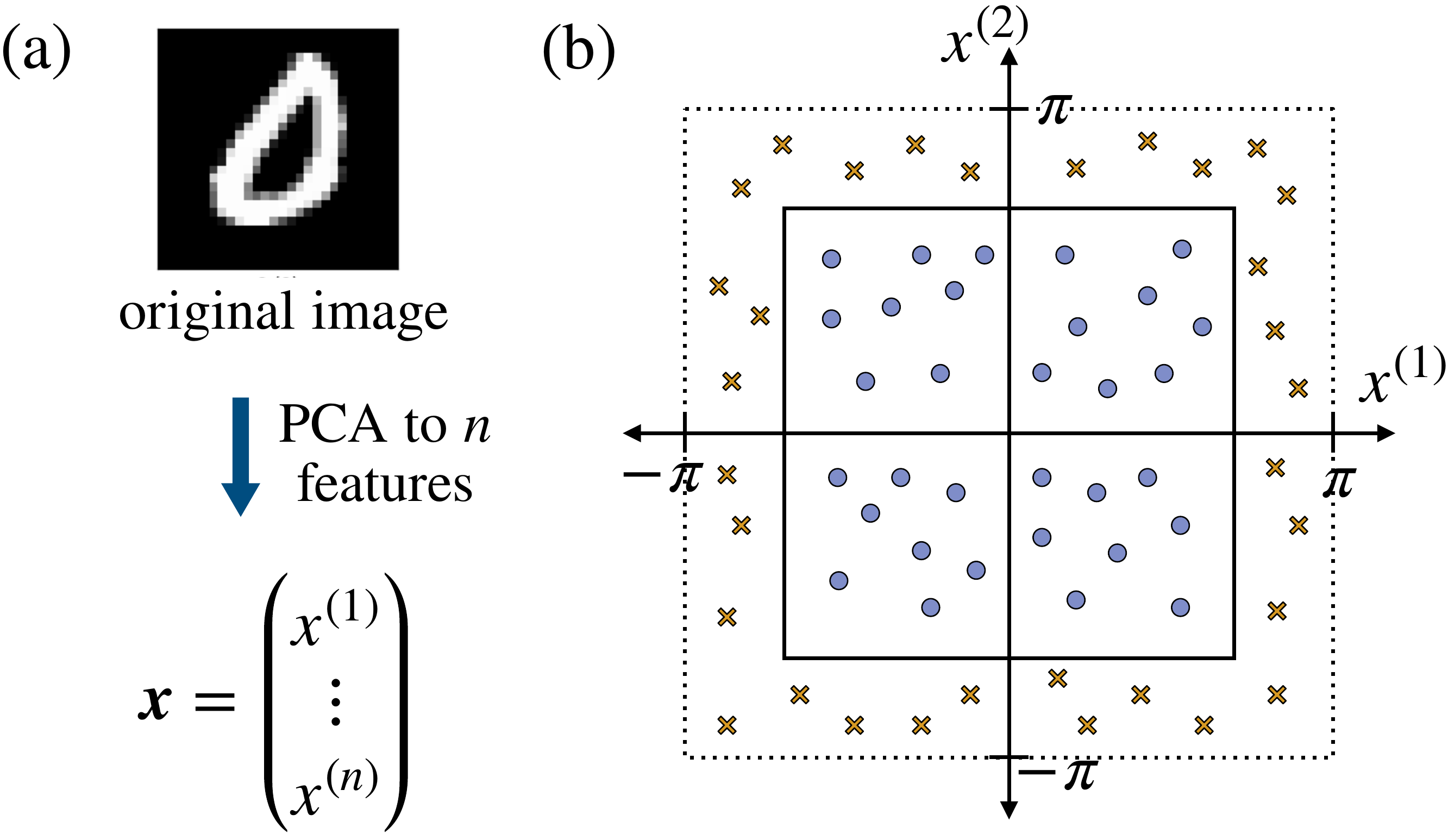}
\caption{\textbf{Datasets.} (a) An input data point $\vec{x}$ is obtained from dimensionally reducing an original MNIST image to $n$ features using principal component analysis. We assign label $-1$ if the original image is digit `0' and $1$ if the original image is digit `1'.
(b) A hypercube of width $2\pi/2^{1/n}$ is centred at the origin. An input data point $\vec{x}$ with each of its component bounded between $-\pi$ and $\pi$ has an associated label $y=1$ if the point is inside the hypercube (represented by a circle) and $y=-1$, otherwise (represented by a cross).}\label{fig:dataset}
\end{figure}

For a dataset $\mathcal{S} = \{ \vec{x}_i, y_i \}_{i=1}^{N_s}$ of size $N_s$, we evaluate the kernel values over all possible different pairs of inputs in $\mathcal{S}$. Thus we consider the set of values: $\mathcal{K}_{\mathcal{S}} = \{ \kappa(\vec{x}_1,\vec{x}_2), \kappa(\vec{x}_1,\vec{x}_3), ... , \kappa(\vec{x}_{N_s - 1},\vec{x}_{N_s})  \}$. We note that kernel values for pairs of identical inputs are always $1$ and are so excluded from $\mathcal{K}_{\mathcal{S}}$. To study the degree to which the quantum kernels probabilistically concentrate, we compute the variance $\Var_{\vec{x},\vec{x}'}[\kappa(\vec{x},\vec{x}')]$ over $\mathcal{K}_{\mathcal{S}}$.

Fig.~\ref{fig:expressivity} shows results for the scaling of the kernel variance as a function of the number of qubits $n$ and HEE layers $L$. As $L$ increases, the expressivity of the ansatz increases, and for sufficiently large $L$ we observe exponential concentration of both the fidelity and projected quantum kernels. We note that while the projected quantum kernel reaches the exponential decay regime at shorter depths (i.e., roughly $L \geq 16$ for the projected kernel, compared to $L \geq 75$ for the fidelity kernel), we generally observe smaller variances (and so stronger concentration) for the fidelity kernel than the projected kernel.

\begin{figure}[t]
\includegraphics[width=.99\columnwidth]{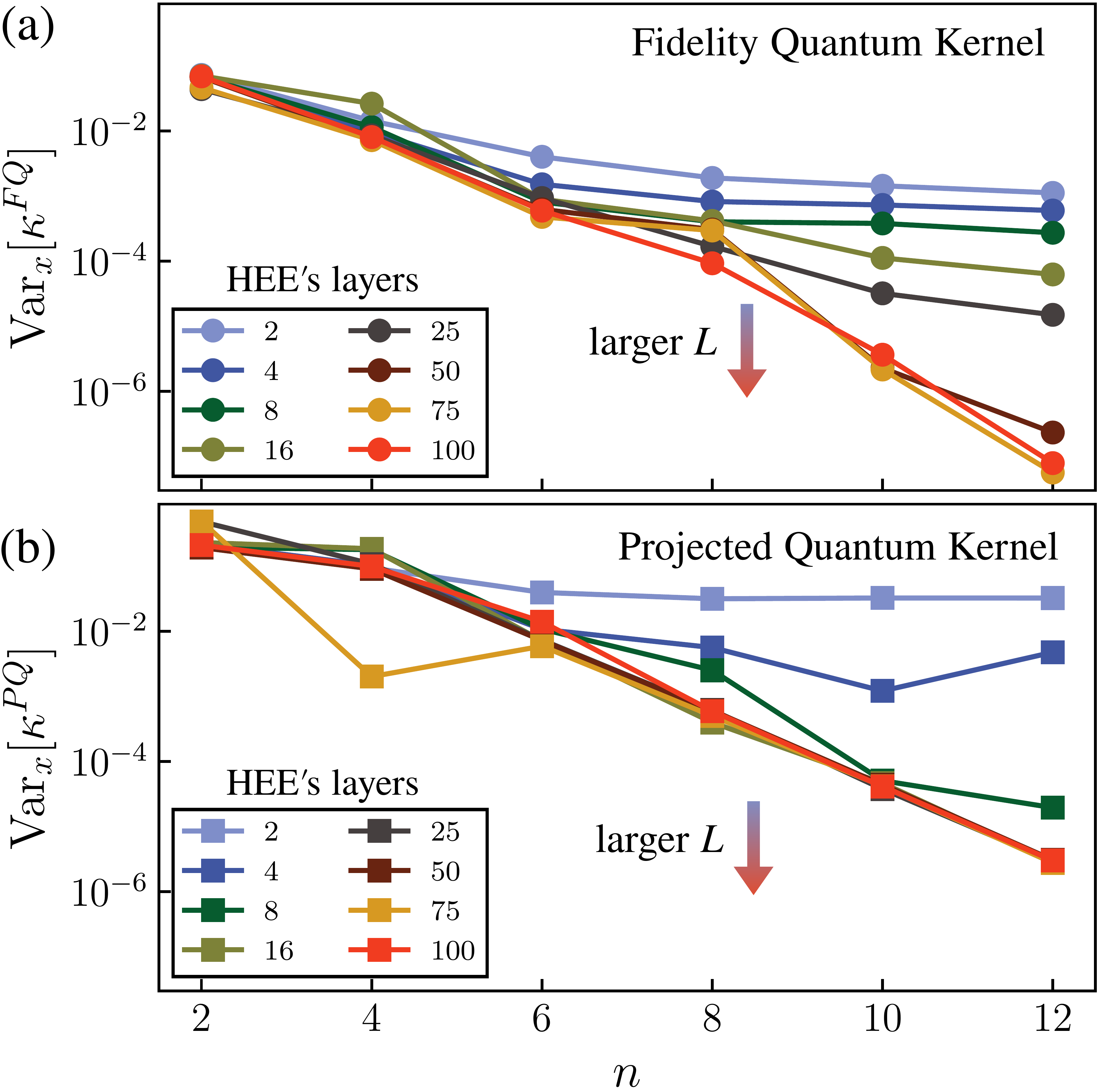}
\caption{\textbf{Effect of expressivity on quantum kernels.} We plot variances of the (a) fidelity  and (b) projected quantum kernels, as a function of $n$ and $L$. The classical data from the MNIST dataset ($N_s = 40$) is encoded via an $L$-layer HEE.}\label{fig:expressivity}
\end{figure}

\medskip

Theorem~\ref{thm:expressivity-kernel} and the numerics presented in Fig.~\ref{fig:expressivity} highlight the importance of the expressivity of quantum kernels. Namely, highly expressive encodings (whether using fidelity or projected kernels) should be avoided.
Or, more concretely, unstructured data-embeddings~\cite{peters2021machine,hubregtsen2021training,wang2021towards,schuld2021supervised} should generally be avoided and the data structure should be taken into account when designing a data-embedding (for instance by constructing geometrically inspired embedding schemes~\cite{larocca2022group,meyer2022exploiting,skolik2022equivariant,sauvage2022building}). 

\subsubsection{Entanglement-induced concentration}\label{sec:entanglement}

In the previous section we saw that high expressivity can be an issue due to the fact that kernels (such as the fidelity kernel) compare inner products of objects in exponentially large spaces. This issue can be mitigated using projected kernels, which reduce the dimension of the feature space. However, a different issue arises here due to the non-local correlations between the qubits. Namely, the entanglement of the encoded state is another potential source of concentration. Intuitively, this follows from the fact that tracing out qubits in very entangled encoded states, leads to local states that are close to maximally mixed. 
\begin{theorem}[Entanglement-induced concentration]\label{thm:entanglement-kernel}
Consider the projected quantum kernel as defined in Eq.~\eqref{eq:projected-gaussian-kernel-mt}. For a given pair of data-encoded states associated with $\vec{x}$ and $\vec{x'}$, we have
\begin{align}
    \left| 1 - \kappa^{PQ}(\vec{x},\vec{x'})\right|   \leq (2\ln2) \gamma \Gamma_s(\vec{x},\vec{x'}) \; ,
\end{align}
where
\begin{align}
    \Gamma_s(\vec{x},\vec{x'}) =  \sum_{k=1}^n \left[ \sqrt{S\left(\rho_k(\vec{x})\Big\|  \frac{\mathbb{1}_k}{2}\right)} +  \sqrt{S\left(\rho_k(\vec{x'})\Big\|  \frac{\mathbb{1}_k}{2}\right)} \right]^2  \;,
\end{align}
where we denote $S\left(\cdot\|  \cdot \right)$ as the quantum relative entropy, $\rho_k$ as a reduced state on qubit $k$, and $\mathbb{1}_k$ as the maximally mixed state on qubit $k$.
\end{theorem}

Theorem~\ref{thm:entanglement-kernel} upper bounds the deviation of kernel values from a fixed value of 1 with the relative entropy between the reduced states of the encoded data and a maximally mixed state of a single qubit. In addition, unlike the results in the previous sections, the exponential concentration bounds here are deterministic. In the case where the entanglement of the encoded states obeys a volume law, that is $S\left(\rho_k(\vec{x})\| \frac{\mathbb{1}_k}{2}\right), S\left(\rho_k(\vec{x'})\| \frac{\mathbb{1}_k}{2}\right) \in \mathcal{O}(1/2^{n-1})$ for all subsystems, the kernel values deterministically exponentially concentrate to 1. 
For encoded states that obey an area-law scaling, i.e. $S\left(\rho_k(\vec{x})\| \frac{\mathbb{1}_k}{2}\right), S\left(\rho_k(\vec{x'})\| \frac{\mathbb{1}_k}{2}\right) \in \mathcal{O}(1)$ for all subsystems, the story is more complex. Theorem~\ref{thm:entanglement-kernel}, as an upper bound, allows (but does not guarantee) that such data states do not concentrate exponentially. 

It is worth highlighting that the entanglement-induced bound in Theorem~\ref{thm:entanglement-kernel} is stated for a given pair of data-encoded states, and not as an average over all possible data pairs. Hence, it is thus natural to determine classes of data and embeddings where concentration will arise  with high probability, e.g., cases when the encoded states obey a volume law of entanglement. First, we note that if the ensemble of encoded data states forms at least a 4-design, then most of the encoded states to obey a volume-law scaling~\cite{low2009large,cotler2022fluctuations}. However, in this case, our bound on expressivity already implies that the kernel's exponentially concentrate so the entanglement-induced result is redundant.

Entanglement-induced concentration can also occur in cases where the embedding is not highly expressive but still leads to states satisfying a volume-law. Here, Theorem~\ref{thm:entanglement-kernel} implies that the kernel values of the projected quantum kernels will exponentially concentrate. 
In this case performing any supervised learning task with the projected quantum kernels will fail with a polynomial number of measurement shots.
As an example, consider binary classification and assume that
one manages to construct a $U(\vec{x})$ that maps the input data into one of the two orthogonal sets of volume-law entangled states depending on the true label of the input. \sth{In this setting, the trained model with the fidelity kernel should not face issues associated with exponential concentration.} However, if we use the projected quantum kernel, we cannot perform the task better than random guessing without spending an exponential number of shots. 
This statement is formalized in the following corollary.

\begin{corollary}\label{coro:entanglement-induced}
Consider the projected quantum kernel as defined in Eq.~\eqref{eq:projected-gaussian-kernel-mt}. If  all the states in the ensemble $\mathbb{S}_{\rm train} $ generated from the training dataset obey volume law scaling, we have
\begin{align}
    \left| 1 - \kappa^{PQ}(\vec{x},\vec{x'})\right|   \in \mathcal{O}(n 2^{-n }) \; ,
\end{align}
for all $\vec{x}$ and $\vec{x'}$ in the training data.
\end{corollary}

Thus, when using projected kernels, highly entangling encodings should be avoided to ensure \sth{predictability on unseen data}. We note that fidelity kernels (with pure input states) are not affected by entanglement in this manner as they do not require tracing qubit out. Lastly, we stress that Theorem~\ref{thm:entanglement-kernel} and Corollary~\ref{coro:entanglement-induced} are readily applied to quantum data. Indeed, entanglement-induced concentration may well be more problematic in this case since if the quantum data is already highly entangled then there is little that one can do. (In contrast, for classical data one may simply avoid highly entangled embeddings.) 
As an example, consider a quantum dataset generated by evolving different initial states with either $U_1$ or $U_2$ where $U_1$ and $U_2$ are unitaries drawn from the Haar measure over the unitary group. Since Haar random evolution leads to a volume-law scaling, classifying whether a given state is evolved by $U_1$ or $U_2$ cannot be done efficiently using projected quantum kernels.

\subsubsection{Global-measurement-induced concentration}\label{sec:global}
Global measurements can be another source of exponential concentration. A global measurement is a measurement that acts non-trivially on all $n$ qubits. Such global measurements are required by design to compute fidelity kernels but not projected kernels. In broad terms global measurements can lead to concentration because we are attempting to extract global information about a state that lives in an exponentially large Hilbert space. While projected quantum kernels do not face these difficulties due to their local construction, we argue that global measurements can lead to problems for the fidelity kernel.

To illustrate this problem, we provide an example where the data embedding has low expressivity and contains no entanglement and yet it is still possible to have exponential concentration. Consider the tensor product unitary data embedding $U(\vec{x}) = \bigotimes_{k=1}^n U_k(x_k)$ with $x_k$ being a $k$-th component  of $\vec{x}$, and $U_k$ being a single-qubit rotation about the $y$-axis on the $k$-th qubit. The following proposition holds.

\begin{proposition}[Global-measurement-induced concentration ]\label{prop:global-measurement}
Consider the fidelity quantum kernel as defined in Eq.~\eqref{eq:fidelity-kernel-mt} where the data embedding is of the form $U(\vec{x}) = \bigotimes_{k=1}^n U_k(x_k)$ with $x_k$ being an input component encoded in the qubit $k$, and $U_k$ being a single-qubit rotation about the $y$-axis on the $k$-th qubit. For an input data point $\vec{x}$, assume that all components of $\vec{x}$ are independent and uniformly sampled in $[-\pi,\pi]$. Given a product initial state  $\rho_0 = \bigotimes_{k=1}^n \dya{0}$, we have,
\begin{align}
    {\rm Pr}_{\vec{x},\vec{x'}}[|\kappa^{FQ}(\vec{x},\vec{x'}) - \sth{1/2^n}  | \geq \delta] \leq \left(\frac{3}{8}\right)^n \cdot \frac{1}{\delta^2}\;.\label{eq:prop1} 
\end{align}
\end{proposition}

Intuitively, the result in Proposition~\ref{prop:global-measurement} can be understood as following from the fact that the fidelity between two product states is usually exponentially small. In the Appendix we further generalize this proposition to the case when $U_k$ is a general unitary, which also leads to a concentration result.

We remark that the assumptions underlying Proposition~\ref{prop:global-measurement} can be relevant in practice. For example, consider classifying whether or not a given point $\vec{x}$ in $n-$dimensional space (with each component bounded between $[-\pi,\pi]$) stays inside a hypercube centred at the origin with the width of $\frac{2\pi}{2^{1/n}}$ (see Fig.~\ref{fig:dataset}(b))~\footnote{The width of the hypercube is chosen so there is a $0.5$ probability of a randomly chosen point being in or out of the hypercube.}. 
For this task, an individual data point in the training dataset is generated by uniformly drawing each vector component from the range $[-\pi, \pi]$. Since here data points are obtained via uniformly sampling each component independently, the above assumptions are satisfied.

\medskip

We numerically study the concentration of the kernels for this classification task in Fig.~\ref{fig:global}. To reduce the effects of expressivity and entanglement, we first select the data embedding to be a single layer of one qubit rotations ($R_x$, $R_y$, Hadamard followed by $R_z$). Similar to the HEE, each component of an individual data point is embedded as a rotation angle. 
We observe an exponential decay in the variance of the fidelity kernel in a good agreement with our theoretical predictions. 

While Proposition 1 is derived with a tensor product embedding, similar results are expected when dealing with more general unstructured embeddings such as hardware efficient embeddings. This is because the additional complexity from using an unstructured embedding can only increase the kernel concentration (due to increased expressivity and entanglement). This is highlighted in Fig.~\ref{fig:global} where we additionally consider an $L$-layered HEE and see that increasing expressivity can accelerate the exponential decay.

\begin{figure}[t]
\includegraphics[width=.99\columnwidth]{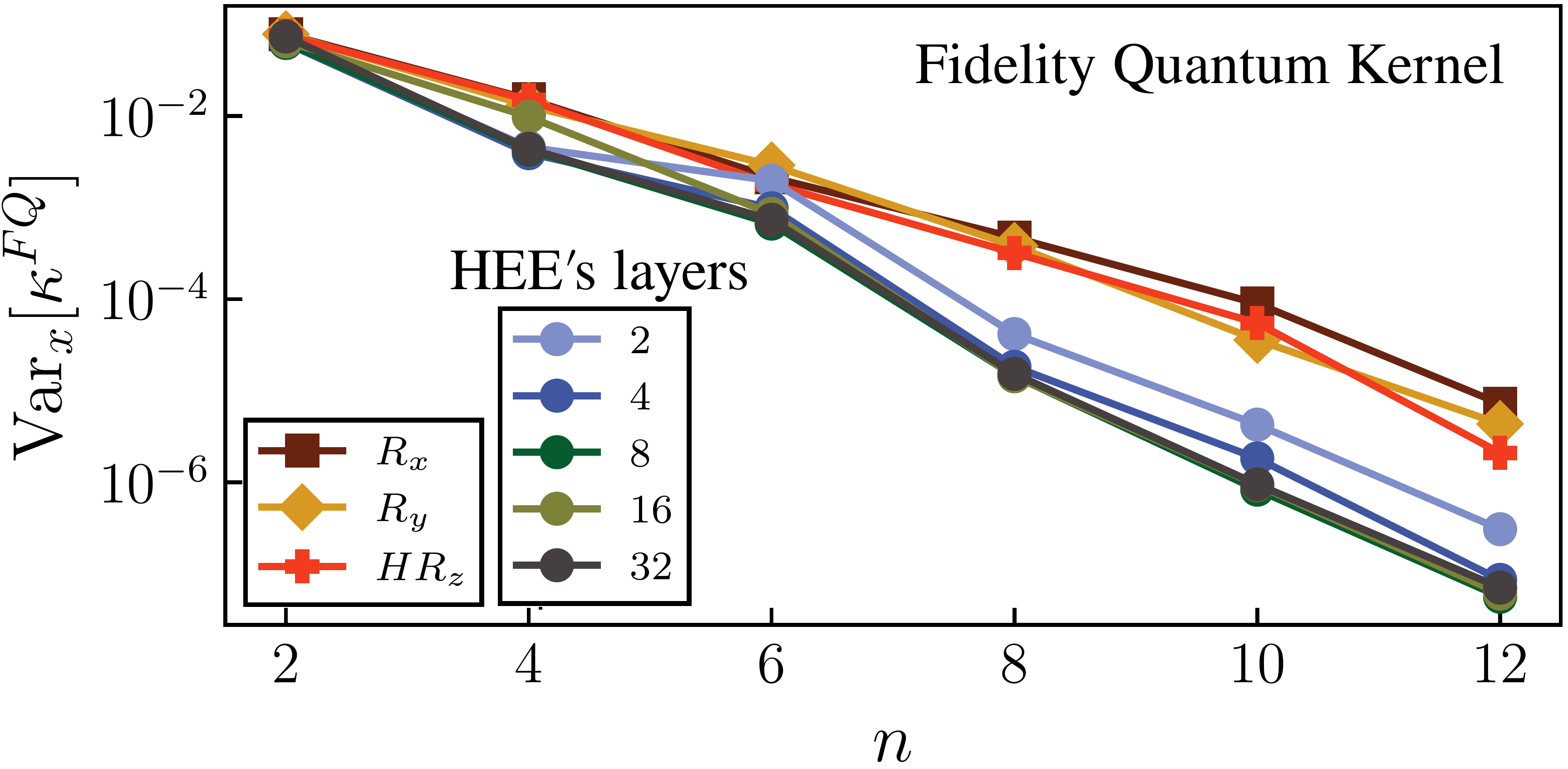}
\caption{\textbf{Global-measurement concentration of quantum kernels.} We plot the variance of the fidelity kernel as a function of $n$ using different data-embeddings, namely a single layer of one qubit rotations ($R_x$, $R_y$, Hadamard followed by $R_z$) and HEE with $L$ layers. The components of $N_s = 40$ input data points are independent and uniformly drawn from $[-\pi,\pi]$). } \label{fig:global}
\end{figure}

Nonetheless, it is important to stress that global measurements do not always lead to exponential concentration. For example, if the encoded quantum states are not ``too far away'' in Hilbert space, such that the fidelity kernel values concentrate no worse than polynomially in $n$, their overlap can be efficiently resolved. 
For example, the MNIST classification task does not satisfy the assumptions of Proposition~\ref{prop:global-measurement}. As a result, as shown in Fig.~\ref{fig:expressivity}, global measurements do not lead to the exponential concentration of the fidelity kernel for low depth ansatze. This demonstrates that the structure of the training data matters and global measurements do not always lead to exponential concentration.

Thus, the key message here is that when using global measurements to evaluate the kernel, the embedding must be chosen particularly carefully such that the fidelity between any pair of encoded quantum states is at least in $\Omega(1/{\poly(n)})$. To achieve this, one can either take the problem's structure into consideration when building the embedding~\cite{liu2021rigorous,glick2021covariant,larocca2022group,meyer2022exploiting} or further reduce the expressivity of problem-agnostic embeddings~\cite{shaydulin2021importance}.

\subsubsection{Noise-induced concentration}\label{sec:noise}

Hardware noise may disrupt and destroy information in the encoded quantum states, providing another source of concentration. 
To analyze the effect of noise, we here further suppose the data-embedding can be decomposed into $L$ layers of data-encoding unitaries
\begin{align} \label{eq:noise-embedding-mt}
    U(\vec{x}) = \prod_{l=1}^L U_l (\vec{x}_l)
\end{align}
where $\vec{x}_l$ is an input associated with $\vec{x}$ that is encoded in the layer $l$. We remark that from our construction, $\vec{x}_l$ can be either the $l^{\rm th}$ component of the input data $\vec{x}$ or a fixed vector $\vec{x}$ that is encoded repeatedly. 
Although the form of the data embedding is slightly less general than the one described in the noiseless sections, it still covers a large class of data embedding ansatze including the Hardware Efficient Embedding (HEE)~\cite{peters2021machine,hubregtsen2021training,wang2021towards,thanasilp2021subtleties,schuld2021supervised,havlivcek2019supervised}, the Quantum Alternative Operator Ansatz (QAOA)~\cite{lloyd2020quantum}, the Hamiltonian Variational Embedding (HVE)~\cite{huang2021power,shaydulin2021importance} and Instantaneous Quantum Polynomial (IQP) embedding~\cite{havlivcek2019supervised,thanasilp2021subtleties,shaydulin2021importance}.

We model the hardware noise as a Pauli noise channel applied before and after every layer of the embedding, similar to the model considered in Ref.~\cite{wang2020noise}. The output state of the noisy embedding circuit is given by
\begin{align}
    \tilde{\rho}(\vec{x}) = \mathcal{N}\circ \UC_L(\vec{x}_L) \circ \mathcal{N} \circ ... \circ \mathcal{N} \circ \UC_1(\vec{x}_1) \circ \mathcal{N} (\rho_0) \label{eq:noise-noise-evolution-mt}
\end{align}
where $\UC_l(\vec{x}_l)$ is the channel corresponding to the unitary $U_l(\vec{x}_l)$ and $\mathcal{N} = \mathcal{N}_1 \otimes ... \otimes \mathcal{N}_n $ is a local Pauli noise channel. Specifically, in this work we consider unital channels such that the effect of $\mathcal{N}_j$ on each local Pauli operator $\sigma \in \{ X,Y,Z \}$ is given by 
\begin{align}
    \mathcal{N}_j (\sigma) = q_\sigma \sigma \;, \label{eq:noise-noise-model}
\end{align}
where $-1 < q_{\sigma} < 1$. We remark that the noiseless regime corresponds to $q_{\sigma} = 1$ for all qubits. The strength of the noise can be quantified by a characteristic noise parameter $\gamma$ which is defined as
\begin{align}
    q = {\rm max}\{ |q_X|, |q_Y|, |q_Z|\} \;. \label{eq:noise-charac}
\end{align}

The following theorem summarizes the impact of noise on quantum kernels.
\begin{theorem}[Noise-induced concentration]\label{thm:noise-kernel}
Consider the $L$-layered data embedding circuit defined in Eq.~\eqref{eq:noise-embedding-mt} with input state $\rho_0$ and the layerwise Pauli noise model defined in Eq.~\eqref{eq:noise-noise-evolution-mt} with characteristic noise parameter $q< 1$. The concentration of quantum kernel values may be bounded as follows
\begin{align}
    \left| \tilde{\kappa}(\vec{x},\vec{x'}) - \mu \right| \leq F(q,L) \; . \label{eq:thm-noise}
\end{align}
\begin{enumerate}
    \item For the fidelity quantum kernel $\tilde{\kappa}(\vec{x},\vec{x'}) =\tilde{\kappa}^{FQ}(\vec{x},\vec{x'})$, we have $ \mu = 1/2^n$, and
    \begin{align}
        F(q,L) = q^{2L+1}  \left\| \rho_0 - \frac{\mathbb{1}}{2^n} \right\|_2\;.
    \end{align}
    \item For the projected quantum kernel $\tilde{\kappa}(\vec{x},\vec{x'}) =\tilde{\kappa}^{PQ}(\vec{x},\vec{x'})$, we have $ \mu = 1$, and
    \begin{align}
        F(q,L) = (8 \ln 2) \gamma n  q^{b(L+1)}S_2\left(\rho_0 \Big\|  \frac{\mathbb{1}}{2^{n}}\right) \;,
    \end{align}
    where $S_2( \cdot \| \cdot)$ denotes the sandwiched 2-R\'enyi relative entropy and $b = 1/(2\ln(2)) \approx 0.72 $.
\end{enumerate}
Additionally, the noisy data-encoded quantum state $\tilde{\rho}(\vec{x})$ concentrates towards the maximally mixed state as
\begin{align} \label{eq:noise-state-con-mt}
     \left\| \tilde{\rho}(\vec{x}) - \frac{\mathbb{1}}{2^n} \right\|_2 \leq q^{L+1} \left\| \rho_0 - \frac{\mathbb{1}}{2^n} \right\|_2 \; .
\end{align}
\end{theorem}

Theorem~\ref{thm:noise-kernel} shows that the concentration of quantum kernels due to noise is exponential in the number of layers $L$ for both the fidelity and projected quantum kernels. This is a consequence of the encoded state concentrating towards the maximally mixed state, as captured in Eq.~\eqref{eq:noise-state-con-mt}. 
In addition, we note that the noise-induced concentration bounds here are deterministic due to the noise acting independently of the input data. 

If quantum kernel-based methods are to provide any quantum advantage, the data embedding part must be hard to classically simulate.
For example, when using  embeddings with local connectivity, we are largely interested in the regime of moderately deep circuits where $L$ scales at least linearly in $n$~\cite{havlivcek2019supervised}. 
However, it is precisely this regime in which our bounds suggest kernels will exponentially concentrate due to an effect of noise. In particular, when the number of layers $L$ scales polynomially with the number of qubits $n$, $F(q,L)$ decays exponentially in the number of qubits.
We stress that the exponential decay nature of the concentration bounds persists for all $q<1$ and different values of noise characteristics only lead to different exponential decay rates. 

\medskip

The impact of noise on the concentration of kernels is studied in Fig~\ref{fig:noise} where we plot the average of $\left| \kappa^{FQ}(\vec{x},\vec{x'}) - 1/2^n\right|$ and $\left| \kappa^{PQ}(\vec{x},\vec{x'}) - 1\right|$ for the MNIST dataset as a function of the depth $L$ of the HEE embedding and the noise characteristic $q$. 
We observe exponential concentration with $L$, with the concentration stronger for higher noise levels $q$, in agreement with Theorem~\ref{thm:noise-kernel}.
We note that in our numerical simulations, noise only acts before and after single-qubit gates and we assume noiseless implementations of entangling gates. Therefore, in real experiments, where gate fidelity of entangling gates is generally worse than single-qubit gates, we expect the noise to dominate at a faster pace.   

\begin{figure}[t]
\includegraphics[width=.99\columnwidth]{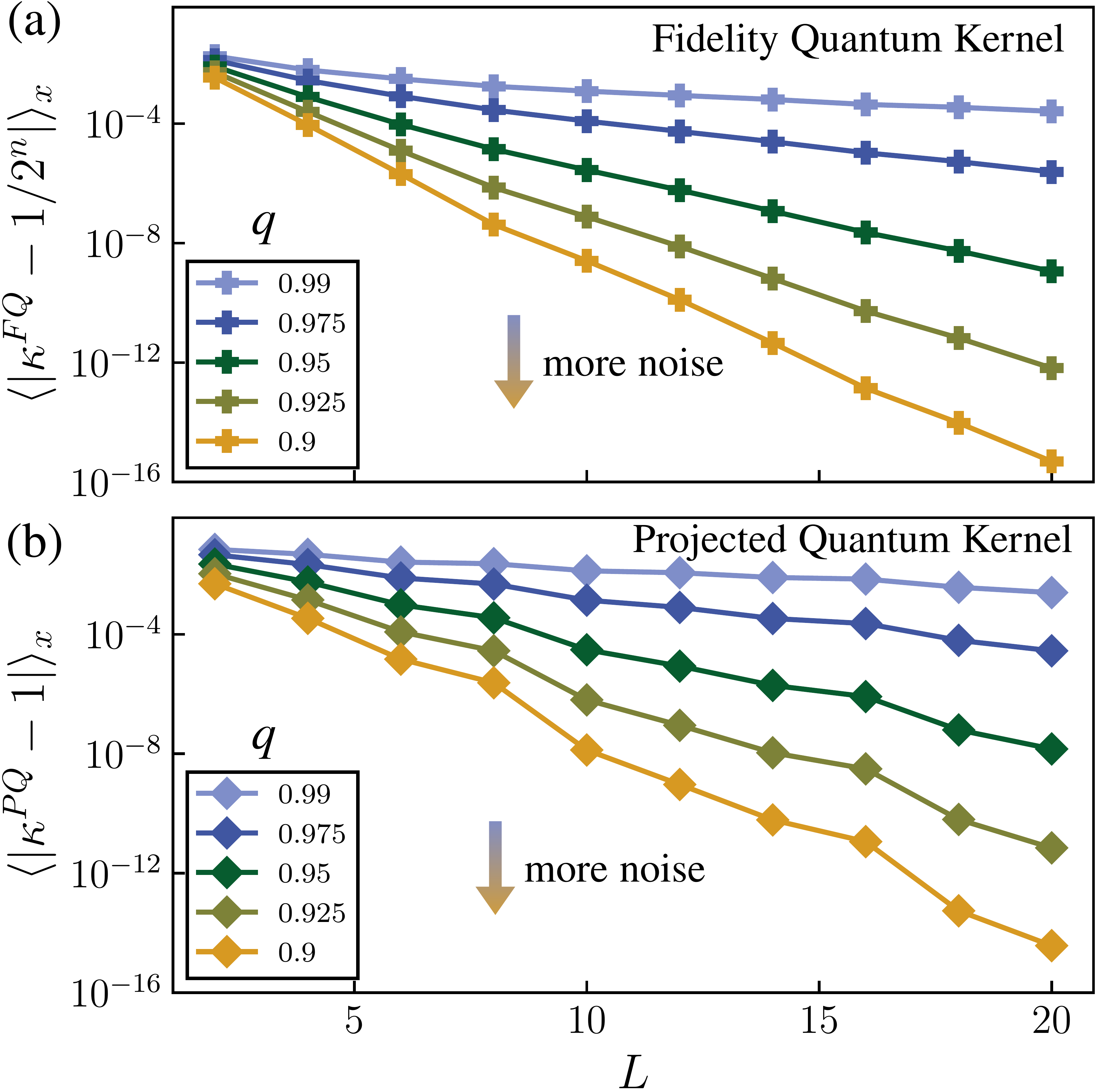}
\caption{\textbf{Effect of noise.} We plot the average of the difference between the quantum kernels  and their respective fixed point $\mu$ over different input data points and for different number of layers $L$ and noise parameter $q$. We consider the fidelity quantum kernel in (a) with $\mu = 1/2^n$ and the projected quantum kernel in (b) with $\mu=1$.
We use the MNIST dataset with $N_s = 40$ and $n = 8$.}\label{fig:noise}
\end{figure}

\medskip

In Appendix~\ref{appendix:em-fail}, we argue that, similar to noise-induced BPs in VQAs~\cite{wang2021can, takagi2021fundamental, quek2022exponentially}, the exponential concentration cannot be resolved with current common error mitigation techniques including Zero-Noise Extrapolation~\cite{li2017efficient,temme2017error,endo2018practical,kandala2018error}, Clifford Data Regression~\cite{czarnik2020error}, Virtual Distillation~\cite{huggins2020virtual,koczor2020exponential} and Probabilistic Error Cancellation~\cite{temme2017error,endo2018practical}. Hence, noise-induced concentration results poses a significant barrier to the successful implementation of quantum kernel methods on near term hardware.

\subsection{Training parameterized quantum kernels}\label{sec:thm-kta}

Given the problems associated with expressivity-induced concentration, it is generally advisable to avoid problem-agnostic embeddings and instead try and take advantage of the data structure of the problem. However, in many cases, constructing such problem-inspired embeddings is highly non-trivial. An alternative is to allow the data embedding itself to be parametrized and then train the embedding. Such strategies have been shown to improve generalization of the kernel-based quantum model~\cite{lloyd2020quantum,hubregtsen2021training}. \sth{We note that this is an additional process to train and select an appropriate embedding before implementing the standard quantum kernel algorithm (with this selected embedding).}

Here we consider a parametrized data embedding  $U(\vec{x},\vec{\theta})$, where $\vec{\theta}$ is a vector of trainable parameters (typically corresponding to single qubit rotation angles).  
For a given input data vector $\vec{x}$, an ensemble of data embedding unitaries can be generated by varying the parameters $\vec{\theta}$. This in turn generates a family of parametrized quantum kernels $\kappa_{\vec{\theta}}(\vec{x},\vec{x'})$.
Let $\vec{\theta} = \vec{\theta^*}$ be the optimal parameters found by training the embedding. The optimally embedded kernel now corresponds to $\kappa(\vec{x},\vec{x'}) = \kappa_{\vec{\theta^*}}(\vec{x},\vec{x'})$ and the remaining process to obtain the optimal model is the same as that described in Section~\ref{sec:framework-kernels-qml}.

The standard approach to obtain the optimal kernel is to train the parameters $\vec{\theta^*}$ via standard optimization techniques~\cite{lloyd2020quantum,hubregtsen2021training,cerezo2020variationalreview}, which in turn requires defining a loss function one needs to minimize. For instance, in a binary classification task where the true labels are either $+1$ or $-1$, the \textit{ideal} kernel is $+1$ if the input data are in the same class and is  $-1$ otherwise. In practice, however, one can only approximate the \textit{ideal} kernel as
\begin{align}
    \kappa_{\rm ideal}(\vec{x}_i,\vec{x}_j) = y_i y_j \;,
\end{align}
using the given training data $\mathcal{S}$. 
The kernel target alignment measures the similarity between the parameterized kernel and the approximated ideal kernel~\cite{hubregtsen2021training,cristianini2006on} 
\begin{align}\label{eq:kernel-ta-mt}
    \text{TA}(\vec{\theta}) = \frac{\sum_{i,j}y_i y_j \kappa_{\vec{\theta}}(\vec{x}_i,\vec{x}_j)}{\sqrt{\left( \sum_{i,j} (\kappa_{\vec{\theta}}(\vec{x}_i,\vec{x}_j))^2\right)\left( \sum_{i,j} (y_i y_j)^2\right)}} \;.
\end{align}
As minimizing the target alignment corresponds to aligning the parametrized kernel to the \textit{ideal} kernel, we can use $\text{TA}(\vec{\theta})$ as a loss function. \sth{Crucially, unlike the training of the model itself, the associated loss function for training the embedding is generally non-convex.}

Training the parameterized data-embedding $U(\vec{x},\vec{\theta})$ has been recently proposed as an approach to improve generalization   quantum kernel-based methods~\cite{hubregtsen2021training,lloyd2020quantum}. In particular,  Ref.~\cite{hubregtsen2021training} showed that optimizing the kernel target alignment ${\rm TA}(\vec{\theta})$ of Eq.~\eqref{eq:kernel-ta-mt} leads to data-embedding schemes with better performance than unstructured embeddings for various MNIST-based binary classification tasks.
However, this assumes that one can successfully train the target alignment. 

Here we study the trainability of ${\rm TA}(\vec{\theta})$. 
Namely, we discuss what features of the parameterized embedding $U(\vec{x},\vec{\theta})$ can lead to exponential concentration and therefore to exponentially flat parameter landscapes (i.e., a BP). 
First, we show that the variance of ${\rm TA}(\vec{\theta})$ with respect to the variational parameters $\vec{\theta}$ is upper bounded by the variances of the parameterized quantum kernels $\kappa_{\vec{\theta}}(\vec{x},\vec{x'})$

\begin{proposition}[Concentration of kernel target alignment]\label{prop:target-kernel}
Consider an arbitrary parameterized kernel $\kappa_{\vec{\theta}}(\vec{x},\vec{x'})$ and a training dataset $\{ \vec{x}_i, y_i \}_{i=1}^{N_s}$ for binary classification with $y_i = \pm 1$. The probability that the kernel target alignment ${\rm TA}(\vec{\theta})$ (defined in Eq.~\eqref{eq:kernel-ta-mt}) deviates from its mean value is approximately bounded as 
\begin{align}
    {\rm Pr}_{\vec{\theta}}[|{\rm TA}(\vec{\theta}) - \mathbb{E}_{\vec{\theta}}[{\rm TA}(\vec{\theta})]| \geq \delta]  \leq \frac{ M \sum_{i,j}  \Var_{\vec{\theta}}[\kappa_{\vec{\theta}}(\vec{x}_i,\vec{x}_j)]}{\delta^2} \;,
\end{align}
with $M =\frac{8+N_s^3\left(9(N_s-1)^2+16\right)}{4N_s}$. 
\end{proposition}

If $ \Var_{\vec{\theta}}[\kappa_{\vec{\theta}}(\vec{x}_i,\vec{x}_j)]$ vanishes exponentially in the number of qubits for all pairs in the training data, the probability that $\text{TA}(\vec{\theta})$ deviates by an amount $\delta$ from its mean vanishes exponentially with the size of the problem.  
In this case, the parameter landscape of $\text{TA}(\vec{\theta})$ becomes exponentially flat and hence $\text{TA}(\vec{\theta})$ is untrainable with a polynomial number of measurement shots. 

In Appendix~\ref{appendix:features-ta}, we analyze features leading to exponentially vanishing variances $ \Var_{\vec{\theta}} [\kappa_{\vec{\theta}}(\vec{x}_i,\vec{x}_j)]$ and find that the same ones that lead to BPs for QNNs lead to BPs here. Namely, features that are deemed detrimental for trainability in QNNs such as deep unstructured circuits~\cite{mcclean2018barren,holmes2021connecting} and global measurements~\cite{cerezo2020cost} also lead to BPs here. Thus these features should be avoided when designing parameterized data embeddings for quantum kernels. 

We numerically demonstrate the effect of global measurements on training an embedding in Fig~\ref{fig:ta-kernel}. The data embedding consists of a single layer of parameterized single-qubit rotations around $y$-axis followed by a single layer of HEE. We study the variance of the kernel target alignment ${\rm TA}(\vec{\theta})$ (which determines the flatness of the training landscape~\cite{mcclean2018barren,cerezo2020cost}) for 500 random initialization of the parameters $\vec{\theta}$. As expected, since the parametrized block  acts globally on all qubits, ${\rm TA}(\vec{\theta})$ exponentially concentrates when one averages over the trainable parameters $\vec{\theta}$. 

\begin{figure}[t]
\includegraphics[width=.99\columnwidth]{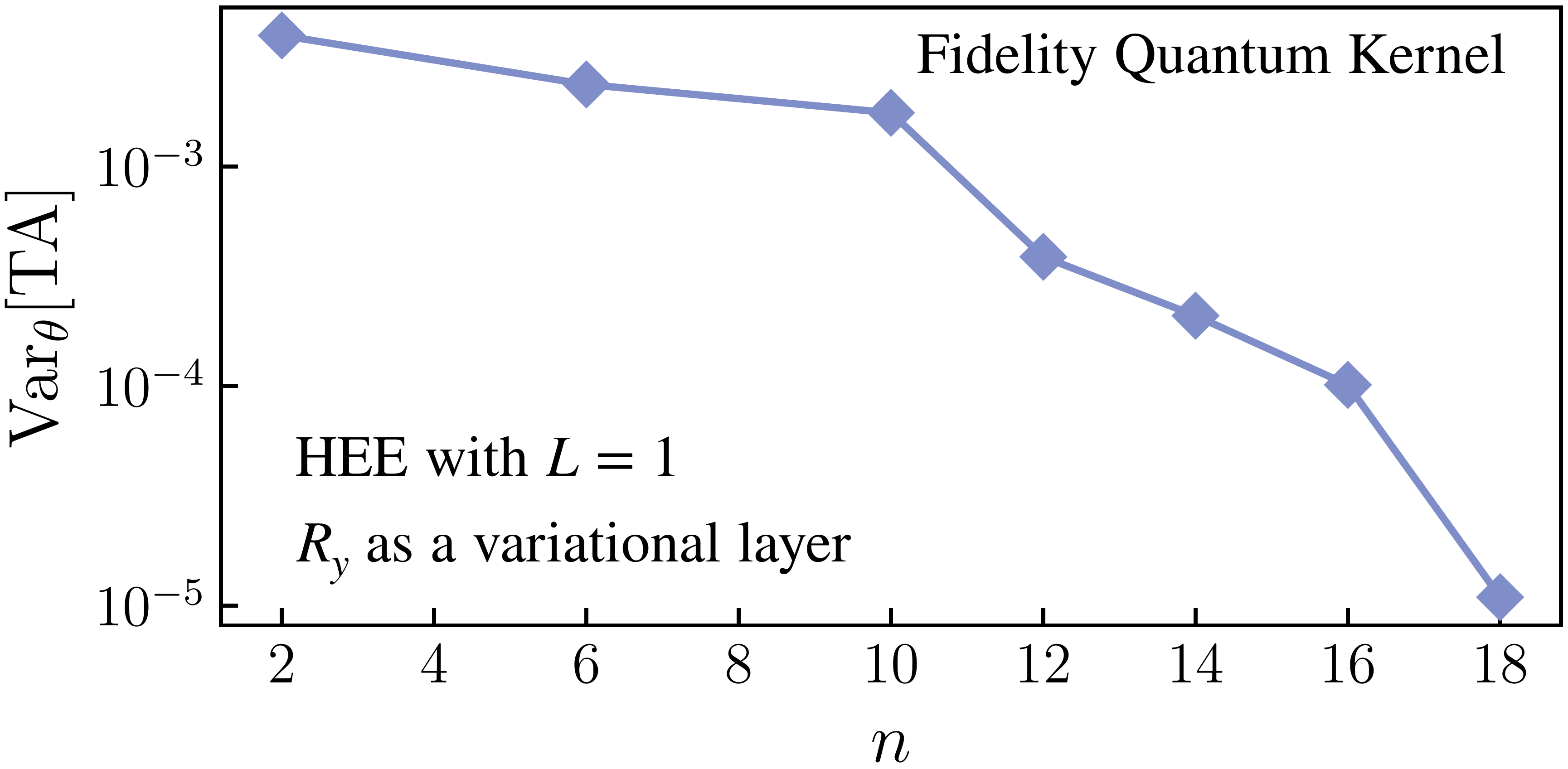}
\caption{\textbf{Kernel target alignment.} The variance of ${\rm TA}(\vec{\theta})$ with respect to variational parameters is plotted as a function of $n$. Here we use the hypercube dataset with $N_s =10$.}
\label{fig:ta-kernel}
\end{figure}

\section{Discussion}
Quantum kernels stand out as a promising candidate for achieving a practical quantum advantage in data analysis. 
This is in part due to the common belief that the optimal quantum kernel-based model can always be obtained~\cite{schuld2022is,schuld2021supervised,gentinetta2022complexity,hofmann2008kernel} due to the convexity of the problem. Although this is true, provided that the kernel values can be efficiently obtained to a sufficiently high precision, here we show that there exist scenarios where quantum kernels are exponentially concentrated towards some fixed value and so exponential resources are required to accurately estimate the kernel values. With only a polynomial number of shots, the predictions of the trained model become insensitive to input data and the model performs trivially on unseen data, that is, generalizes poorly. Crucially, in this context generalization cannot be improved by training on more input data points but rather by increasing the number of measurement shots (or using a more appropriate embedding).
It is worth stressing that as we assume very little on the form of the data embedding $U(\vec{x})$, our analytical bounds hold for a wide range of embedding architectures and schemes, including both problem-agnostic and problem-inspired embeddings.  

Our results highlight four aspects to carefully consider when choosing a data embedding for quantum kernels.
While much of the literature currently focuses on using problem-agnostic quantum embeddings for quantum kernels~\cite{peters2021machine,hubregtsen2021training,wang2021towards,schuld2021supervised}; these are typically highly expressive and as such should generally be avoided. 
Entanglement can also be detrimental when combined with local quantum kernels such as the projected quantum kernels, and suggests that one should be mindful about using embeddings leading to states satisfying volume-laws of entanglement. 
Our results on global measurements demonstrate that the fidelity kernel can exponentially concentrate even with a simple embedding that has low expressivity and no entanglement. 
Consequently, the fidelity kernel should only be used for datasets where the data-embedded states are  `not too distant' in the Hilbert space. 
Finally, our study of noise suggests that polynomial-depth data embeddings in noisy hardware suffer from exponential concentration, thus presenting a serious barrier to achieve a meaningful quantum advantage in the near-term. 

In addition, we show that training parametrized quantum kernels using kernel target alignment suffers from an exponentially flat training landscape under similar conditions to those leading to barren plateaus in QNNs. That is, when constructing the parametrized part of the data embeddings, one should avoid features that induce BPs as QNNs such as global measurements and deep unstructured circuits.

Our work provides a systematic study of the barriers to the successful scaling up of quantum kernel methods posed by exponential concentration. Prior work on BPs motivated the community to search for ways to avoid or mitigate BPs such as employing correlated parameters~\cite{volkoff2021large} using tools from quantum optimal control~\cite{larocca2021diagnosing,larocca2021theory}, or developing the field of geometrical quantum machine learning~\cite{larocca2022group,meyer2022exploiting,skolik2022equivariant,sauvage2022building}.
In a similar manner, we stress our results should not be understood as condemning quantum kernel methods, but rather a prompt to develop exponential-concentration-free embeddings for quantum kernels. Crucially, incorporating quantum aspects to machine learning does not always lead to better performance. Indeed, often it will only worsen the performance of the learning models. In particular, if one remains restricted to mimicking the classical techniques without carefully taking into account quantum phenomena, it is unlikely that one will achieve a quantum advantage. Hence distinctly quantum approaches, using specialized quantum structures/symmetries, may prove to be the way forward~\cite{havlivcek2019supervised, glick2021covariant}.

\section{Data Availability}
Data generated and analyzed during the current study are available from the corresponding author upon reasonable request.

\section{Code Availability}
Code used for the current study is available from the corresponding author upon reasonable request.

\bigskip

\bibliography{quantum.bib}

\section{Acknowledgements}
We thank the reviewers at Nature Communications and QIP for their valuable feedback and Jonas M. K\"{u}bler for his comments on Appendix A. ST is supported by the National Research Foundation, Prime Minister's Office, Singapore and the Ministry of Education, Singapore under the Research Centres of Excellence programme and subsequent support from the Sandoz Family Foundation-Monique de Meuron
program for Academic Promotion. SW is supported by the Samsung GRP grant. 
M.C. was initially supported by ASC Beyond Moore’s Law project at Los Alamos National Laboratory (LANL). This work was also supported by the Quantum Science Center (QSC), a National Quantum Information Science Research Center of the U.S. Department of Energy (DOE). ZH acknowledges initial support from the LANL Mark Kac Fellowship and subsequent support from the Sandoz Family Foundation-Monique de Meuron program for Academic Promotion.


\section{Competing interests}
The authors declare no competing interests.

\clearpage
\newpage
\onecolumngrid
\setcounter{theorem}{0}
\setcounter{proposition}{2}
\setcounter{corollary}{0}
\setcounter{definition}{0}

\appendix
\vspace{0.5in}
\begin{center}
	{\Large \bf Appendix} 
\end{center}
\section{Related work}\label{ap:priorwork}

\subsection{Exponential concentration in fidelity quantum kernels}

The observation that using fidelity quantum kernel could lead to poor generalization was first made in Ref.~\cite{huang2021power}. In particular, this paper provided a rigorous generalization bound that could be used to compare the predictive power of quantum and classical kernel-based models. Based on this bound, the authors argued that for high dimensional problems the embedded states are likely to be `far from each other' and so have either comparable or inferior performance compared with its classical counterparts. This work also provided numerical evidence of this observation for up to 30 qubits. Similar numerical evidence was provided more recently in Ref.~\cite{slattery2023numerical}. 

Subsequently, Ref.~\cite{kubler2021inductive} argued that not only are quantum models unable to outperform classical models but, more generally, embeddings that lack an inductive bias lead to models that generalize poorly. To demonstrate this, the authors analyzed spectral properties of the quantum fidelity kernel integral operator. Specifically, they lower bounded a model's generalization error in terms of the largest eigenvalue of the kernel integral operator (Theorem 3 in Appendix~D). This result is then used to show that fidelity kernels with an unstructured product embedding will lead to large risks and so poor generalization (Theorem~1). 
Since this result holds for a product embedding, in our language this could be viewed as globality induced concentration. While not shown explicitly in  Ref.~\cite{kubler2021inductive} it is plausible that Theorem~3 could also be used as an alternative approach to proving that highly expressive or entangling embeddings lead to poor generalization.

We note that while both these important works highlight problems with the fidelity kernel related to exponential concentration, the exact causes of exponential concentration were not analysed in details. Furthermore, in both cases, the detrimental effect of exponential concentration was studied assuming direct access to quantum states without shot noise. 

\subsection{Exponential concentration in projected quantum kernels}

As a potential solution to the problems with the fidelity quantum kernel, the authors in Ref.~\cite{huang2021power} introduced the projected quantum kernel where the data encoded quantum states are projected back onto local subspaces with the similarity between quantum states being collectively compared at the local level as defined in Eq.~\eqref{eq:projected-gaussian-kernel-mt}. The projected quantum kernel can be challenging to evaluate (having gone through the exponentially large Hilbert space before projection) and yet remain reasonably expressive. The authors showed that for a synthesized dataset a model based on the projected quantum kernel can outperform a wide range of classical machine learning models and numerically verified this up to 30 qubits. 

Nevertheless, there remain some open questions including how exactly expressivity of the data embedding can affect the performance of the projected quantum kernel and whether too much expressivity can lead to exponential concentration or not. In Ref.~\cite{kubler2021inductive} the authors considered a simplified version of the original projected quantum kernel defined as the overlap between two reduced data encoded states onto the first qubit. The authors then argued that the projected quantum kernel with an embedding consisting of a layer of data-dependent single qubit rotations followed by a fixed data-independent Haar random unitary could have an inductive bias that is hard to simulate classically. However, the authors then prove that the embedded quantum states exponentially concentrate towards the maximally mixed state. Thus they suggest (but do not explicitly prove) that exponentially many measurement shots will be needed for such embeddings (similarly to the barren plateaus phenomena in QNNs). Crucially, exponential concentration here does not directly come from the randomness in the input data distribution but rather from a data-independent part of the embedding. In contrast to this, our work concerns the expressivity induced by the interplay between the input data distribution and the embedding. Additionally, we consider an arbitrary data-dependent embedding and the original form of the projected quantum kernel (but our results can be easily extended to other forms of the projected kernel including the one in Ref.~\cite{kubler2021inductive}). More generally, we identify noise and entanglement as additional sources of exponential concentration for projected kernels and analyse the consequences of exponential concentration in the presence of shot noise for model predictions.

\subsection{Attempts to mitigate the exponential concentration}

One proposal to mitigate exponential concentration for the fidelity quantum kernel is to re-scale the input data with some hyperparameter~\cite{shaydulin2021importance, canatar2022bandwidth}. Consequently, the data encoded quantum states become clustered closer together and hence result in a lower expressivity. This idea was first numerically demonstrated in Ref.~\cite{shaydulin2021importance} and analytical treatment on how this hyperparameter affects the kernel spectrum was done in the follow-up work in Ref.~\cite{canatar2022bandwidth}. While this could ensure that the fidelity kernel does not suffer from expressivity induced concentration, it was shown later in Ref.~\cite{slattery2023numerical} with numerical simulations of up to 20 qubits that this approach is unlikely to provide any quantum advantage over classical models. 

Ref.~\cite{suzuki2022quantum} proposed a new type of quantum kernel known as the quantum Fisher kernel. The kernel encodes geometric information of the input data. The author considered the alternating layered ansatz where the embedding consists of layers of unitary blocks that act on local qubits and analytically showed the absence of exponential concentration with log depth layers. We note that extending the layers to linear depth still leads to a highly expressive embedding and exponential concentration.  

Another line of work aims to optimize a parameterized embedding for a quantum kernel when there is no prior knowledge of the data structure. Ref.~\cite{hubregtsen2021training} used kernel target alignment (as defined in Eq.~\eqref{eq:kernel-ta-mt}) to align a parameterized embedding with an ideal embedding approximated with the given training data. In Ref.~\cite{incudini2022structure}, the authors rely on a more heuristic approach to slowly build the embedding from a set of unitary blocks, which is similar to a layer-wise training strategy in variational quantum algorithms. We note that, as shown in Sec.~\ref{sec:thm-kta}, that training the embedding can potentially lead to barren plateaus when the trainable part of the embedding is not constructed properly. 

\subsection{Effect of shot noise in the absence of exponential concentration}
Ref.~\cite{liu2021rigorous} proved a quantum advantage for a quantum support vector machine using a fidelity quantum kernel for solving a particular classification task where the dataset is engineered based on the discrete log problem. As the discrete log problem is strongly believed to be inefficient for classical computers and efficient for quantum computers, the authors proved that this classical hardness carries over to the learning task. Thus using a carefully constructed embedding with a strong inductive bias aligned to the problem structure allows for a quantum advantage. Importantly, the advantage remains in the presence of shot noise. In particular, the kernel values and model predictions can be efficiently evaluated with a polynomial number of measurement shots. 

In Ref.~\cite{gentinetta2022complexity}, the effect of shot noise in quantum support vector machines was studied for an arbitrary classification task. It was analytically shown that in the absence of exponential concentration and under the assumption that a separation between two classes is polynomially large the performance of a kernel model is robust against shot noise.

\section{Preliminaries for statistical indistinguishability} 
\label{appendix:stat-indis-basic}
Here we quote some key technical tools on the distinguishability of probability distributions from binary hypothesis testing, which we will use to establish how the exponential concentration affects the kernel methods in Appendix~\ref{appendix:fidelity-swap} and~\ref{appendix:projected}. For a more extensive exposition we refer the reader to Ref.~\cite{tsybakov2009introduction}. 

\subsection{One sample}
\sth{
\begin{lemma}\label{lemma:guess-distribution}
Consider two probability distributions $\PC$ and $\QC$ over some finite set of outcomes $\IC$. Suppose we are given a single sample $S$ drawn from either $\PC$ or $\QC$ with equal probability, and we have the following two hypotheses:
\begin{itemize}
    \item Null hypothesis $\mathcal{H}_0$: $S$ is drawn from $\mathcal{P}$\,,
    \item Alternative hypothesis $\mathcal{H}_1:$ $S$ is drawn from $\mathcal{Q}$\,.
\end{itemize}
The probability of correctly deciding the true hypothesis is upper bounded as
\begin{align}
    {\rm Pr}[``{\rm right \; decision \; between \, } \HC_0 \, {\rm and} \, \HC_1"] \leq \frac{1}{2} + \frac{\| \PC - \QC \|_1}{4} \;, 
\end{align}
where we denote $\| \PC - \QC \|_1 = \sum_{s \in \IC} |p(s) - q(s)|$ as the 1-norm between the probability vectors (2 $\times$ the total variation distance).
\end{lemma}}
\begin{proof}
There exists a region $\AC$ such that $p(s) > q(s)$ for all $s \in \AC$. The optimal decision making strategy is to choose that the given sample $S$ is drawn from $\PC$ if it falls in the region i.e., $S \in \AC$ and choose $\QC$, otherwise. The probability of \sth{making the right decision} can be expressed as
\begin{align}
    {\rm Pr}[\sth{``{\rm right \; decision \; between \, } \HC_0 \, {\rm and} \, \HC_1"}] =&   {\rm Pr}( S \in \AC | S \sim \PC) {\rm Pr}(S \in \PC) + {\rm Pr}( S \notin \AC | S \sim \QC){\rm Pr}(S \in \QC) \\
    =& \frac{1}{2}\left[ {\rm Pr}( S \in \AC | S \sim \PC) + {\rm Pr}( S \notin \AC | S \sim \QC)  \right] \\
    =& \frac{1}{2}\left[ \sum_{s \in \AC} p(s) + \sum_{s \notin \AC} q(s) \right] \;, \label{eq:appx-swap1}
\end{align}
where the second equality is due to the sample being equally likely to be drawn from either $\PC$ or $\QC$. In the last equality, we use the fact that given that the sample is from $\PC$, the probability that this sample takes any value within the region $\AC$ is simply $ \sum_{s \in \AC} p(s)$, and similarly for $s \notin \AC$.

The 1-norm between probability vectors can be written as 
\begin{align}
    \| \PC - \QC \|_1 = & \sum_{s \in \IC} |p(s) - q(s)| \\
    = & \sum_{s \in \AC} (p(s) - q(s)) + \sum_{s \notin \AC} (q(s) - p(s)) \label{eq:tvd}\;,
\end{align}
where we have separated terms in the sum based on the region $\AC$.
Lastly, we notice that
\begin{align}
    \frac{2 + \| \PC - \QC \|_1}{2} & = \frac{1}{2}\left(\sum_{s\in\IC}p(s) + \sum_{s\in\IC}q(s)+ \| \PC - \QC \|_1 \right)\\
    & = \sum_{s \in \AC} p(s) + \sum_{s \notin \AC} q(s) \;,
\end{align}
where in the second line we have used Eq.~\eqref{eq:tvd}. Substituting this back to Eq.~\eqref{eq:appx-swap1}, we obtain the desired result. 
\end{proof}

\subsection{Many samples}
We now consider the scenario where instead of a single sample we are given $N$ samples from either $\PC$ and $\QC$ is given to us and we have to guess which of the two distributions these samples are drawn from. At first glance, it may seem that Lemma~\ref{lemma:guess-distribution} is not applicable to this scenario, since we now have a set of outcomes rather than just one sample. However, we can consider the product distributions $\PC^{\otimes N}$ and $\QC^{\otimes N}$, where a single sample corresponds to $N$ samples from $\PC$ and  $\QC$ respectively.

We first state a generic inequality on product distributions.

\sth{\begin{lemma}\label{lemma:1norm-product}
The 1-norm between discrete product distributions $\PC^{\otimes N}$ and $\QC^{\otimes N}$ can be upper bounded as
\begin{align}
    \| \PC^{\otimes N} - \QC^{\otimes N} \|_1 \leq N \| \PC - \QC \|_1 \;.
\end{align}
\end{lemma}
\begin{proof}
We have
\begin{align}
     \| \PC^{\otimes N} - \QC^{\otimes N} \|_1 & = \| \PC^{\otimes N} \,-\, \QC^{}\otimes \PC^{\otimes N-1} \,+\, \QC^{}\otimes \PC^{\otimes N-1} \,-\, \QC^{\otimes 2}\otimes \PC^{\otimes N-2} \,+\,  ... \,+\,  \QC^{\otimes N-1} \otimes \PC^{} \,-\, \QC^{\otimes N} \|_1 \\
     & \leq \| \PC^{\otimes N} \,-\, \QC^{}\otimes \PC^{\otimes N-1}\|_1 \,+\, \|\QC^{}\otimes \PC^{\otimes N-1} \,-\, \QC^{\otimes 2} \otimes \PC^{\otimes N-2}\|_1 \,+\,  ... \,+\,  \|\QC^{\otimes N-1} \otimes \PC^{} \,-\, \QC^{\otimes N} \|_1 \\
     & = \| \PC - \QC \|_1\, \|\PC^{\otimes N-1} \|_1 \,+\, \|\QC \|_1\, \|\PC^{}- \QC^{}\|_1\, \|\PC^{\otimes N-2} \|_1 \,+\,  ... \,+\,  \|\QC^{\otimes N-1}\|_1\,\|\PC^{} - \QC \|_1 \\
     & =  N \| \PC - \QC \|_1 \;,
\end{align}
\normalsize
where in the first line we have added and subtracted terms, the inequality is due to the triangle inequality, and the third line is due to the fact that the 1-norm factorizes.
\end{proof}}

Supplemental Lemma \ref{lemma:1norm-product} along with Supplemental Lemma \ref{lemma:guess-distribution} immediately implies an upper bound on a hypothesis testing experiment using $N$ samples. In the following proposition we specify this for binary distributions.

\sth{\begin{supplemental_proposition}\label{sup-prop:indistin-prob}
Consider two binary probability distributions $\PC_0=(p_0,1-p_0)$ and $\PC_\varepsilon=(p_\varepsilon,1-p_\varepsilon)$ where $p_\epsilon = p_0+\varepsilon$. Suppose we are given a $N$ samples (denoted $\MC$) drawn from either $\PC_0$ or $\PC_\varepsilon$ with equal probability, and we have the following two hypotheses:
\begin{itemize}
    \item Null hypothesis $\mathcal{H}_0$: $\mathcal{M}$ is drawn from $\mathcal{P}_0$\,,
    \item Alternative hypothesis $\mathcal{H}_1: \mathcal{M}$ is drawn from $\PC_\varepsilon$\,.
\end{itemize}
The probability of correctly deciding the true hypothesis is upper bounded as
\begin{align}
    {\rm Pr}[``{\rm right \; decision \; between \, } \HC_0 \, {\rm and} \, \HC_1"] \leq \frac{1}{2} + \frac{N|\varepsilon|}{2} \;.
\end{align}
\end{supplemental_proposition}

\begin{proof}
We remark that the combination of Supplemental Lemma \ref{lemma:1norm-product} along with Supplemental Lemma \ref{lemma:guess-distribution} gives success probability 
\begin{align}
     {\rm Pr}[``{\rm right \; decision \; between \, } \HC_0 \, {\rm and} \, \HC_1"] \leq \frac{1}{2} + \frac{N \|\PC_0 - \PC_{\varepsilon}\|_1}{4}
\end{align}
and explicit evaluation shows that $\|\PC_0 - \PC_{\epsilon}\|_1 = 2|\varepsilon|$.
\end{proof}}

In the rest of this work, we mainly focus on a perturbation that is exponentially small in the system size $n$, i.e., $\epsilon \in \OC(1/b^n)$ for some $b>1$. 
\color{Black}

\section{Practical implications of exponential concentration on kernel methods}

In this section we analyse the consequences of exponential concentration on kernel methods. Specifically, we show that when using a polynomial number of measurements the statistical estimate of the Gram matrix is with high probability independent of input data. Consequently, training with the estimated Gram matrix results in a data-independent model. It then follows that the final (trained) output model is independent of the training data and cannot generalize. To make this argument more concrete we consider the example of kernel ridge regression; however the fundamental problem of the data-independence of the output prediction caries over similarly to other learning tasks. 

\medskip

\sth{
We present this argument for both the fidelity quantum kernel and projected quantum kernel, for different corresponding strategies to prepare them. The rest of this section is structured as follows.

In Appendix~\ref{appendix:fidelity} we discuss the fidelity quantum kernel where two measurement strategies to estimate kernel values are considered.
\begin{itemize}
    \item Appendix~\ref{appendix:fideliy-overlap-test} concerns with the Loschmidt Echo test to estimate kernel values. In the presence of exponential concentration, we rigorously show that the statistical estimates of the kernel values concentrate at zero with high probability (see Supplemental Proposition~\ref{sup-prop:fidelity-overlap}). As a consequence, the estimated Gram matrix is likely to simply be the identity matrix and the estimated model predictions also concentrate to zero with high probability (see Supplemental Corollary~\ref{sup-cor:train-fqk-overlap}). 
    \item Appendix~\ref{appendix:fidelity-swap} is concerned with the SWAP test to estimate kernel values. Here, we rely on a reduction to hypothesis testing (see Appendix~\ref{appendix:stat-indis-basic} for preliminaries), and we define notions of statistical indistinguishability (see Definition~\ref{def:StatIndist} for distributions and Definition~\ref{def:StatIndistOutputs} for outputs). When kernel values exponentially concentrate, their estimates become statistically indistinguishable with high probability (see Supplemental Lemma~\ref{sup-lem:fidelity-swap}). Consequently, outputs of the model trained on these kernel estimates are also indistinguishable and insensitive to unseen input data (see Supplemental Corollary~\ref{sup-coro:fidelity-swap}).     
    \item In Appendix~\ref{appendix:fidelity-numerics} we present numerical simulations of the fidelity kernel to support the theoretical results in the previous sub-sections.
\end{itemize}

We then discuss the projected quantum kernel in Appendix~\ref{appendix:projected} and investigate two measurement strategies to obtain kernel estimates.
\begin{itemize}
    \item Appendix~\ref{appendix:projected1} is concerned with the practical consequence of exponential concentration. For both measurement strategies, the effect appears in an identical manner as in the SWAP test for the fidelity kernel. That is, we have statistical indistinguishability of kernel estimates (see Supplemental Proposition~\ref{sup-prop-norm-indis-pqk}) and their model predictions (see Corollary~\ref{sup-coro:indis-model-pqk}) . Since the setting in the projected kernel is more complicated than the fidelity kernel with the SWAP test, we encourage interested readers to first review Appendix~\ref{appendix:fidelity-swap} as the ideas are similar (or Appendix~\ref{appendix:stat-indis-basic} for preliminaries on hypothesis testing). 
    \item In Appendix~\ref{appendix:projected2}, we illustrate numerical results to back up our theoretical findings.
    \item In order to not interrupt the flow when going through the first two sub-sections, we group all the proofs together in  Appendix~\ref{appendix:pqk-proof}
\end{itemize}

In Appendix~\ref{appendix:indis-states}, we discuss the extension of kernel concentration to state concentration. This leads to a stronger concentration result which cannot be resolved even with quantum access to polynomial state copies. Lastly, in Appendix~\ref{appendix:further-discuss-impact}, we provide discussion on a sufficient condition to resolve the exponential concentration issue. We found that the number of measurement shots has to scale exponentially in the number of qubits in order to acquire enough resolution in the kernel estimates. This exponential scaling is impractical for large problem sizes.}

\begin{figure}[t]
\includegraphics[width=.65\columnwidth]{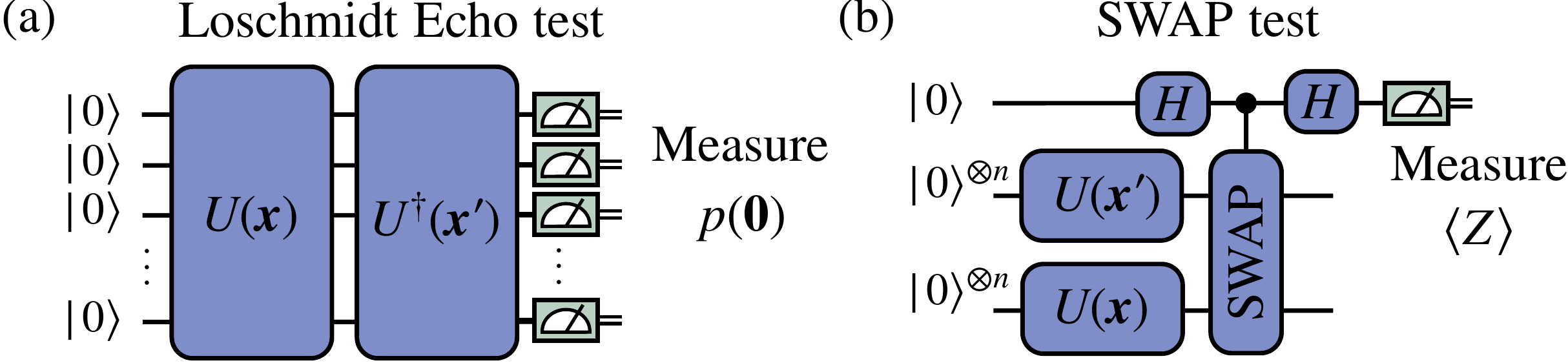}
\caption{\textbf{Schematic diagram of tests.} We illustrate two different strategies to estimate kernel values. In panel (a) we show the Loschmidt Echo test where the estimated kernel value is equivalent to the empirical probability of measuring the all-zero bitstring. In panel (b) we show the SWAP test where the kernel value is estimated with the expectation value of Pauli Z operator on an ancilla qubit.}\label{fig-sup:schematic-swap-overalp}
\end{figure}

\subsection{Fidelity quantum kernel}\label{appendix:fidelity}
Let us start by recalling that an input $\vec{x}$ is encoded into a data-encoding quantum state $\rho(\vec{x})$ through an embedding unitary $U(\vec{x})$ and the fidelity quantum kernel is of the form
\begin{align}
    \kappa^{\rm FQ}(\vec{x},\vec{x'}) = \Tr[\rho(\vec{x})\rho(\vec{x'})] \;.
\end{align}
The exact value of the kernel is inaccessible and instead we obtain a statistical estimate using measurement outcomes/shots from quantum computers. There are two common measurement strategies to estimate the fidelity quantum kernel: (i) the Loschmidt Echo test or (ii) the SWAP test, as shown in Fig.~\ref{fig-sup:schematic-swap-overalp}. In either case, the fidelity quantum kernel is equivalent to the expectation value of an observable $O$ with some quantum state $\rho$ with the exact expression for $O$ and $\rho$ depending on the strategy used. If we write the eigendecomposition of the observable as $O = \sum_i o_i |o_i\rangle\langle o_i |$ where $o_i$ and $ |o_i\rangle$ are eigenvalues and eigenvectors of $O$, then the statistical estimate after $N$ measurements is of the form
\begin{align}\label{eq:estimate-fidelity-kernel}
    \widehat{\kappa}^{\rm FQ}(\vec{x},\vec{x'}) = \frac{1}{N}\sum_{m=1}^{N} \lambda_m \;, 
\end{align}
where $\lambda_m$ is the outcome of the $m^{\rm th}$ measurement and can be treated as a random variable which takes the value $o_i$ with probability $p_i = \Tr[|o_i\rangle\langle o_i | \rho]$. \\

\sth{We now restate the definition of exponential concentration discussed in the main text.
\begin{definition} [Exponential concentration]\label{def:appx-exp-concentration}
Consider a quantity $X(\vec{\alpha})$ that depends on a set of variables $\vec{\alpha}$ and can be measured from a quantum computer as the expectation of some observable. $X(\vec{\alpha})$ is said to be deterministically  exponentially concentrated in the number of qubits $n$ towards a certain $\vec{\alpha}$-independent value $\mu$ if
\begin{align}
    |X(\vec{\alpha}) - \mu |\leq \beta \in O(1/b^n) \;,
\end{align}
for some $b>1$ and all $\vec{\alpha}$. Analogously, $X(\vec{\alpha})$ is probabilistically exponentially concentrated if
\begin{align} \label{eq:def-prob-concentration}
    {\rm Pr}_{\vec{\alpha}}[|X(\vec{\alpha}) - \mu| \geq \delta] \leq \frac{\beta}{\delta^2} \;\; , \; \beta \in O(1/b^n) \;,
\end{align}
for $b> 1$. That is, the probability that $X(\vec{\alpha})$ deviates from $\mu$ by a small amount $\delta$ is exponentially small for all $\vec{\alpha}$.

In addition, if $\mu$ exponentially vanishes in the number of qubits i.e., $\mu \in \OC(1/b'^n)$ for some $b' >1$, we say that $X(\vec{\alpha})$ exponentially concentrates towards an exponentially small value.

\end{definition}
\noindent In the context of quantum kernels, $X(\vec{\alpha}) = \kappa^{\rm FQ}(\vec{x},\vec{x'})$ with the set of variables corresponding to an input data pair $\vec{\alpha} = \{ \vec{x}, \vec{x'}\}$. When $\mu$ vanishes exponentially, we remark that the probability of deviating from zero by an arbitrary constant amount is exponentially small.}

We note that Supplemental Proposition~\ref{sup-prop:fidelity-overlap} 
(and the first part of Supplemental Corollary~\ref{sup-coro:fidelity-swap}) is a full version of Proposition~\ref{prop-stat-kernel-overlap} (and Proposition~\ref{prop-stat-kernel-swap}) in the main text, which concern about the practical implications on estimated kernel values and the Gram matrix. Additionally, a full statement of Corollary~\ref{coro:opt-params} which considers the impact of kernel concentration on model predictions is presented in Supplemental Corollary~\ref{sup-cor:train-fqk-overlap} and the latter half of Supplemental Corollary~\ref{sup-coro:fidelity-swap}.

\subsubsection{Loschmidt Echo test}\label{appendix:fideliy-overlap-test}
For the Loschmidt Echo test, the quantum fidelity kernel is the probability of measuring the all-zero bitstring. That is, the observable is the global projector for the all-zero state $O = |\vec{0} \rangle\langle \vec{0} |$ and $\rho = U^{\dagger}(\vec{x'})U(\vec{x}) |\vec{0} \rangle\langle \vec{0} | U^{\dagger}(\vec{x})U(\vec{x'})$. 
The measurement outcome is $+1$ when the all-zero bitstring is observed and is $0$ for any other bitstrings. Thus, the statistical estimate is simply the ratio of the number of observed all-zero bitstrings to the total number of measurements. When the kernel value exponentially concentrates \sth{an exponentially small quantity}, the statistical estimate of the kernel is $0$ with a probability exponentially close to $1$. This is shown in the following proposition. 

\bigskip

\begin{supplemental_proposition}[A full version of Proposition~\ref{prop-stat-kernel-overlap}] \label{sup-prop:fidelity-overlap}
Consider the fidelity quantum kernel as defined in Eq.~\eqref{eq:fidelity-kernel-mt}. Assume that the kernel values $\kappa^{\rm FQ}(\vec{x},\vec{x'})$ exponentially concentrate towards some \sth{exponentially small value} $\mu$. 
Given that the Loschmidt Echo test is used to estimate the kernel value between an input data pair $\vec{x}$ and $\vec{x'}$ with a polynomial number of measurement shots \sth{$N$}, the probability that the statistical estimate of the kernel value $\widehat{\kappa}^{\rm FQ}(\vec{x},\vec{x'})$ \sth{is zero} is exponentially close to $1$. That is,
\begin{align}
    {\rm Pr}\left[ \widehat{\kappa}^{\rm FQ}(\vec{x},\vec{x'}) = 0 \right] \geq 1 - \delta \; \; , \; \delta \in \OC(c^{-n}) \; 
\end{align}
for some constant $c > 1$. In addition, for any training dataset $\SC = \{\vec{x}_i , y_i\}$ of size $N_s \in \OC(\poly(n))$, with a probability exponentially close to $1$, the statistical estimate of the Gram matrix $\widehat{K}$ is equal to the identity matrix. That is,
\begin{align}
    {\rm Pr}[ \widehat{K} =  \mathbb{1} ] \geq 1 - \delta' \; \; , \; \delta' \in \OC(c'^{-n}) \; 
\end{align}
for some constant $c' > 1$.
\end{supplemental_proposition}
\begin{proof} \sth{First, we recall that if the fidelity kernel concentrates to some exponentially small value over possible input data pairs as per Definition~\ref{def:exp-concentration}, we have
\begin{align}
    {\rm Pr}_{\vec{x},\vec{x'}}\left[ \left| \kappa^{\rm FQ}(\vec{x},\vec{x'}) - \mu\right| \geq \delta_c \right] \leq \frac{\beta}{\delta_c^2} \;,
\end{align}
such that
\begin{align}\label{eq:appx-mean-exp-small}
    \mu \in \OC(1/b'^n) \;,
\end{align}
for some $b'>1$, and
\begin{align}\label{eq:appx-var-exp-small}
    \beta \in \OC(1/b^n) \;,
\end{align}
for some $b>1$. By specifying $\delta_c = \beta^{1/4}$ and inverting the inequality, we have
\begin{align}\label{eq:proof-fidel-overlap101}
    {\rm Pr}_{\vec{x},\vec{x'}}\left[ \left| \kappa^{\rm FQ}(\vec{x},\vec{x'}) - \mu\right| \leq \beta^{1/4} \right] \geq 1- \sqrt{\beta} \;.
\end{align}
This implies that the probability of $\kappa^{\rm FQ}(\vec{x},\vec{x'})$ to be between $\mu - \beta^{1/4}$ and $\mu + \beta^{1/4}$ is at least $1-\sqrt{\beta}$.
}

We now show that for any given pair of $\vec{x}$ and $\vec{x'}$, it is exponentially likely that the statistical estimate of the kernel is zero. This is equivalent to proving that none of obtained bitstrings is all-zero bitstring. After $N$ measurements, the probability of this event happening can be expressed as
\begin{align}
    {\rm Pr}[\widehat{\kappa}^{\rm FQ}(\vec{x},\vec{x'}) = 0] = &\int_{0}^1 {\rm Pr}\left[\widehat{\kappa}^{\rm FQ}(\vec{x},\vec{x'}) = 0 \big|\kappa^{\rm FQ}(\vec{x},\vec{x'}) = s \right] {\rm Pr}\left[\kappa^{\rm FQ}(\vec{x},\vec{x'}) = s \right] ds \\
    = & \int_{0}^1 (1-s)^{N} {\rm Pr} \left[\kappa^{\rm FQ}(\vec{x},\vec{x'}) = s \right] ds \\
    \geq & \int_{\mu - \beta^{1/4}}^{\mu + \beta^{1/4}} (1-s)^{N} {\rm Pr} \left[\kappa^{\rm FQ}(\vec{x},\vec{x'}) = s \right] ds \\
    \geq & (1 - (\mu + \beta^{1/4}))^{N} \int_{\mu - \beta^{1/4}}^{\mu + \beta^{1/4}} {\rm Pr} \left[\kappa^{\rm FQ}(\vec{x},\vec{x'}) = s \right] ds \label{eq:appx-a-proof-fqk0} \\
    \geq & (1 - (\mu + \beta^{1/4}))^{N} (1 - \sqrt{\beta})  \label{eq:appx-a-proof-fqk} \\
    \geq &  (1 - N(\mu +  \beta^{1/4}))(1 - \sqrt{\beta}) \; , \label{eq:appx-a-proof-fqk2}
\end{align}
where in the first equality Bayes' theorem is used to introduce the conditional probability of measuring none all-zero bitstring for given $s= \kappa^{\rm FQ}(\vec{x},\vec{x'})$ and the marginal probability is acquired by integrating all possible values of $\kappa^{\rm FQ}(\vec{x},\vec{x'})$. The second equality is due to the fact that measurement outcomes are independent. In the first inequality, we limit the range of integration to $\mu \pm \beta^{1/4}$. The second inequality is from taking the minimum value of $(1-s)$ within the integration range. \sth{The next inequality is due to Eq.~\eqref{eq:proof-fidel-overlap101}.} To reach the last line in Eq.~\eqref{eq:appx-a-proof-fqk2}, we apply Bernoulli's inequality.

Now, it remains to show that the lower bound in Eq.~\eqref{eq:appx-a-proof-fqk2} is exponentially close to $1$. We recall that if the kernel values exponentially concentrate towards some exponentially small value, we have that $\mu$ and $\beta$ follow Eq.~\eqref{eq:appx-mean-exp-small} and Eq.~\eqref{eq:appx-var-exp-small}. Lastly, for a polynomial number of measurement shots, i.e., $N \in \OC(\poly(n))$, we have that $N(\mu + \beta^{1/4})$ vanishes exponentially, leading to
\begin{align}
    {\rm Pr}\left[ \widehat{\kappa}^{\rm FQ}(\vec{x},\vec{x'}) = 0 \right] \geq 1 - \delta \; \; , \; \delta \in \OC(c^{-n}) \label{eq:appx-a-proof-fqk22}\; ,
\end{align}
with $\delta = N(\mu + \beta^{1/4}) + \sqrt{\beta} - N(\mu + \beta^{1/4})  \sqrt{\beta} \in \OC(c^{-n})$ for some $c>1$. For large $n$, $\delta$ becomes exponentially smaller than $1$. This completes the first half of the proof. 

We now proceed to the second half of the proof. Consider a training dataset $\SC = \{\vec{x}_i, y_i \}$ with $N_s \in \OC(\poly(n))$. The event that the statistical estimate of the Gram matrix $\widehat{K}$ is equal to identity is equivalent to the event that all statistical estimates of kernel values for all pairs $\vec{x}_i$ and $\vec{x}_j$ (such that $i\neq j$) are zeros. Since each data point in the training dataset is drawn independently, estimating kernel values from different input data pairs in $\SC$ are independent events
\begin{align}
    {\rm Pr}\left[ \widehat{K} =  \mathbb{1} \right] & = {\rm Pr} \left[ \widehat{\kappa}^{\rm FQ}(\vec{x}_i, \vec{x}_j) = 0 \; ; \forall i, j \; , \;i \neq j  \right] \\
    & = \prod_{i < j} {\rm Pr}\left[ \widehat{\kappa}^{\rm FQ}(\vec{x}_i,\vec{x}_j) = 0 \right] \\
    & \geq (1 - \delta)^{N_s (N_s - 1)/2} \\
    & \geq 1 - N_s (N_s - 1)\delta/2 \;,
\end{align}
where the second equality uses the fact that the individual kernel values correspond to independent events (note that since the Gram matrix is symmetric we only have to estimate $N_s(N_s-1)/2$ kernel values). The first inequality is from applying the result in Eq.~\eqref{eq:appx-a-proof-fqk22} (as the kernel values are concentrated) and the last inequality is from Bernoulli's inequality. Since $N_s \in \OC(\poly(n))$, we have that $\delta' = N_s (N_s - 1)\delta/2 \in \OC(c^{-n})$ for some $c>1$.
\end{proof}

\medskip

Supplemental Proposition~\ref{sup-prop:fidelity-overlap} rigorously shows that the estimated Gram matrix is, \textit{for any choice in input data}, is likely to be an identity matrix. 
It follows that the trained model will also, with high probability, be independent of the training data, and thus in all likelihood not very useful. 
To demonstrate this, we consider the example of kernel ridge regression and show that the predictions of the output model concentrate at zero with high probability.

\begin{supplemental_corollary}[A full version of the Loschmidt Echo part of Corollary~\ref{coro:opt-params}]\label{sup-cor:train-fqk-overlap}
Consider implementing kernel ridge regression with the fidelity quantum kernel, a squared loss function and a training dataset $\SC = \{\vec{x}, y_i \}_i$ with $N_s \in \OC(\poly(n))$. Denote an input data independent fixed point $\vec{a}_0(\vec{y},\lambda) =  \vec{y}/(1 - \lambda)$ where $\vec{y}$ is a vector of output data points with its $i^{\rm th}$ element equal to $y_i$ and $\lambda$ is regularization in the loss function. Under the same assumptions as in Supplemental Proposition~\ref{sup-prop:fidelity-overlap}, if the model is trained with a statistical estimate of the Gram matrix $\widehat{K}$ where each element is evaluated with a polynomial number of measurement shots $N \in \OC(\poly(n))$ then, with probability exponentially close to $1$, the optimal parameters $\vec{a}_{\rm opt}$ are identical to the fixed point $\vec{a}_0(\vec{y},\lambda)$. That is, we have
\begin{align}
    {\rm Pr}\left[\vec{a}_{\rm opt} =\vec{a}_0(\vec{y},\lambda) \right] \geq 1 - \delta \; \; , \; \delta \in \OC(c^{-n}) \;,
\end{align}
for some constant $c > 1$.
In addition, the \sth{estimated} model prediction on an unseen input $\widehat{f}(\vec{x}) = \sum_{i} a^{(i)}_{\rm opt} \widehat{\kappa}^{\rm FQ}(\vec{x}_i, \vec{x})$ will, with a probability exponentially close to $1$, be $0$,
\begin{align}
    {\rm Pr}\left[ \widehat{f}(\vec{x}) = 0 | \vec{x} \notin \SC \right] \geq 1 - \delta' \; \; , \; \delta' \in \OC(c'^{-n}) \;,
\end{align}
for some constant $c' >1$.
\end{supplemental_corollary}
\begin{proof} According to Supplemental Proposition~\ref{sup-prop:fidelity-overlap}, we obtain the statistical estimate of the Gram matrix to be an identity $\widehat{K} = \mathbb{1}$ with probability exponentially close to $1$. The optimal parameters for a kernel ridge regression with a squared loss function are given by
\begin{align}\label{eq:optparams}
    \vec{a}_{\rm opt} = & (\widehat{K} \sth{-} \lambda \mathbb{1})^{-1} \vec{y} \\
    = & \frac{\vec{y}}{1 \sth{-} \lambda} \; .
\end{align}
For Supplemental Proposition~\ref{sup-prop:fidelity-overlap}, this is obtained with probability at least $1-\delta$ with $\delta \in \OC(\poly(n))$. This proves the first part of the corollary. 

Secondly, it follows from the Representer Theorem that the model prediction is of the form
\begin{align}
    f(\vec{x}) = \sum_{i=1}^{N_s} a_{\rm opt}^{(i)} \kappa^{\rm FQ}(\vec{x}_i,\vec{x}) \;.
\end{align}
Computing the model prediction for an unseen input data $\vec{x} \notin \SC$ requires computing a statistical estimate of the kernel values between the new data point and the training data points. When computed with a polynomial number of shots these estimates $\widehat{\kappa}^{\rm FQ}(\vec{x_i},\vec{x})$ will be $0$ with high probability. Specifically, we can bound this probability as 
\begin{align}
    {\rm Pr}\left[ \widehat{f}(\vec{x}) = 0 | \vec{x} \notin \SC \right] = & {\rm Pr}\left[\widehat{\kappa}^{\rm FQ}(\vec{x},\vec{x}_i) = 0 ; \forall i \right] \\
    = & \prod_{i = 1}^{N_s} {\rm Pr}\left[\widehat{\kappa}^{\rm FQ}(\vec{x},\vec{x}_i) = 0 \right] \\ 
    \geq & (1 - \delta)^{N_s} \\
    \geq & 1 - N_s \delta \; ,
\end{align}
where the second equality is due to the statistical independence of the kernel value estimates, the first inequality is from applying Supplemental Proposition~\ref{sup-prop:fidelity-overlap} and the final inequality is from Bernoulli's inequality. Since $N_s \in \OC(\poly(n))$, we have that $\delta' = N_s \delta \in \OC(c^{-n})$ for some $c>1$.
\end{proof}

Importantly, since all information concerning the training output data is in effect hard-coded in the formula for the optimal parameters, Eq.~\eqref{eq:optparams}, a low training error can be obtained. On the other hand, the model prediction is entirely insensitive to the input data and hence the model generalizes poorly. As supported by the numerics in the main text (see Fig.~\ref{fig:effect-exp-con-gen}), this poor generalization has a different flavor to the type of generalization usually quantified by generalization bounds in that it cannot be resolved simply by training on more data points. Instead, one must supply at least an exponential number of shots for hope of good generalization.

\sth{\subsubsection{SWAP test}\label{appendix:fidelity-swap}

For the SWAP test, the kernel value is the expectation value of the Pauli Z operator on an ancilla qubit. For each measurement, the outcome is drawn from a certain distribution $\PC_{\kappa^{\rm FQ}(\vec{x},\vec{x'})}$ that encodes the kernel information. More precisely, the outcome is either $+1$ with probability $p_+ = 1/2 + \kappa^{\rm FQ}(\vec{x},\vec{x'})/2$ or $-1$ with probability $p_- = 1 - p_+$ i.e.,
\begin{align} \label{eq:dist-kernel-swap}
    \PC_{\kappa^{\rm FQ}(\vec{x},\vec{x'})} := \left\{ \frac{1+ \kappa^{\rm FQ}(\vec{x},\vec{x'})}{2}, \frac{1- \kappa^{\rm FQ}(\vec{x},\vec{x'})}{2} \right\} \;.
\end{align}
In other words, the kernel value can be thought of as encoding a perturbation to the uniform distribution
\begin{align}\label{eq:dist-no-input}
    \PC_0 := \left\{\frac{1}{2}, \frac{1}{2}\right\} \;.
\end{align}

\begin{figure}[t]
\includegraphics[width=.75\columnwidth]{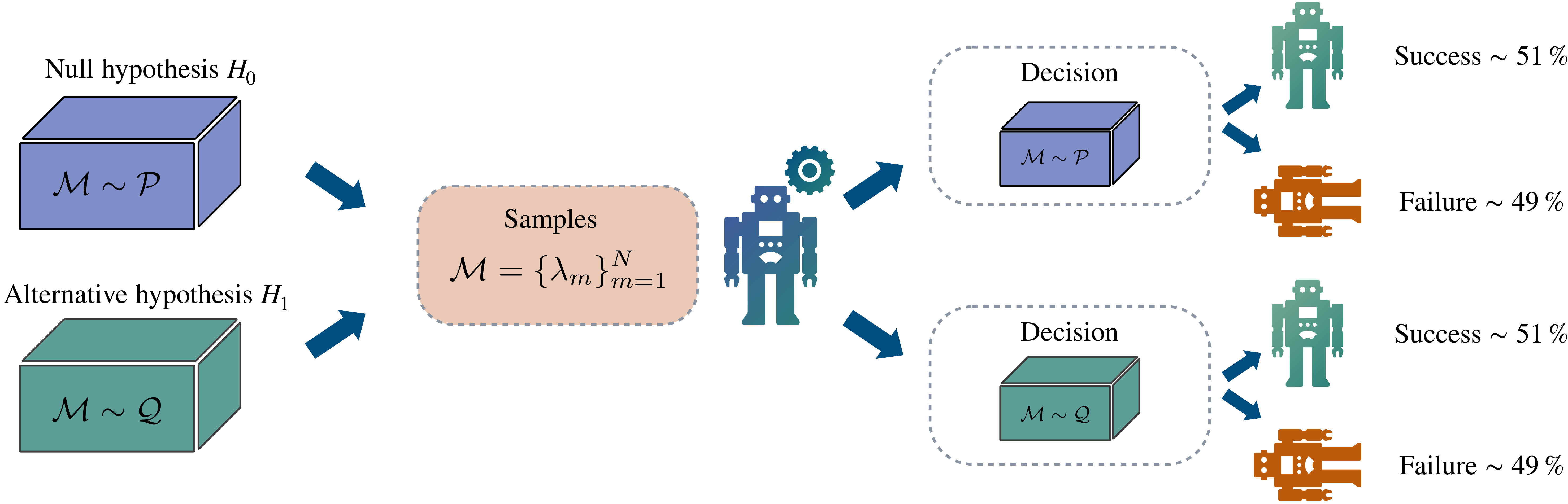}
\caption{\sth{\textbf{Statistical indistinguishability.} Suppose we are given a set of $N$ samples $\MC$ that are either drawn from $\PC$ (Null hypothesis) or from $\QC$ (Alternative hypothesis), with the two possibilities equally probable. Two distributions $\PC$ and $\QC$ are said to be statistically indistinguishable with $N$ samples if there exists no algorithm to reliably pass this binary hypothesis test.}}\label{fig-sup:stat-indis}
\end{figure}

We will consider the following notion of statistical indistinguishability (also illustrated in Fig.~\ref{fig-sup:stat-indis}).
\begin{definition}\label{def:StatIndist} [Statistical indistinguishability (of distributions)]
    Two probability distributions $\PC$ and $\QC$ are statistically indistinguishable with $N$ samples if a binary hypothesis test cannot be passed with  probability at least $0.51$. That is, given a set of $N$ samples $\MC$ drawn from either $\PC$ or $\QC$ (with an equal probability), consider the following hypotheses
    \begin{itemize}
        \item Null hypothesis $\HC_0$: $\MC$ is drawn from $\PC$\,,
        \item Alternative hypothesis $\HC_1$: $\MC$ is drawn from $\QC$ \,,
    \end{itemize}
    where $\PC$ and $\QC$ are statistically indistinguishable (with $N$ samples) if for any algorithm the probability of correctly identifying the correct hypothesis, ${\rm Pr}[``{\rm right \; decision \; between \, } \HC_0 \, {\rm and} \, \HC_1"]$, satisfies:
    \begin{align}
         {\rm Pr}[``{\rm right \; decision \; between \, } \HC_0 \, {\rm and} \, \HC_1"] \leq 0.51 \;.
    \end{align}
\end{definition}
Note that the threshold $0.51$ in the definition is arbitrary chosen to be close to that of random guessing. We refer the reader to Appendix~\ref{appendix:stat-indis-basic} for a recap of basic results on hypothesis testing that we will use in the following.

We now recall that the exact kernel value between two given input vectors $\vec{x}$ and $\vec{x'}$ is fixed and not random. However, when $\kappa^{\rm FQ}(\vec{x},\vec{x'})$ concentrates towards an exponentially small value $\mu$, over the uniform distribution of input data pairs this kernel value is exponentially likely to be close to $\mu$. That is, the exact kernel value is exponentially likely to be exponentially small. 

In practice (for moderate to large-scale problems, i.e., for the problems we are ultimately interested in using quantum kernels for) we are limited to $N$ samples where $N$ scales polynomially with problem size. In this case, the distribution associated with the true kernel value, $\PC_{\kappa^{\rm FQ}(\vec{x},\vec{x'})}$,  and the uniform binary distribution, $\PC_0$, are statistically indistinguishable. This argument is illustrated in Fig.~\ref{fig-sup:pipeline-stat} and formalized in the following Supplemental Lemma.

\begin{lemma}\label{sup-lem:fidelity-swap}
Suppose that the fidelity quantum kernel $\kappa^{\rm FQ}(\vec{x},\vec{x'})$ is exponentially concentrated over input data $\vec{x}$ and $\vec{x'}$ to some exponentially small value $\mu$ according to Definition~\ref{def:exp-concentration}. For any given $\vec{x}$ and $\vec{x'}$, 
we consider measuring  $\kappa^{\rm FQ}(\vec{x},\vec{x'})$ using a SWAP test, that is, samples are drawn from the distribution $\PC_{\kappa^{\rm FQ}(\vec{x},\vec{x'})}$ in Eq.~\eqref{eq:dist-kernel-swap}. 
Let $\MC_s$ denote a set of $N \in \text{poly}(n)$ samples drawn either from $\PC_{\kappa^{\rm FQ}(\vec{x},\vec{x'})}$ or $\PC_0$ (with equal probability). 
We then perform a hypothesis test with:
    \begin{itemize}
        \item Null hypothesis $\HC_0$: $\MC_s$ is drawn from the uniform distribution $\PC_0$\,,
        \item Alternative hypothesis $\HC_s$: $\MC_s$ is drawn from $\PC_{\kappa^{\rm FQ}(\vec{x},\vec{x'})}$ \,.
    \end{itemize}
With probability at least $1 - \delta_\kappa$ over the pairs of input data $\vec{x}$ and $\vec{x'}$, 
we have that 
\begin{align}\label{eq:swaptestsuccessv2}
 {\rm Pr}\left({\rm ``right \, decision \, between \,} \HC_{s} \, {\rm and} \, \HC_0" \right)  \leq \frac{1}{2} + \epsilon \, ,
\end{align}
with $\delta_\kappa \in \OC(c^{-n})$ for some $c>1$ and $\; \epsilon  \in \OC(c'^{-n})$ for some $c' > 1$. That is, with exponentially high probability over input data pairs $(\vec{x},\vec{x'})$, the distributions $\PC_{0}$ and $\PC_{\kappa^{\rm FQ}(\vec{x},\vec{x'})}$ are statistically indistinguishable for large problem sizes with a polynomial number of samples $N$ (as per Definition~\ref{def:StatIndist}). 
\end{lemma}

\begin{figure}[t]
\includegraphics[width=1.0\columnwidth]{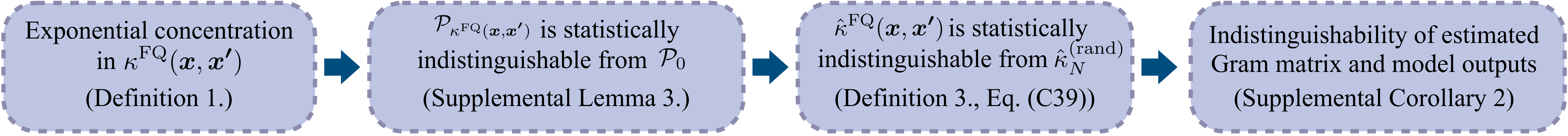}
\caption{\sth{\textbf{Summary of the impact of kernel concentration on the model outputs.} In the presence of kernel concentration, for any given input pair, its kernel value is highly likely to be exponentially close to a exponentially small value. This leads to the indistinguishability (with polynomial samples) between a distribution associated with a kernel $\PC_{\kappa^{\rm FQ}(\vec{x},\vec{x'})}$ and a data-independent uniform distribution $\PC_0$. Since the distributions themselves are indistinguishable, an estimate of kernel value (which is an empirical mean over samples) is also indistinguishable from an empirical mean over samples drawn from $\PC_0$. Ultimately, a model trained on these kernel estimates behaves indistinguishably from an data-independent model.}
}\label{fig-sup:pipeline-stat} 
\end{figure}
\begin{proof}
Our proof strategy is to show that, due to exponential concentration, the exact kernel value is very likely (i.e., with probability exponentially close to 1) to be exponentially small for any given input data pair $\vec{x}$ and $\vec{x'}$. This exponentially small kernel value corresponds to an exponentially small perturbation from the uniform distribution. Combining this observation with Supplemental Proposition~\ref{sup-prop:indistin-prob} we can then establish that it is hard to decide the correct hypothesis with polynomial resources. 

More explicitly, we first note that it follows from Supplemental Proposition~\ref{sup-prop:indistin-prob} that 
\begin{align}
 {\rm Pr}\left({\rm ``right \, decision \, between \,} \HC_{s} \, \text{and} \, \HC_0" \big| \kappa^{\rm FQ}(\vec{x},\vec{x'})= s \right)
    \leq &  \left( \frac{1}{2} + \frac{Ns}{4} \right) \, .
\end{align}
For $s \in \OC(c'^{-n})$ for some $c' > 1$ and $N \in \OC(\poly(n))$ we have $\epsilon  \in \OC(c'^{-n})$  as claimed. It remains to determine with what probability we have $s \in \OC(c'^{-n})$. 
By the assumption that the fidelity kernel concentrates to an exponentially small value over possible input data pairs (as per Definition~\ref{def:exp-concentration}), we have
\begin{align}\label{eq:appx-chebyshev}
    {\rm Pr}_{\vec{x},\vec{x'}}\left[ \left| \kappa^{\rm FQ}(\vec{x},\vec{x'}) - \mu\right| \geq \delta_c \right] \leq \frac{\beta}{\delta_c^2} \;,
\end{align}
with
\begin{align}
    \beta \in \OC(1/b^n) \;\;,\;\; \mu \in \OC(1/b'^n) \;,
\end{align}
for some $b,b'>1$. We then choose $\delta_c = \beta^{1/4}$ and invert the inequality of Eq.~\eqref{eq:appx-chebyshev}, leading to
\begin{align}
    {\rm Pr}_{\vec{x},\vec{x'}}\left[\left|\kappa^{\rm FQ}(\vec{x},\vec{x'}) - \mu \right| \leq \beta^{1/4} \right] \geq 1 - \sqrt{\beta} \; .
\end{align}
It follows that $\kappa^{\rm FQ}(\vec{x},\vec{x'})$ takes value between $\mu - \beta^{1/4}$ and $\mu + \beta^{1/4}$ (which are exponentially small) with probability at least $1 - \sqrt{\beta}$ (which is exponentially close to $1$). Recalling the form  $\PC_{\kappa^{\rm FQ}(\vec{x},\vec{x'})}$ and $\PC_{0}$ take in Eq.~\eqref{eq:dist-kernel-swap} and Eq.~\eqref{eq:dist-no-input}, we see that $\PC_{\kappa^{\rm FQ}(\vec{x},\vec{x'})}$ is an exponentially small perturbation of $\PC_{0}$ and the result follows by invoking Supplemental Proposition \ref{sup-prop:indistin-prob}.

\end{proof}

The central thesis of this section is that when fidelity kernels satisfy the conditions specified in Supplemental Lemma \ref{sup-lem:fidelity-swap}, they lead to useless models that do not generalize well. The argument is structured as follows. Due to the Representer Theorem, model outputs on unseen data are the output of some linear map on the statistical estimates obtained from experimental samples (which can be  thought of as some post-processing). Thus, if we were able to take the model outputs and distinguish them from the outputs constructed from an (essentially useless) model based on the uniform distribution, then we would succeed in the hypothesis test specified in Supplemental Lemma \ref{sup-lem:fidelity-swap}.  Hence, by contradiction, it must not be possible to distinguish the model outputs constructed from such fidelity kernel values from those outputted from a model based on the uniform distribution. These models constructed from such fidelity kernels is then clearly useless.
The last part of the argument is to observe that in an experimental setting, one has a strictly weaker setting than in the  hypothesis test inr Supplemental Lemma \ref{sup-lem:fidelity-swap}, as one does not have access to the exact kernel values. Thus, the conclusion follows by reduction.

\medskip

The above paragraph intuitively summarizes the consequences of exponential concentration on kernel-based quantum models. In what follows we present this argument in more detail. We start by defining a notion of indistinguishability for empirical outcomes sampled from distributions.

\begin{definition}[Statistical indistinguishability (of outputs)]\label{def:StatIndistOutputs}
    Consider a map $\Phi:\mathbb{R}^N\rightarrow \mathbb{R}^M$ \sth{(with $M$ being the dimension of the output)} and two distributions $\PC$ and $\QC$ which are statistically indistinguishable under $N$ samples according to Definition~\ref{def:StatIndist}. Draw $N$ respective samples from $\PC$ and $\QC$, which we respectively denote as $\MC_{\PC}$ and $\MC_{\QC}$.  We say that $\Phi(\mathcal{M}_{\mathcal{P}})$ and $\Phi(\mathcal{M}_{\mathcal{Q}})$ are statistically indistinguishable outputs.
\end{definition}

We introduce Definition~\ref{def:StatIndistOutputs} to describe the outputs and subsequent processing of samples drawn from indistinguishable distributions. Specifically, any distribution which satisfies Definition~\ref{def:StatIndist} automatically has outputs which satisfy Definition~\ref{def:StatIndistOutputs}. In addition, as $\Phi(\mathcal{M}_{\mathcal{P}})$ and $\Phi(\mathcal{M}_{\mathcal{Q}})$ are constructed from samples of distributions, they themselves are random variables. 

As an example, we would say that an experimentally obtained kernel value (an empirical mean) between any given input data $\vec{x}$ and $\vec{x'}$ estimated with a polynomial number of measurement outcomes/samples is statistically indistinguishable (with probability exponentially close to 1) from the empirical mean of the samples from the uniform distribution
\begin{align}\label{eq:k0-no-input-appx}
    \widehat{\kappa}^{(\rm rand)}_N = \frac{1}{N} \sum_{m = 1}^N \Tilde{\lambda}_m \;, 
\end{align}
where each $\Tilde{\lambda}_m$ equally likely takes value $+1$ or $-1$. 

\medskip

Given a training dataset $\SC$ of polynomial size $N_s$, consider the set of kernel values over possible pairs in $\SC$ (excluding the trivial ones where $\kappa^{\rm FQ}(\vec{x},\vec{x}) = 1 $) 
\begin{align}\label{eq:appx-kernel-set}
    \KC = \left\{ \kappa^{\rm FQ}(\vec{x}, \vec{x'}) \; | \; \forall \{\vec{x}, \vec{x'}\} \subseteq \SC \; ; \; \vec{x} \neq \vec{x'} \right\} \;.
\end{align}
Due to exponential concentration, each kernel value in this set is highly likely to be exponentially small and so Supplementary Lemma~\ref{sup-lem:fidelity-swap} will apply to each of these kernel values. 

It then follows that any model computed by post-processing these samples is also, with high probability, statistically indistinguishable (as per Definition~\ref{def:StatIndistOutputs}) from the model produced from the uniform binary distribution for each kernel entry. That is, the model predictions are independent of the input data and for all intents and purposes useless. This holds for any kernel method including both supervised and unsupervised learning tasks. For concreteness, let us again consider kernel ridge regression.

\begin{supplemental_corollary}\label{sup-coro:fidelity-swap}[Full version of Proposition~\ref{prop-stat-kernel-swap} and the SWAP part of Corollary~\ref{coro:opt-params}]
Consider a kernel ridge regression task with the fidelity quantum kernel, a squared loss function and a training dataset $\SC = \{\vec{x}_i, y_i \}_i$ of  size  $N_s \in \OC(\poly(n))$. 
Given that the kernel value is estimated using the SWAP test and under the same assumptions as in Supplemental Lemma~\ref{sup-lem:fidelity-swap},  the following statements hold with probability exponentially close to $1$ (i.e., with the probability $1 - \delta_a$ with $\delta_a \in \OC(\tilde{c}^{-n})$ for $\tilde{c}>1$)
\begin{itemize}
    \item The estimated Gram matrix $\hat{K}$ is statistically indistinguishable (Def.~\ref{def:StatIndistOutputs}) from an input-data-independent random matrix $\widehat{K}_N^{\rm (rand)}$  whose diagonal elements are $1$ and the off-diagonal elements are instances of $\widehat{\kappa}_N^{\rm (rand)}$ in Eq.~\eqref{eq:k0-no-input-appx}.
    \item The estimated optimal parameters are statistically indistinguishable (Def.~\ref{def:StatIndistOutputs}) from the input-data-independent random variables
\begin{align}
    \vec{a}_{\rm rand}(\vec{y}, \lambda) = \left( \widehat{K}_N^{\rm (rand)} - \lambda \mathbb{1} \right)^{-1}\vec{y}\,, 
\end{align}
where $ \widehat{K}_N^{\rm (rand)}$ is a random matrix whose diagonal elements are $1$ and off-diagonal elements are instances of $\widehat{\kappa}_{N}^{\rm (rand)}$ in Eq.~\eqref{eq:k0-no-input-appx}, $\vec{y}$ is a vector of output data points with its $i^{\rm th}$ element equal to $y_i$ and $\lambda$ is the regularization parameter. 
\item The model prediction on an unseen input data $\vec{x}$ is  statistically indistinguishable (Def.~\ref{def:StatIndistOutputs}) from the input-data-independent random variable
\begin{align}
    h_{\rm rand} = \vec{a}_{\rm rand}(\vec{y}, \lambda) ^T\vec{k}_N^{(\rm rand)} \;,
\end{align}
where $\vec{k}_N^{(\rm rand)}$ is a random vector where each of its element is an instance of $\widehat{\kappa}_{N}^{\rm (rand)}$.
\end{itemize}
\end{supplemental_corollary}
\begin{proof}
We use Supplementary Lemma~\ref{sup-lem:fidelity-swap} with a union bound over the individual kernel values. Since $N_s \in \OC(\poly(n))$, there are a polynomial number of kernel values to be estimated and each of the estimated kernel value is, with exponentially high probability, statistically indistinguishable to an instance of $\widehat{\kappa}^{(\rm rand)}_N$.

More concretely, for the Gram matrix, we are required to estimate each kernel value in the kernel set $\KC$ in Eq.~\eqref{eq:appx-kernel-set} i.e., the off-diagonal elements. This amounts to $N_s(N_s-1)/2$ unique kernel values. Denote $\kappa_i$ as an $i^{\rm th}$ element in $\KC$ with $i$ running from $1$ to $N_s(N_s-1)/2$. Let $E_i$ be the event that the estimate of $\kappa_i$ is statistically indistinguishable from $\widehat{\kappa}^{\rm rand}_N$. From Supplementary Lemma~\ref{sup-lem:fidelity-swap}, we have
\begin{align}\label{eq:appx-proof-coro-1}
    {\rm Pr}[E_i] \geq 1 - \delta_\kappa \;\;, \forall \kappa_i \in \KC \;,
\end{align}
with $\delta_\kappa \in \OC(c^{-n})$ for $c > 1$. Now, the probability that all $E_i$ occur can be bounded as
\begin{align}
    {\rm Pr}\left[ \bigcap_{i} E_i \right] & =  1 - {\rm Pr}\left[ \bigcup_i \bar{E}_i \right]  \\
    & \geq  1 - \sum_{i=1}^{|\KC|} {\rm Pr}\left[\bar{E}_i \right] \\
    & \geq 1 - \frac{N_s (N_s-1)\delta_k}{2} \;,
\end{align}
where $\bar{E}_i$ is a conjugate event of $E_i$, we use the union bound in the second line and use ${\rm Pr}[\bar{E}_i] \leq \delta_k$ by reversing the final inequality in Eq.~\eqref{eq:appx-proof-coro-1}. Since $N_s \in \OC(\poly(n))$, we have that the probability that each of the kernel values are statistically indistinguishable (and so the Gram matrix is statistically indistinguishable) is $1 - \delta_K$ with $\delta_K := N_s (N_s-1)\delta_k/2 \in \OC(\tilde{c}^{-n})$ for some $\tilde{c} >1$.

The statistical indistinguishability of the optimal parameters directly follows from the above result. This is because estimating the optimal parameters is simply a post processing of the Gram matrix.

Lastly, to show indistinguishability of the model predictions, we have to take into account the kernel values for the test input. This requires estimating an additional $N_s$ kernel values. On repeating the same argument using the union bound, it follows that the probability that all estimated kernel values (from the Gram matrix and new ones) are indistinguishable from $\widehat{\kappa}^{(\rm rand)}_N$ is exponentially close to $1$.
\end{proof}

In general a model that generalizes well must produce outputs that are data-dependent.  Thus, these outputs must at minimum be distinguishable from a data-independent distribution. Hence, the models trained from these estimated Gram matrix have poor generalization. Lastly, similar to the Loschmidt Echo test, the training error can remain low as the correct output labels are effectively cooked into the model for the input data. Thus the model can "train" well. However, the trained model is insensitive to input data and thus poorly generalizes.}

\begin{figure}[t]
\includegraphics[width=0.6\columnwidth]{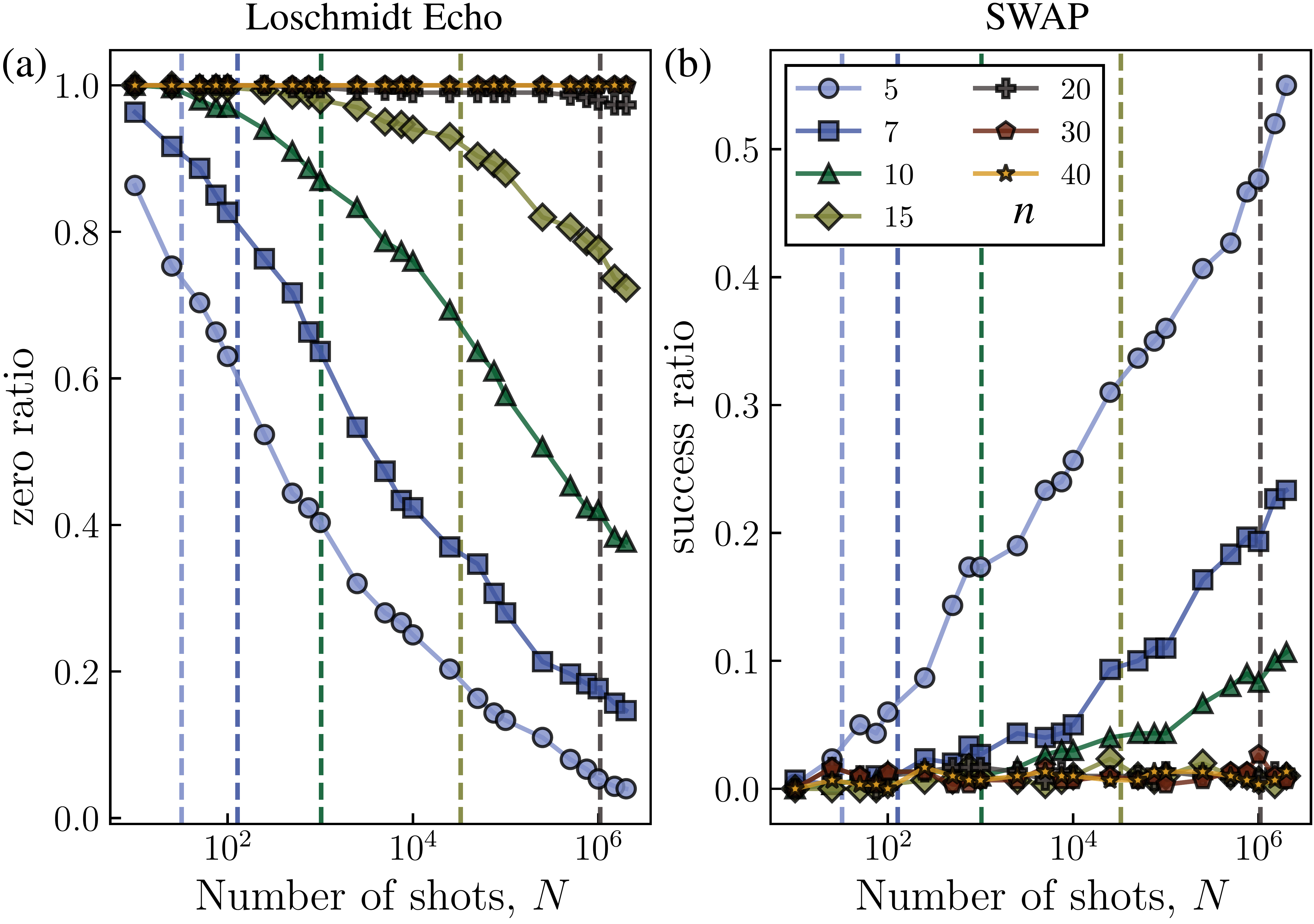}
\caption{\textbf{Effect of exponential concentration on estimated Gram matrix.} In panel (a), where the kernel values are estimated with the Loschmidt Echo test, we plot the zero ratio (i.e. the number of estimates that are zero compared to the total number of kernel values) \sth{with respect to the number of measurement shots used $N$, for different total number of qubits $n$}. In panel (b), where the SWAP test is employed, we plot the success ratio (which is the number of estimates that pass the binomial test from the uniform distribution with p-value$=0.01$ compared to the total number of kernel values). The x-axis indicates the number of shots used per kernel value and vertical lines indicate the dimension of the (exponentially increasing) Hilbert space $2^n$. Here the training data size is $N_s = 25$.
}\label{fig:effect-exp-con-gram}
\end{figure}

\subsubsection{Numerical simulation}\label{appendix:fidelity-numerics}

Fig.~\ref{fig:effect-exp-con-gram} demonstrates how kernel concentration affects the statistical estimates of kernel values via the Loschmidt Echo test (in panel (a)) and the SWAP test (in panel (b)) via a numerical example. We consider a training set where each input data point is a $n-$dimensional vector with each element uniformly drawn from $[0, 2\pi]$, and is encoded via a tensor product embedding which consists of a layer of single-qubit $R_y$ rotation gates. In this setting, kernel values exponentially concentrate to \sth{an exponentially small value $\mu$} (see Sec~\ref{sec:global} in the main text). Each unique off-diagonal element in the Gram matrix is evaluated with an increasing number of measurement shots (as indicated in the x-axis of both panels). Note that in our analysis we fix diagonal elements to be $1$ without evaluation as one may do in realistic setting.

In panel (a) where the Loschmidt Echo test is employed, we plot the ratio of the number of statistical estimates that are zero to the total number of kernel values as a function of qubits and measurement shots. When the ratio is $1$, this indicates that all kernel estimates are zero. In general, this fraction becomes smaller with increasing measurement shots. We also observe that in order to achieve a fixed ratio of non-zero values (e.g. $\sim 0.75$), exponentially many measurement shots are required i.e., $N\in\Omega(2^n)$. 
Particularly, at $30$ and $40$ qubits, all of the estimates are zero even with $2\times 10^6$ shots per kernel value.

In panel (b) where the kernel values are estimated with SWAP tests, for each individual kernel value, we perform a binomial hypothesis test on the measurement outcomes to see whether or not there is sufficient statistical significance to distinguish the shots from those obtained from the uniform distribution. 
We plot the ratio of the estimates that pass the binomial test (with p-value below $0.01$) as a function of qubits and measurement shots. A low ratio indicates that most of the estimates are statistically indistinguishable from the ones estimated with the data-independent uniform distribution. It can be seen from the panel that to maintain a constant success ratio the number of measurement shots needs to scale at least exponentially with the number of qubits.

\subsection{Projected quantum kernel}\label{appendix:projected}
As discussed in our main text, the projected quantum kernel is an alternative approach to comparing data-encoded quantum states. It takes the form
\begin{align}
    \kappa^{PQ}(\vec{x},\vec{x'}) = {\rm exp}\left( - \gamma \sum_{k=1}^n \| \rho_k(\vec{x}) - \rho_k(\vec{x'})\|^2_2\right) \;, \label{eq:projected-gaussian-kernel-appx}
\end{align}
where $\rho_k(\vec{x})$ is the reduced state of $\rho(\vec{x})$ on the $k$-th qubit, $\|\cdot\|_2$ is the Schatten 2-norm and $\gamma$ is a positive hyperparameter. 

Estimating the projected quantum kernel in practice requires us to first obtain statistical estimates of the 2-norms on all individual qubits and then classically post-process them to estimate the kernel value. Here we consider two common strategies to estimate the 2-norms. 

\medskip

\noindent\underline{Tomography strategy:} First, we perform full state tomography on the reduced density matrices. 
As these are single-qubit states the number of required measurements is constant with respect to the number of qubits. In particular, the reduced state to the $k^{\rm th}$ qubit can be expressed in the Pauli basis as
\begin{align} \label{eq:reduced-state-pauli-basis}
    \rho_{k}(\vec{x}) = \frac{1}{2} \left(\mathbb{1}_k + c_{x_k}(\vec{x}) X_k + c_{y_k}(\vec{x}) Y_k + c_{z_k}(\vec{x}) Z_k \right) \;,
\end{align}
where $\{X_k, Y_k, Z_k\}$ are single X, Y and Z Pauli matrices on the qubit $k$ with corresponding coefficients $\{ c_{x_k}(\vec{x}), c_{y_k}(\vec{x}), c_{z_k}(\vec{x})\}$. Each coefficient is simply the expectation value with the respective Pauli observable 
\begin{align}\label{eq:appx-coeff-pqk}
    c_{\sigma_k}(\vec{x}) = \Tr[ \rho_{k}(\vec{x}) \sigma_k ] \;,
\end{align}
with $\sigma_k \in \{X_k,Y_k,Z_k\}$. \sth{To estimate each of the $3n$ coefficients, we can make local measurements in each respective basis. The measurement outcome is either $+1$ with probability $p_+ = 1/2 + c_{\sigma_k}(\vec{x})/2$ and $-1$ with probability $1 - p_+$. That is, we have the distribution
\begin{align}\label{eq:appx-coeff-dist-pqk}
    \PC_{{\sigma_k}, \vec{x}} = \left\{\frac{1 + c_{\sigma_k}(\vec{x})}{2}, \frac{1 - c_{\sigma_k}(\vec{x})}{2} \right\} \;.
\end{align}
After some specified $N$ measurement samples, the statistical estimate of the expectation value is obtained via taking their empirical mean in the usual way. After an estimate of each reduced density matrix  is obtained for all qubits and data values, an estimate of the kernel values in Eq.~\eqref{eq:projected-gaussian-kernel-appx} can be evaluated via matrix algebra.} 

\medskip

\noindent\underline{Local SWAP strategy:} Alternatively, we can employ local SWAP tests to evaluate the 2-norms. In particular, by explicitly expanding the 2-norm, we have
\begin{align}
    \| \rho_k(\boldsymbol{x}) - \rho_k(\boldsymbol{x'})\|_2^2 = \Tr[\rho_k^2(\boldsymbol{x})] + \Tr[\rho_k^2(\boldsymbol{x'})] - 2 \Tr[\rho_k(\boldsymbol{x})\rho_k(\boldsymbol{x'})] \label{eq:appx-2norm-expand}\;.
\end{align}
That is, the 2-norm distance contains the overlap between two reduced states $\Tr[\rho_k(\vec{x}) \rho_k(\vec{x'})]$ and the purity of each individual reduced state $\Tr[\rho_k^2(\vec{x})]$, $\Tr[\rho_k^2(\vec{x'})]$. 
Each term in Eq.~\eqref{eq:appx-2norm-expand} can be estimated using the local SWAP test. 
Similar to the fidelity case previously, each term in the 2-norm is equal to the expectation value of the Pauli-Z operator on an ancilla qubit where each measurement gives either $+1$ or $-1$. \sth{More precisely, denote $m_k(\vec{x}, \vec{x'}) = \Tr[\rho_k(\boldsymbol{x})\rho_k(\boldsymbol{x'})] $. When making measurements, an individual outcome takes $+1$ with probability $p_+ = 1/2 + m_k(\vec{x},\vec{x'})/2$ and $-1$ with $p_+ = 1 - p_-$ i.e. we sample from the distribution
\begin{align}
    \PC_{m_k(\vec{x},\vec{x'})} = \left\{ \frac{1 + m_k(\vec{x},\vec{x'})}{2},\frac{1 - m_k(\vec{x},\vec{x'})}{2}  \right\} \;.
\end{align}
Then, the statistical estimate of $m_k(\vec{x},\vec{x'})$ can be obtained as an empirical mean of these outcomes.}

\subsubsection{Consequence of exponential concentration}\label{appendix:projected1}

To see how concentration affects the projected quantum kernel in practice, we take as our starting point the assumption that 
\begin{align}\label{eq:2norm-vanish}
    \Ebb_{\vec{x} \in \XC} \left\| \rho_{k}(\vec{x}) - \frac{\mathbb{1}_k}{2} \right\|_2 \leq \beta \in \OC\left(\frac{1}{b^n} \right) \;,
\end{align}
for all $k \in \{1,...,n\}$, where the expectation value is taken over some chosen distribution over $\XC$. In the later sections, we show that all sources of the exponential concentration in the projected quantum kernel also lead to this exponential vanishing of the 2-norm distance between the reduced quantum states. The following lemma shows the connection between this reduced state concentration and kernel concentration. 
\begin{lemma}
Given that the Eq.~\eqref{eq:2norm-vanish} is satisfied, it follows that the projected quantum kernel exponentially concentrates 
\begin{align}\label{eq:kernelconc}
    \Pr_{\vec{x},\vec{x'} \in \XC}\left[  \left| \kappa^{PQ}(\vec{x},\vec{x'}) - \mu \right|  \geq \delta \right] \leq \frac{\beta}{\delta^2} \;,
\end{align}
where $\beta \in \OC(1/b^n)$ for some $b>1$. 
\end{lemma}
We defer the proof to Appendix~\ref{appendix:pqk-proof}.

\medskip

We now show that if the distance of the reduced state from the maximally mixed state is exponentially small, i.e. Eq.~\eqref{eq:2norm-vanish} holds, then for both tomography and local SWAP strategies, the (data-dependent) distribution associated with each quantity of interest is statistically indistinguishable from a fixed distribution, as per Definition~\ref{def:StatIndist}.

\sth{\begin{lemma}\label{sup-lem:pqk-tomo-swap}
Assume the 2-norm distance between the reduced data encoding states exponentially vanishes as in Eq.~\eqref{eq:2norm-vanish}. 
Consider the following two scenarios
\begin{enumerate}
    \item For the tomography strategy, suppose we measure any coefficient $c_{\sigma_k(\vec{x})}$ of a reduced state on the qubit $k$ in Eq.~\eqref{eq:appx-coeff-pqk} for a given input data $\vec{x}$ with a polynomial number of measurement shots. The associated distribution $\PC_{\sigma_k, \vec{x}}$ defined in Eq.~\eqref{eq:appx-coeff-dist-pqk} is statistically indistinguishable (as per Definition~\ref{def:StatIndist}) from the data independent uniform distribution $\PC_0 = \{1/2, 1/2\}$ in Eq.~\eqref{eq:dist-no-input} (with the probability exponentially close to $1$). 
    \item For the local SWAP strategy, suppose we measure any one of the terms $m_k(\vec{x},\vec{x'})$ in Eq.~\eqref{eq:appx-2norm-expand} for a given input data pair $\vec{x}$ and $\vec{x'}$ with a polynomial number of measurement shots. The associated distribution $\PC_{m_k(\vec{x},\vec{x'})}$ is statistically indistinguishable from a data-independent fixed distribution $\Tilde{\PC}_0 = \{3/4, 1/4\}$ (with the probability exponentially close to 1).
\end{enumerate}
\end{lemma}
Again, we provide the proof in Appendix~\ref{appendix:pqk-proof}.

Supplemental Lemma~\ref{sup-lem:pqk-tomo-swap} plays the same pivotal role for the projected kernel as Supplemental Lemma~\ref{sup-lem:fidelity-swap} for the fidelity kernel. They both capture the statistical indistinguishability of the distributions obtained with quantum computers when performing an experiment. Therefore, an identical reasoning can be applied here, that is, we argue that any trained model built frompolynomial samples is insensitive to input data and thus poorly generalizes. Similar to before, we note that estimating projected kernel values (via Eq.~\eqref{eq:projected-gaussian-kernel-appx}) and model predictions (via the Representer Theorem) can be seen as forms of post processing measurement outcomes. If these model predictions can be distinguished from the model predictions constructed based on the fixed distribution, then the hypothesis task described in Supplemental Lemma~\ref{sup-lem:pqk-tomo-swap} would be succeeded and hence contradict the conclusion of the lemma. Finally, we note that the setting in Supplemental Lemma~\ref{sup-lem:pqk-tomo-swap} is strictly stronger than what we have in practice where there is no access to exact quantities of interest (i.e., coefficients in the tomography strategy and purity/overlap terms in the local SWAP strategy) and in turn exact kernel values.

\bigskip

In the following, we formalize the above. We will again use our notion of statistical indistinguishability out outputs from Definition~\ref{def:StatIndistOutputs}. Following Supplemental Lemma~\ref{sup-lem:pqk-tomo-swap}, we will observe the following quantities to be indistinguishable:
\begin{itemize}
    \item For the tomography strategy, the coefficient $c_{\sigma_k}(\vec{x})$ in Eq.~\eqref{eq:appx-coeff-pqk} estimated with polynomial samples is statistically indistinguishable from $\widehat{\kappa}^{(\rm rand)}_N$ in Eq.~\eqref{eq:k0-no-input-appx}.
    \item For the local SWAP strategy, the statistical estimate of the purity/overlap term is indistinguishable from another data-independent random variable i.e., 
\begin{align}\label{eq:k0-no-input2}
    \widehat{\kappa}^{(\rm biased \; rand)}_N = \frac{1}{N} \sum_{m = 1}^N \Tilde{\lambda}_m \;, 
\end{align}
where $\Tilde{\lambda}_m$ takes $+1$ value with probability $3/4$ and $-1$ value with probability $1/4$.
\end{itemize} 
To obtain an estimate of a 2-norm distance on qubit $k$, one has to estimate either all coefficients for the reduced states in Eq.~\eqref{eq:reduced-state-pauli-basis} for the tomography strategy, or all purity/overlap terms Eq.~\eqref{eq:appx-coeff-dist-pqk} for the local SWAP strategy. Each term follows Supplemental Lemma~\ref{sup-lem:pqk-tomo-swap}. This leads to the statistical indistinguishability between the estimated 2-norm on the qubit $k$ and some data-independent random variable (with high probability). 

To construct an estimate for a kernel value one requires statistical estimates of the 2-norms associated with all single qubit subsystems and then one sums them as in Eq.~\eqref{eq:projected-gaussian-kernel-appx}. In the following Supplemental Proposition~\ref{sup-prop-norm-indis-pqk} we argue these kernel values are indistinguishable. 

\begin{supplemental_proposition}\label{sup-prop-norm-indis-pqk}
Consider the same assumption as in Supplemental Lemma~\ref{sup-lem:pqk-tomo-swap} and a polynomial number of measurement shots used to estimate the 2-norms. For any input data pair $\vec{x}$ and $\vec{x'}$, we have the following, with the probability exponentially close to $1$ (i.e., with the probability $1 - \delta_a$ with $\delta_a \in \OC(\tilde{c}^{-n})$ for $\tilde{c}>1$)

\begin{itemize}
    \item For the tomography strategy, the statistical estimate of the 2-norm between two reduced quantum states is statistically indistinguishable (as per Defintion \ref{def:StatIndistOutputs}) from an input-data-independent random variable
\begin{align}
    \ell_N^{(\rm rand, T)} = \frac{1}{2}\left[ \left(\widehat{\kappa}_{N,1}^{(\rm rand)} - \widehat{\kappa}_{N,2}^{(\rm rand)} \right)^2 + \left(\widehat{\kappa}_{N,3}^{(\rm rand)} - \widehat{\kappa}_{N,4}^{(\rm rand)} \right)^2 + \left(\widehat{\kappa}_{N,5}^{(\rm rand)} - \widehat{\kappa}_{N,6}^{(\rm rand)} \right)^2 \right] \;,
\end{align}
where $\left\{\widehat{\kappa}_{N,i}^{(\rm rand)} \right\}_{i=1}^6$ are different instances of $\widehat{\kappa}_N^{(\rm rand)}$ defined in Eq.~\eqref{eq:k0-no-input-appx}.
\item For the SWAP strategy, the statistical estimate of the 2-norm between two reduced quantum states is statistically indistinguishable from an input-data-independent random variable
\begin{align}
     \ell_N^{(\rm rand, S)} = \widehat{\kappa}^{(\rm biased \; rand)}_{N,1} + \widehat{\kappa}^{(\rm biased \; rand)}_{N,2} - 2 \widehat{\kappa}^{(\rm biased \; rand)}_{N,3} \;,
\end{align}
where  $\left\{\widehat{\kappa}^{(\rm biased \; rand)}_{N,i} \right\}_{i=1}^3$ are different instances of $\widehat{\kappa}^{(\rm biased \; rand)}_N$ as defined in Eq.~\eqref{eq:k0-no-input2}.
\end{itemize}
In addition, the estimate of the projected quantum kernel is statistically indistinguishable from an input-data-independent random variable
\begin{align} \label{eq:rand-pqk-tomography}
    \widehat{\kappa}^{(\rm rand, PQ)}_N = \exp\left[ \gamma \sum_{k=1}^n \ell_{N,k}^{(\rm rand)}\right] \;, 
\end{align}
where $\left\{\ell_{N,k}^{(\rm rand)}\right\}_{k=1}^n$ are different instances of either $\ell_N^{(\rm rand, T)}$ for the tomography strategy, or $\ell_N^{(\rm rand, S)}$ for the SWAP strategy.
\end{supplemental_proposition}
The proof is detailed in Appendix~\ref{appendix:pqk-proof}.

Identical to the fidelity kernel setting, we now argue that the statistical indistinguishability between estimated projected kernel value and some data-independent random variable leads to estimated model predictions that are insensitive to unseen input data (with high probability). To proceed, consider a training dataset $\SC$ of polynomial size $N_s$ which corresponds to a set that contains projected kernel values over possible training data pairs (excluding the trivial ones  $\kappa^{\rm PQ}(\vec{x},\vec{x}) = 1 $). 
\begin{align}\label{eq:appx-kernel-set-projected}
    \KC_{\rm PQ} = \left\{ \kappa^{\rm PQ}(\vec{x}, \vec{x'}) \; | \; \forall \vec{x}, \vec{x'} \in \SC \; ; \; \vec{x} \neq \vec{x'} \right\} \;.
\end{align}

Supplemental Proposition~\ref{sup-prop-norm-indis-pqk} applies for each kernel value in $\KC_{\rm PQ}$, leading to its estimated kernel being indistinguishable with probability exponentially close to $1$. Since the cardinality of $\KC_{\rm PQ}$ is at most polynomial in the number of qubits, it follows that the probability of all estimated kernel values being indistinguishable remains exponentially close to $1$. This leads to the indistinguishability of the estimated Gram matrix from some data-independent random matrix. This precludes the usefulness of the rest of the kernel methods pipeline. Again we demonstrate this for the task of kernel ridge regression and provide the form of an input-data-independent random variable that the model prediction approximately takes.
\begin{supplemental_corollary} \label{sup-coro:indis-model-pqk}
Consider a kernel ridge regression task with the projected quantum kernel, a squared loss function and a training dataset $\SC = \{\vec{x}, y_i \}_i$ with $N_s \in \OC(\poly(n))$. Given that the kernel value is estimated using either the SWAP test or tomography strategies and under the same assumption as in Supplemental Lemma~\ref{sup-lem:pqk-tomo-swap}, the following statements hold with probability exponentially close to $1$ (i.e., with the probability $1 - \delta_a$ with $\delta_a \in \OC(\tilde{c}^{-n})$ for $\tilde{c}>1$)
\begin{itemize}
    \item The estimated Gram matrix $\widehat{K}$ is statistically indistinguishable (as per Definition~\ref{def:StatIndistOutputs}) from an input-data-independent matrix $\widehat{K}_N^{\rm (rand, PQ)}$ whose diagonal elements are $1$ and off-diagonal elements are instances of $\widehat{\kappa}^{(\rm rand, PQ)}_N$ determined by Supplemental Proposition~\ref{sup-prop-norm-indis-pqk}.
    \item The estimated optimal parameters are statistically indistinguishable from the input-data-independent random variables
\begin{align}
    \vec{a}_{\rm rand}(\vec{y}, \lambda) = \left( \widehat{K}_N^{\rm (rand, PQ)} - \lambda \mathbb{1} \right)^{-1}\vec{y} \;. 
\end{align}
where $\vec{y}$ is a vector of output data points with its $i^{\rm th}$ element equal to $y_i$ and $\lambda$ is the regularization parameter. 
\item The model prediction on an unseen input data $\vec{x}$ is  statistically indistinguishable from the input data independent random variable
\begin{align}
    h_{\rm rand} = \vec{a}_{\rm rand}(\vec{y}, \lambda) ^T\vec{k}_N^{(\rm rand, PQ)} \;,
\end{align}
where $\vec{k}_N^{(\rm rand, PQ)}$ is a random vector where each of its element is an instance of $\widehat{\kappa}^{(\rm rand, PQ)}_N$ in Eq.~\eqref{eq:rand-pqk-tomography}.
\end{itemize}
\end{supplemental_corollary}
We refer the reader to Appendix~\ref{appendix:pqk-proof} for the proof of the corollary.

Again, as in the case of the quantum fidelity kernel, the training process hard encodes the training label which makes it possible to have a small training error. However, crucially the model does not obtain any information about the input data both during the training and the prediction phases which results in very poor generalization.} 

\begin{figure}[t]
\includegraphics[width=.6\columnwidth]{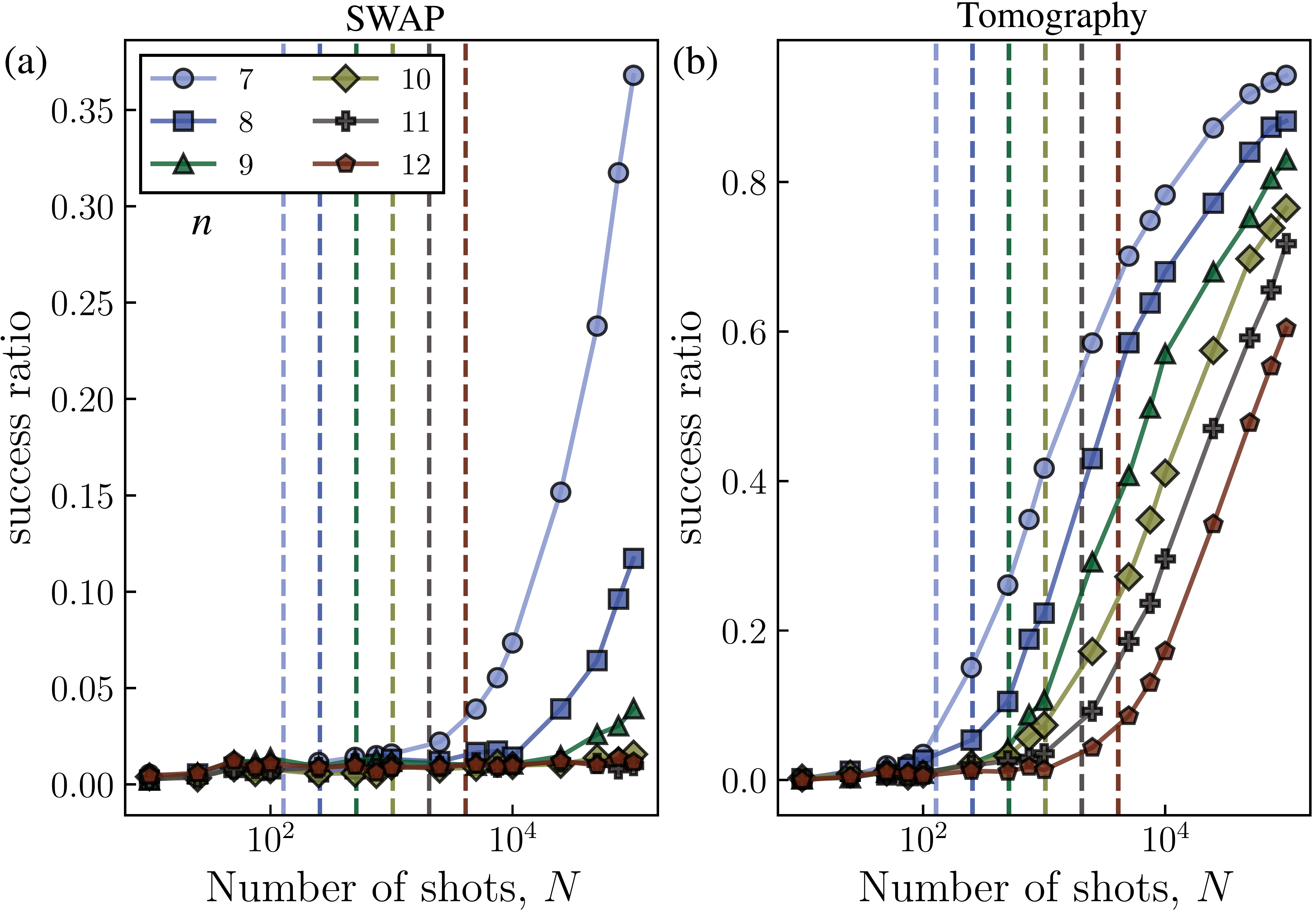}
\caption{\textbf{Effect of exponential concentration on the estimated Gram matrix for the projected kernel.} We plot the success ratio, i.e. the number of estimates of relevant quantities that pass a binomial test from their respective fixed distribution for a p-value$=0.01$ to the total number of estimates as a function of \sth{measurement shots $N$ and qubits $n$}. Two strategies for preparing the projected quantum kernel are considered, as discussed at the start of this section. In panel (a) we consider the SWAP test, where the relevant experimental quantities are the overlaps and purities of reduced data encoded states. In panel (b), where we consider the tomography strategy, the relevant quantities are the Pauli coefficients of reduced data encoded states. The x-axis indicates the number of shots used per kernel value, and vertical lines indicate the (exponentially increasing) dimension of the Hilbert space $2^n$. Here the training data set is of size $N_s = 25$.}\label{fig-sup:effect=expo-proj-gram}
\end{figure}

\begin{figure}[t]
\includegraphics[width=.55\columnwidth]{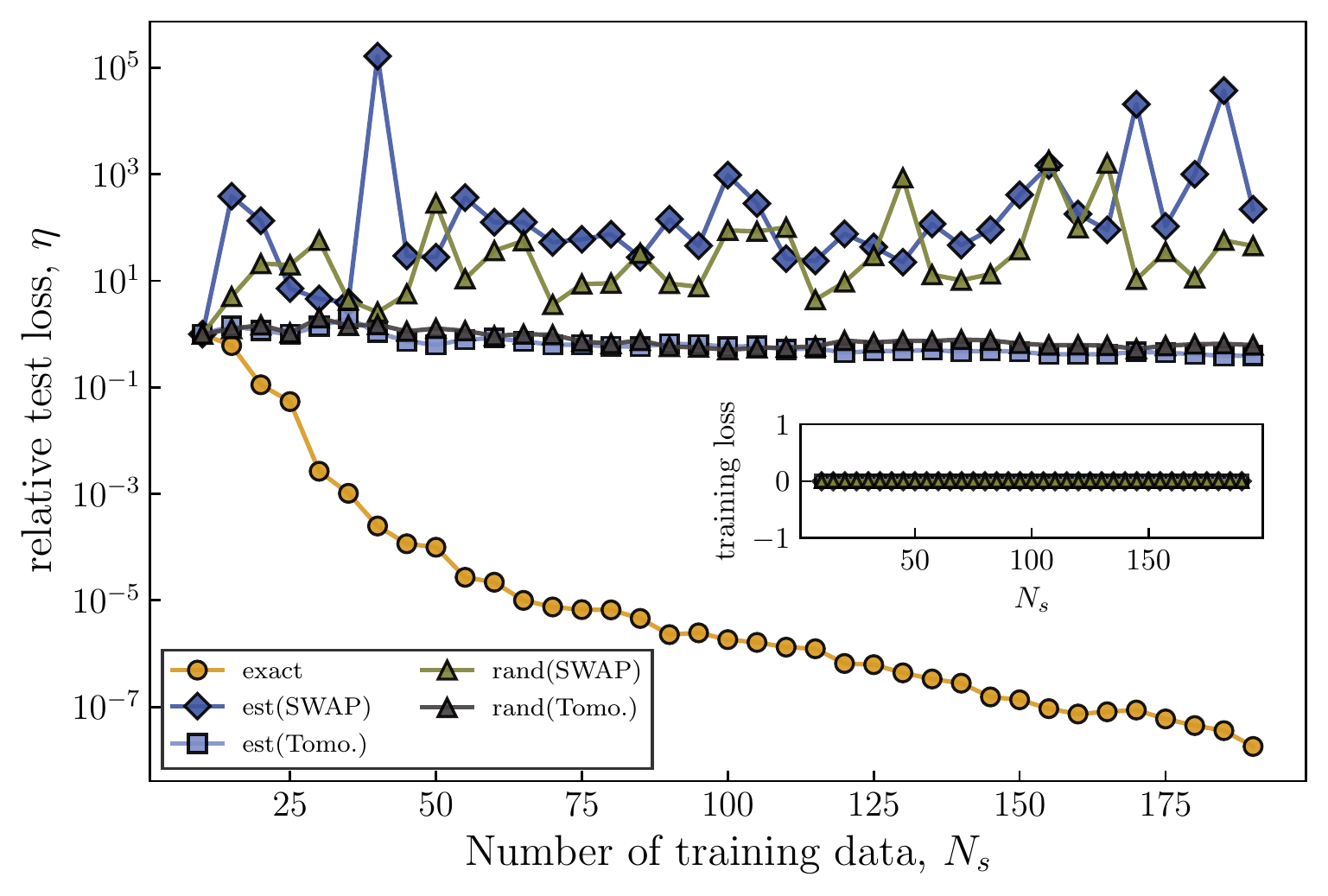}
\caption{\textbf{Effect of exponential concentration on training and generalization performance for projected quantum kernels.} 
In the main plot, a \textit{relative} loss on a test dataset $\SC_{\rm test}$ with respect to its initial value is plotted as a function of increasing training data and how the kernel values are obtained. In the inset, an \textit{absolute} training error is plotted as a function of increasing data. We note that each kernel value is estimated with $N=1000$ by either SWAP or tomography strategies, and the number of testing data points is $20$. The training is carried out with no regularization $\lambda = 0$. 
}\label{fig-sup:effect=expo-proj-gen}
\end{figure}

\subsubsection{Numerical simulation}\label{appendix:projected2}
In Fig.~\ref{fig-sup:effect=expo-proj-gram}, we numerically study the consequences of exponential concentration on the indistinguishability of statistical estimates for projected quantum kernels. Here, we consider a training dataset of size $N_s = 25$ and the data embedding is chosen such that it maps a classical input to a maximally expressive state, which leads to the exponential concentration of quantum states and in turn the projected quantum kernel (see. Sec.~\ref{sec:expressivity} and Theorem~\ref{thm:expressivity-kernel}). We perform binomial hypothesis tests on statistical estimates of the data-dependent quantities as though they were obtained from a quantum experiment to see whether or not they are statistically significant enough to be distinguishable from the associated fixed distributions. 
More specifically, in panel (a), local SWAP tests are employed to measure purities and overlaps of reduced states, we plot the success ratio of the estimates that pass the binomial test (with p-value below $0.01$) to the total number of estimates. Similarly, panel (b) illustrates this success ratio when the tomography is used to estimate coefficients of reduced states. In either case, we can see that to obtain a fixed success ratio an exponential number of measurement shots are required to have a significant increase in the ratio. Consequently, for larger system sizes, where only polynomial number of shots is feasible, statistical estimates of these relevant quantities are indistinguishable from the data-independent random variables. 

More practical matters of trainability and generalizability are numerically studied in Fig.~\ref{fig-sup:effect=expo-proj-gen} for a $12$-qubit simulation. Here, we consider a training dataset of size $N_s = 200$ where each individual input data is mapped to a maximally expressive state leading to exponential concentration. A true label is constructed in a similar manner as in Fig.~\ref{fig:effect-exp-con-gen} which gives perfect generalization when training on the whole dataset with exact kernel values. In the main plot, we study the generalization performance when a fraction of the dataset is used to train the model. With direct access to exact kernel values, the model generalizes better with increasing data, as expected. On the other hand, when kernel values are estimated with either the SWAP or tomography strategy and limited measurement samples ($N=1000$), generalization does not get better with increasing training data. In addition, the behavior of the \textit{relative} error also closely follows a similar trend to when the model is trained on a random matrix as its Gram matrix. Lastly, we observe perfect training errors in all cases. We posit this is again because the optimization process directly encodes the information about training labels.

\setcounter{lemma}{3}
\setcounter{supplemental_proposition}{2}
\setcounter{supplemental_corollary}{2}

\sth{\subsubsection{Proof of analytical results}\label{appendix:pqk-proof}
Here we provide the proofs of the main analytical results regarding the practical consequence of exponential concentration on the projected quantum kernel as stated in Subsection \ref{appendix:projected1}. For the readers' convenience, we restate formal statements before detailing proofs.

\begin{lemma}
Given that the Eq.~\eqref{eq:2norm-vanish} is satisfied, it follows that the projected quantum kernel exponentially concentrates 
\begin{align}\label{eq:kernelconc}
    \Pr_{\vec{x},\vec{x'} \in \XC}\left[  \left| \kappa^{PQ}(\vec{x},\vec{x'}) - \mu \right|  \geq \delta \right] \leq \frac{\beta}{\delta^2} \;,
\end{align}
where $\beta \in \OC(1/b^n)$ for some $b>1$. 
\end{lemma}
\begin{proof} We first show that the variance of the projected kernel is exponentially small due to the state concentration in Eq.~\eqref{eq:2norm-vanish}. The variance of the projected kernel can be bounded
\begin{align}
    \Var_{\vec{x},\vec{x'}}[k^{PQ}(\vec{x},\vec{x'})] & = \Var_{\vec{x},\vec{x'}}[1 - k^{PQ}(\vec{x},\vec{x'})]\label{eq:appx-var-pqk-2norm1} \\
    & \leq \mathbb{E}_{\vec{x},\vec{x'}}[(1 - k^{PQ}(\vec{x},\vec{x'}))^2] \\
    & \leq \mathbb{E}_{\vec{x},\vec{x'}}[1 - k^{PQ}(\vec{x},\vec{x'})] \\
    & =  \mathbb{E}_{\vec{x},\vec{x'}} \left[ 1 -  e^{ - \gamma \sum_{k=1}^n \| \rho_k(\vec{x}) - \rho_k(\vec{x'})\|^2_2} \right] \\
    & \leq \mathbb{E}_{\vec{x},\vec{x'}}  \left[ \gamma \sum_{k=1}^n \| \rho_k(\vec{x}) - \rho_k(\vec{x'})\|^2_2 \right]  \\
    & = \gamma \sum_{k=1}^n \mathbb{E}_{\vec{x},\vec{x'}} \| \rho_k(\vec{x}) - \rho_k(\vec{x'})\|^2_2 \;, \label{eq:proof-thm1-pqk-var} \\
    & \leq \gamma \sum_{k=1}^n \mathbb{E}_{\vec{x},\vec{x'}} \left(\left\|\rho_k(\vec{x}) - \frac{\mathbb{1}_k}{2} \right\|_2 +  \left\|\rho_k(\vec{x'}) - \frac{\mathbb{1}_k}{2} \right\|_2  \right)^2 \\
    & \leq 2\gamma \sum_{k=1}^n  \left(   \mathbb{E}_{\vec{x}} \left\|  \rho_k(\vec{x})  - \frac{\mathbb{1}_k}{2}\right \|_2 +   \mathbb{E}_{\vec{x'}} \left\|  \rho_k(\vec{x'})  - \frac{\mathbb{1}_k}{2}\right \|_2 \right)     \\ 
    & \in \OC\left(\frac{n}{b^n} \right) \;,
\end{align}
where the first equality is due to the fact that $\Var_{\vec{\alpha}}[c_1 A(\vec{\alpha}) + c_2] = c_1^2 \Var_{\vec{\alpha}}[A(\vec{\alpha})]$ for constants $c_1$ and $c_2$, in the second inequality we use the bound $(1 - k^{PQ}(\vec{x},\vec{x'})) \leq 1$ and take the upper bound, the second equality follows from substituting in the kernel definition in Eq.~\eqref{eq:projected-gaussian-kernel-appx} and in the third inequality we use $1 - e^{-t} \leq t $. Then, in the fourth inequality, we denote $\mathbb{1}_k$ as the identity matrix on qubit $k$ and use the triangle inequality. The fifth inequality is from $(t+s)^2 \leq 2t^2 + 2s^2$ and in the last line the state concentration in Eq.~\eqref{eq:2norm-vanish} is used. Finally, the kernel concentration, Eq.~\eqref{eq:kernelconc}, follows directly from Chebyshev's inequality. 
\end{proof}

\bigskip

\begin{lemma}
Assume the 2-norm distance between the reduced data encoding states exponentially vanishes as in Eq.~\eqref{eq:2norm-vanish}. 
Consider the following two scenarios
\begin{enumerate}
    \item For the tomography strategy, suppose we measure any coefficient $c_{\sigma_k(\vec{x})}$ of a reduced state on the qubit $k$ in Eq.~\eqref{eq:appx-coeff-pqk} for a given input data $\vec{x}$ with a polynomial number of measurement shots. The associated distribution $\PC_{\sigma_k, \vec{x}}$ defined in Eq.~\eqref{eq:appx-coeff-dist-pqk} is statistically indistinguishable (as per Definition~\ref{def:StatIndist}) from the data independent uniform distribution $\PC_0 = \{1/2, 1/2\}$ in Eq.~\eqref{eq:dist-no-input} (with the probability exponentially close to $1$).
    \item For the local SWAP strategy, suppose we measure any one of the terms $m_k(\vec{x},\vec{x'})$ in Eq.~\eqref{eq:appx-2norm-expand} for a given input data pair $\vec{x}$ and $\vec{x'}$ with a polynomial number of measurement shots. The associated distribution $\PC_{m_k(\vec{x},\vec{x'})}$ is statistically indistinguishable from a data-independent fixed distribution $\Tilde{\PC}_0 = \{3/4, 1/4\}$ (with probability exponentially close to 1).
\end{enumerate}
\end{lemma}

\medskip

\begin{proof}
To prove our result, we first consider a fixed distribution $\PC = \{ p, 1-p\}$ and some perturbed distribution $\PC_{\varepsilon} = \{ p + \varepsilon, 1 - (p + \varepsilon)$. Recall from Supplemental Proposition~\ref{sup-prop:indistin-prob} that the perturbation $\varepsilon$ plays a crucial role in the success probability of the hypothesis test (with $N$ samples) 
\begin{align}
 {\rm Pr}\left({\rm ``right \, decision \, between \,} \HC_{0} \, \text{and} \, \HC_1" \right)
    \leq &  \left( \frac{1}{2} + \frac{N \varepsilon}{4} \right) \, .
\end{align}
Then, if the perturbation becomes exponentially small, $\PC$ and $\PC_\epsilon$ are statistically indistinguishable with the polynomial samples $N\in\OC(\poly(n))$ for large $n$ by Definition~\ref{def:StatIndist}. 

Our proof for each of the scenarios mainly consists of using the exponential concentration of the 2-norm in Eq.~\eqref{eq:2norm-vanish} (i) to show that for a given input the quantity that we are interested in measuring is exponentially close to some fixed data-independent value (with high probability) and (ii) to identify the relevant fixed distribution (in case of the tomography strategy, this fixed distribution is $\PC_0 = \{1/2, 1/2\}$, while in the case of the local SWAP test we have $\Tilde{\PC_0} = \{ 3/4, 1/4\}$). Together, this establishes statistical indistinguishability between the distribution associated with the quantity and the corresponding fixed distribution.

\medskip

\noindent\underline{Tomography strategy:} The quantity of interest is the expectation value of a Pauli observable on the qubit $k$ i.e., $c_{\sigma_k, \vec{x}}$. We now show that $c_{\sigma_k, \vec{x}}$ is exponentially concentrated by looking at its variance.
\begin{align}
    \Var_{\vec{x}} \left[c_{\sigma_k, \vec{x}}  \right]  & = \Var_{\vec{x}} \left[ \Tr[\rho_k(\vec{x}) \sigma_k] \right] \\
    & = \Var_{\vec{x}} \left[ \Tr \left[ \left(\rho_k(\vec{x}) - \frac{\mathbb{1}_k}{2} \right) \sigma_k \right] \right] \\
    & \leq \Ebb_{\vec{x}}\left[ \Tr \left[ \left(\rho_k(\vec{x}) - \frac{\mathbb{1}_k}{2} \right) \sigma_k \right] \right] ^2 \\
    & \leq \Ebb_{\vec{x}} \left[ \left\| \rho_k(\vec{x}) - \frac{\mathbb{1}_k}{2} \right\|_2^2 \left\|\sigma_k \right\|_2^2\right] \\
    & \leq  \sqrt{2} \Ebb_{\vec{x}}  \left\| \rho_k(\vec{x}) - \frac{\mathbb{1}_k}{2} \right\|_2 \\
    & \in \OC\left( \frac{1}{b^n}\right) \;,
\end{align}
where the second equality is due to the Pauli operator being traceless, the second inequality is from using H\"{o}lder's inequality and in the third inequality we use the fact that $\| \sigma_k\|_2^2 = 2$ as well as $\|\rho_k(\vec{x}) - \mathbb{1}_k/2 \|_1 \leq 1/\sqrt{2}$. The final line follows from the assumption that the 2-norm distance between the single qubit reduced data-encoded states and the maximally mixed state vanishes as per Eq.~\eqref{eq:2norm-vanish}. This shows that the exponential concentration of the variance towards its mean. The mean can be further shown to exponentially vanish as follows
\begin{align}
    \Ebb_{\vec{x}} \left[c_{\sigma_k, \vec{x}}  \right] & = \Ebb_{\vec{x}}\left[ \Tr \left[ \left(\rho_k(\vec{x}) - \frac{\mathbb{1}_k}{2} \right) \sigma_k \right] \right] \\
    & \leq \Ebb_{\vec{x}}\left[  \left\| \rho_k(\vec{x}) - \frac{\mathbb{1}_k}{2} \right\|_2 \left\|\sigma_k \right\|_2\right] \\
    & \in \OC\left( \frac{1}{b^n}\right) \;,
\end{align}
where the equality is due to the Pauli operator being traceless and the inequality is from using H\"{o}lder's inequality. Together, we have that the coefficient exponentially concentrates towards some exponentially small value. 

By applying the Chebyshev's inequality, it can be shown that for any given $\vec{x}$ the coefficient is exponentially small with high probability. That is, by denoting $\mu = \Ebb_{\vec{x}} \left[c_{\sigma_k, \vec{x}}  \right]$ and $\sigma^2 = \Var_{\vec{x}} \left[c_{\sigma_k, \vec{x}}  \right]$, we have
\begin{align} 
    {\rm Pr}_{\vec{x}}\left[\left| c_{\sigma_k, \vec{x}} - \mu \right| \geq k \sigma \right] \leq \frac{1}{k^2} \;.
\end{align}
Then by choosing $k = 1/\sqrt{\sigma}$ and inverting the inequality of Eq.~\eqref{eq:appx-chebyshev}, this gives us
\begin{align}
   {\rm Pr}_{\vec{x}}\left[\left|c_{\sigma_k, \vec{x}} - \mu \right| \leq \sqrt{ \sigma} \right] \geq 1 - \sigma\; .
\end{align}
Therefore, $c_{\sigma_k, \vec{x}} $ takes a value between $\mu - \sqrt{\sigma}$ and $\mu + \sqrt{\sigma}$ (which are exponentially small) with the probability at least $1 - \sigma$ (which is exponentially close to $1$). Lastly, the fixed distribution can be identified by replacing $c_{\sigma_k, \vec{x}}$ in $\PC_{\sigma_k,\vec{x}}$ with $0$ (i.e., no perturbation) leading to $\PC_0 =\{1/2, 1/2\}$. This finishes the first part of the proof.

\bigskip

\noindent\underline{SWAP strategy:} Here we are interested in estimating $m_k(\vec{x},\vec{x'}) = \Tr[\rho_k(\vec{x})\rho_k(\vec{x'})]$ which corresponds to the purity if $\vec{x} = \vec{x'}$ and to the overlap if  $\vec{x} \neq \vec{x'}$. 
To see the concentration of $m_k(\vec{x},\vec{x'})$, we consider the variance of $m_k(\vec{x},\vec{x'})$.
\begin{align}    \Var_{\vec{x},\vec{x'}}\left[m_k(\vec{x},\vec{x'})\right] & = \Var_{\vec{x},\vec{x'}}\left[m_k(\vec{x},\vec{x'}) - \frac{1}{2}\right] \\
    & \leq \Ebb_{\vec{x},\vec{x'}}\left[ \left( \Tr[\rho_k(\vec{x})\rho_k(\vec{x'})] -\frac{1}{2}\right)^2\right] \\
    & = \Ebb_{\vec{x},\vec{x'}}\left[ \left( \Tr[\left(\rho_k(\vec{x}) - \frac{\mathbb{1}_k}{2}\right)\rho_k(\vec{x'})] \right)^2\right] \\
    & \leq \Ebb_{\vec{x},\vec{x'}} \left\| \rho(\vec{x}) - \frac{\mathbb{1}_k}{2}\right\|_2^2 \left\| \rho(\vec{x}) \right\|_2^2 \\
    & \leq \Ebb_{\vec{x}} \left\| \rho(\vec{x}) - \frac{\mathbb{1}_k}{2}\right\|_2 \\
    & \in  \OC\left( \frac{1}{b^n}\right) \label{eq:overlap-var-range}\;,
\end{align}
where the second inequality is by H\"{o}lder's inequality and the third inequality is due to $\| \rho_k(\vec{x})\|_2^2 \leq 1$. In addition, we can show the mean itself concentrates towards $1/2$.
\begin{align}
    \left| \Ebb_{\vec{x},\vec{x'}}\left[ m_k(\vec{x},\vec{x'})\right] - \frac{1}{2}\right| & \leq  \Ebb_{\vec{x},\vec{x'}} \left| m_k(\vec{x},\vec{x'}) - \frac{1}{2}\right| \\
    & =  \Ebb_{\vec{x}, \vec{x'}} \left|\Tr[\rho_k(\vec{x})\rho_k(\vec{x'})] - \frac{1}{2} \right| \\
    & = \Ebb_{\vec{x}, \vec{x'}} \left| \Tr\left[\left(\rho_k(\vec{x}) - \frac{\mathbb{1}_k}{2} + \frac{\mathbb{1}_k}{2} \right)\rho_k(\vec{x'})\right]  - \frac{1}{2} \right| \\
    & = \Ebb_{\vec{x}, \vec{x'}} \left| \Tr\left[\left(\rho_k(\vec{x}) - \frac{\mathbb{1}_k}{2}  \right)\rho_k(\vec{x'})\right]\right| \\
    & \leq \Ebb_{\vec{x}, \vec{x'}} \left\|\rho_k(\vec{x}) - \frac{\mathbb{1}_k}{2}  \right\|_2 \left\| \rho_k(\vec{x'})\right\|_2 \\
    & \leq \Ebb_{\vec{x}}  \left\|\rho_k(\vec{x}) - \frac{\mathbb{1}_k}{2}  \right\|_2 \\
    & \in  \OC\left( \frac{1}{b^n}\right) \;, \label{eq:overlap-concen}
\end{align}
where the first inequality is due to Jensen's inequality, in the second inequality we apply H\"{o}lder's inequality and the third inequality is due to the fact that $\| \rho_k(\vec{x})\|_2 \leq 1$ for any quantum state. Together, we have found that $m_k(\vec{x}, \vec{x'})$ exponentially concentrates towards $1/2$. Note that for the purity case, one can repeat the above steps with $\vec{x} = \vec{x'}$, which gives the same result.

We now show that, for any given input pair, $m_k(\vec{x},\vec{x'})$ takes a value exponentially close to $1/2$ with probability exponentially close to $1$ using Chebyshev's inequality together with Eq.~\eqref{eq:overlap-var-range}. Following the same steps as for the local SWAP strategy we obtain
\begin{align}
   {\rm Pr}_{\vec{x}, \vec{x'}}\left[\left|m_k(\vec{x},\vec{x'}) - \mu \right| \leq \sqrt{ \sigma} \right] \geq 1 - \sigma\; ,
\end{align}
with $\mu = \Ebb_{\vec{x},\vec{x'}} \left[m_k(\vec{x},\vec{x'}) \right]$ and $\sigma^2 = \Var_{\vec{x},\vec{x'}} \left[m_k(\vec{x},\vec{x'}) \right]$. This implies that, with probability at least $1 - \sigma$ such that $\sigma \in \OC(b^{-n/2})$, the quantity $m_k(\vec{x}, \vec{x'})$ takes the value within the range between $\mu - \sqrt{\sigma}$ and $\mu + \sqrt{\sigma}$. Furthermore, by using Eq.~\eqref{eq:overlap-concen}, we conclude that $1/2 - \sqrt{\sigma} - \beta \leq m_k(\vec{x}, \vec{x'}) \leq 1/2 + \sqrt{\sigma} + \beta$ with probability at least $1-\sigma$. 

All that remains is to identify the appropriate fixed distribution. This can be found by replacing $m_k(\vec{x},\vec{x'})$ in $\PC_{m_k(\vec{x},\vec{x'})}$ with $1/2$ (which is the concentration point of the mean) leading to $\widehat{\PC}_0 = \{3/4,1/4\}$. This completes the proof. 

\end{proof}

\bigskip

\begin{supplemental_proposition}
Consider the same assumptions as in Supplemental Lemma~\ref{sup-lem:pqk-tomo-swap} and suppose a polynomial number of measurement shots is used to estimate the 2-norms. For any input data pair $\vec{x}$ and $\vec{x'}$, the following statements hold, with probability exponentially close to $1$ (i.e., with probability $1 - \delta_a$ with $\delta_a \in \OC(\tilde{c}^{-n})$ for $\tilde{c}>1$)

\begin{itemize}
    \item For the tomography strategy, the statistical estimate of the 2-norm between two reduced quantum states is statistically indistinguishable (as per Defintion \ref{def:StatIndistOutputs}) from an input-data-independent random variable
\begin{align}
    \ell_N^{(\rm rand, T)} = \frac{1}{2}\left[ \left(\widehat{\kappa}_{N,1}^{(\rm rand)} - \widehat{\kappa}_{N,2}^{(\rm rand)} \right)^2 + \left(\widehat{\kappa}_{N,3}^{(\rm rand)} - \widehat{\kappa}_{N,4}^{(\rm rand)} \right)^2 + \left(\widehat{\kappa}_{N,5}^{(\rm rand)} - \widehat{\kappa}_{N,6}^{(\rm rand)} \right)^2 \right] \;,
\end{align}
where $\left\{\widehat{\kappa}_{N,i}^{(\rm rand)} \right\}_{i=1}^6$ are different instances of $\widehat{\kappa}_N^{(\rm rand)}$ defined in Eq.~\eqref{eq:k0-no-input-appx}.
\item For the SWAP strategy, the statistical estimate of the 2-norm between two reduced quantum states is statistically indistinguishable from an input-data-independent random variable
\begin{align}
     \ell_N^{(\rm rand, S)} = \widehat{\kappa}^{(\rm biased \; rand)}_{N,1} + \widehat{\kappa}^{(\rm biased \; rand)}_{N,2} - 2 \widehat{\kappa}^{(\rm biased \; rand)}_{N,3} \;,
\end{align}
where  $\left\{\widehat{\kappa}^{(\rm biased \; rand)}_{N,i} \right\}_{i=1}^3$ are different instances of $\widehat{\kappa}^{(\rm biased \; rand)}_N$ as defined in Eq.~\eqref{eq:k0-no-input2}.
\end{itemize}
In addition, the estimate of the projected quantum kernel is statistically indistinguishable from an input-data-independent random variable
\begin{align} \label{eq:rand-pqk-tomography}
    \widehat{\kappa}^{(\rm rand, PQ)}_N = \exp\left[ \gamma \sum_{k=1}^n \ell_{N,k}^{(\rm rand)}\right] \;, 
\end{align}
where $\left\{\ell_{N,k}^{(\rm rand)}\right\}_{k=1}^n$ are different instances of either $\ell_N^{(\rm rand, T)}$ for the tomography strategy, or $\ell_N^{(\rm rand, S)}$ for the SWAP strategy.
\end{supplemental_proposition}

\medskip

\begin{proof}
To prove this result, we incorporate Supplemental Lemma~\ref{sup-lem:pqk-tomo-swap} with a union bound over different terms we need to measure. Since the total number of terms required to be measured on a quantum computer are at most polynomial in the number of qubits, we have that the total probability that all of them are simultaneously indistinguishable remains exponentially close to $1$. We note that this proof follows the same steps as Supplemental Corollary~\ref{sup-coro:fidelity-swap} (though, we still clearly unfold it for completeness).

First, consider estimating the 2-norm distance between two reduced data-encoded states on the qubit $k$ i.e.~$\| \rho_k(\vec{x}) - \rho_k(\vec{x'})\|_2$. In either strategy, we construct an estimation of this by  measuring a number of different quantities as detailed at the start of Section \ref{appendix:projected}. For the tomography strategy, the 2-norm can be expressed as
\begin{align} \label{eq:appx-2norm-tomo-proof}
    \| \rho_k(\boldsymbol{x}) - \rho_k(\boldsymbol{x'})\|_2^2 = \frac{1}{2}\left( (c_{x_k}(\vec{x}) - c_{x_k}(\vec{x'}))^2 + (c_{y_k}(\vec{x}) - c_{y_k}(\vec{x'}))^2 + (c_{z_k}(\vec{x}) - c_{z_k}(\vec{x'}))^2  \right) \;,
\end{align}
which means we are required to measure $6$ different expectation values (3 expectation values to reconstruct each reduced state). On the other hand, for the local SWAP strategy, the 2-norm can be expressed as
\begin{align}\label{eq:appx-2norm-swap-proof}
    \| \rho_k(\boldsymbol{x}) - \rho_k(\boldsymbol{x'})\|_2^2 = \Tr[\rho_k^2(\boldsymbol{x})] + \Tr[\rho_k^2(\boldsymbol{x'})] - 2 \Tr[\rho_k(\boldsymbol{x})\rho_k(\boldsymbol{x'})] \;, 
\end{align}
which implies there are total of $3$ different terms to be measured ($2$ purities and $1$ state overlap). When measuring each term (with polynomial measurement shots $N\in\OC(\poly(n))$, it follows from Supplemental Lemma~\ref{sup-lem:pqk-tomo-swap} that
\begin{align}\label{eq:appx-pqk-proof1}
    {\rm Pr}[E_i^{(k)}] \geq 1 - \delta \;, \;\; i \in \{1,2, ... , m \} \;,
\end{align}
where $\delta\in \OC(b^{-n})$ for some $b>1$ and $m$ is the total number of terms to be measured i.e., $m=6$ for the tomography strategy and $m=3$ for the local SWAP strategy. Here, we have denoted $E_i^{(k)}$ as the event that an estimate of a chosen term (among the $m$ terms) on the qubit $k$ is statistically indistinguishable from a data-independent random variable i.e., $\widehat{\kappa}^{\rm (rand)}$ in Eq.~\eqref{eq:k0-no-input-appx} for the tomography strategy, and $\widehat{\kappa}^{(\rm biased \; rand)}_N$ in Eq.~\eqref{eq:k0-no-input2}) for the local SWAP strategy. Thus, the probability that all $E_i^{(k)}$ simultaneously occur can be upper bounded as
\begin{align}
    {\rm Pr}\left[ \bigcap_{i=1}^m E_i^{(k)} \right] & =  1 - {\rm Pr}\left[ \bigcup_{i=1}^m \bar{E}_i^{(k)} \right] \label{eq:appx-pqk-lower-prob0}  \\
    & \geq  1 - \sum_{i=1}^{m} {\rm Pr}\left[\bar{E}^{(k)}_i \right] \\
    & \geq 1 - m\delta \;, \label{eq:appx-pqk-lower-prob1}
\end{align}
where we denote $\bar{E}^{(k)}_i$ has the conjugate event of $E^{(k)}_i$, the union bound is applied in the second line, and to reach the final line we use the fact that ${\rm Pr}[\bar{E}^{(k)}_i] \leq \delta_k$ by reversing the inequality in Eq.~\eqref{eq:appx-pqk-proof1}. This shows that for the two considered strategies the estimated 2-norm is indistinguishable from some data-independent random variable with probability exponentially close to $1$, in that part of the collected measurement statistics are distinguishable. Specifically, the statistical estimate of the 2-norm from the tomography strategy is indistinguishable from
\begin{align}
    \ell_N^{(\rm rand, T)} = \frac{1}{2}\left[ \left(\widehat{\kappa}_{N,1}^{(\rm rand)} - \widehat{\kappa}_{N,2}^{(\rm rand)} \right)^2 + \left(\widehat{\kappa}_{N,3}^{(\rm rand)} - \widehat{\kappa}_{N,4}^{(\rm rand)} \right)^2 + \left(\widehat{\kappa}_{N,5}^{(\rm rand)} - \widehat{\kappa}_{N,6}^{(\rm rand)} \right)^2 \right] \;.
\end{align}
This is obtained by replacing coefficients in Eq.~\eqref{eq:appx-2norm-tomo-proof} with $\left\{\widehat{\kappa}_{N,i}^{(\rm rand)} \right\}_{i=1}^6$ which are different instances of $\widehat{\kappa}_N^{(\rm rand)}$ defined in Eq.~\eqref{eq:k0-no-input-appx}. Similarly, the statistical estimate of the 2-norm from the local SWAP strategy is indistinguishable from
\begin{align}
     \ell_N^{(\rm rand, S)} = \widehat{\kappa}^{(\rm biased \; rand)}_{N,1} + \widehat{\kappa}^{(\rm biased \; rand)}_{N,2} - 2 \widehat{\kappa}^{(\rm biased \; rand)}_{N,3} \;.
\end{align}
where $\left\{\widehat{\kappa}^{(\rm biased \; rand)}_{N,i} \right\}_{i=1}^3$ are different instances of $\widehat{\kappa}^{(\rm biased \; rand)}_N$ in Eq.~\eqref{eq:k0-no-input2}. This completes the first part of the proof.

\medskip

For the second half of the proof, we show that the statistical estimate of the projected quantum kernel with polynomial measurement shots is also indistinguishable from some data-independent random variable (with high probability). The projected kernel in Eq.~\eqref{eq:projected-gaussian-kernel-appx} is estimated by (classically) post-processing estimated 2-norms over all single qubit subsystems. It follows from the conclusion of the first half of the proof that these estimated 2-norms on any single qubit subsystem are statistically indistinguishable with probability at least $1-m\delta$. We note that the estimate of the projected kernel is statistically indistinguishable when all estimated 2-norms on a single qubit subsystem are indistinguishable. 

To show this, we again use the union bound. Let $F_k = \bigcap_{i=1}^m E_i^{(k)}$ as the event that the estimated 2-norm on the qubit $k$ is indistinguishable. We have that the probability of all $F_k$ simultaneously occur to be bounded as
\begin{align}
    {\rm Pr}\left[ \bigcap_{k=1}^n F_k \right] & =  1 - {\rm Pr}\left[ \bigcup_{k=1}^n \bar{F}_k \right]  \\
    & \geq  1 - \sum_{k=1}^{n} {\rm Pr}\left[\bar{F}_k \right] \\
    & \geq 1 - nm\delta \;, \label{eq:appx-pqk-lower-prob1}
\end{align}
where the steps follow identically as Eq.~\eqref{eq:appx-pqk-lower-prob0} to Eq.~\eqref{eq:appx-pqk-lower-prob1}. We note that the total measurement shots spent here are $nmN$ which remains polynomial in the number of qubits.

In other words, for any given input pair, the statistical estimate of the projected kernel is indistinguishable with the polynomial shots from the data-independent random variable $\widehat{\kappa}^{(\rm rand, PQ)}_N$ with the probability at least $1 - nm\delta $ such that
\begin{align}
    \widehat{\kappa}^{(\rm rand, PQ)}_N = \exp\left[ \gamma \sum_{k=1}^n \ell_{N,k}^{(\rm rand)}\right] \;, 
\end{align}
where $\left\{\ell_{N,k}^{(\rm rand)}\right\}_{k=1}^n$ are different instances of either $\ell_N^{(\rm rand, T)}$ for the tomography strategy, or $\ell_N^{(\rm rand, S)}$ for the SWAP strategy. 
\end{proof}

\bigskip

\begin{supplemental_corollary} 
Consider a kernel ridge regression task with the projected quantum kernel, a squared loss function and a training dataset $\SC = \{\vec{x}, y_i \}_{i=1}^{N_s}$ with $N_s \in \OC(\poly(n))$. Given that the kernel value is estimated using either the SWAP test or tomography strategies and under the same assumption as in Supplemental Lemma~\ref{sup-lem:pqk-tomo-swap}, the following statements hold with probability exponentially close to $1$ (i.e., with the probability $1 - \delta_a$ with $\delta_a \in \OC(\tilde{c}^{-n})$ for $\tilde{c}>1$)
\begin{itemize}
    \item The estimated Gram matrix $\widehat{K}$ is statistically indistinguishable (as per Definition~\ref{def:StatIndistOutputs}) from an input-data-independent matrix $\widehat{K}_N^{\rm (rand, PQ)}$ whose diagonal elements are $1$ and off-diagonal elements are instances of $\widehat{\kappa}^{(\rm rand, PQ)}_N$ determined by Supplemental Proposition~\ref{sup-prop-norm-indis-pqk}.
    \item The estimated optimal parameters are statistically indistinguishable from the input-data-independent random variables
\begin{align}
    \vec{a}_{\rm rand}(\vec{y}, \lambda) = \left( \widehat{K}_N^{\rm (rand, PQ)} - \lambda \mathbb{1} \right)^{-1}\vec{y} \;. 
\end{align}
where $\vec{y}$ is a vector of output data points with its $i^{\rm th}$ element equal to $y_i$ and $\lambda$ is the regularization parameter. 
\item The model prediction on an unseen input data $\vec{x}$ is  statistically indistinguishable from the input data independent random variable
\begin{align}
    h_{\rm rand} = \vec{a}_{\rm rand}(\vec{y}, \lambda) ^T\vec{k}_N^{(\rm rand, PQ)} \;,
\end{align}
where $\vec{k}_N^{(\rm rand, PQ)}$ is a random vector where each of its element is an instance of $\widehat{\kappa}^{(\rm rand, PQ)}_N$ in Eq.~\eqref{eq:rand-pqk-tomography}.
\end{itemize}
\end{supplemental_corollary}

\medskip

\begin{proof}
We remark that the proof steps are identical to the proofs of Supplemental Corollary~\ref{sup-coro:fidelity-swap} (with the fidelity kernel replaced by the projected kernel). Nevertheless, we fully provide the proof of this result for the completeness and convenience. 

To prove our main result here, we combine Supplemental Proposition~\ref{sup-prop-norm-indis-pqk} with a union bound over the individual kernel values. More explicitly, Supplemental Proposition~\ref{sup-prop-norm-indis-pqk} concludes that for any given input pair the estimate of the projected kernel is indistinguishable (with the polynomial shots) from the data-independent random variable $\widehat{\kappa}^{(\rm rand, PQ)}_N$ in Eq.~\eqref{eq:rand-pqk-tomography} with probability exponentially close to $1$. Taking into account a polynomial number of training data $N_s \in \OC(\poly(n))$, we are required to measure a polynomial number of kernel values and each kernel follows Supplemental Proposition~\ref{sup-prop-norm-indis-pqk}. 

Now, consider the Gram matrix which can be constructed by measuring each kernel in the set $\KC_{\rm PQ}$ in Eq.~\eqref{eq:appx-kernel-set-projected}. This results in $N_s(N_s-1)/2$ unique kernel values. To proceed, we denote $\kappa_i$ as an $i^{\rm th}$ element in $\KC_{\rm PQ}$ with $i$ running from $1$ to $|\KC_{\rm PQ}| = N_s(N_s-1)/2$. In addition, let $E_i$ be the event that the estimate of $\kappa_i$ is statistically indistinguishable from $\widehat{\kappa}^{(\rm rand, PQ)}_N$. By invoking Supplemental Proposition~\ref{sup-prop-norm-indis-pqk}, we have
\begin{align}\label{eq:appx-proof-coro-pqk1}
    {\rm Pr}[E_i] \geq 1 - \delta_\kappa \;\;, \forall \kappa_i \in \KC_{\rm PQ} \;,
\end{align}
where $\delta_\kappa \in \OC(c^{-n})$ for $c > 1$. Then, by applying the union bound, the probability of all $E_i$ occurring at the same time is lower bounded as
\begin{align}
    {\rm Pr}\left[ \bigcap_{i} E_i \right] & =  1 - {\rm Pr}\left[ \bigcup_i \bar{E}_i \right]  \\
    & \geq  1 - \sum_{i=1}^{|\KC_{\rm PQ}|} {\rm Pr}\left[\bar{E}_i \right] \\
    & \geq 1 - \frac{N_s (N_s-1)\delta_k}{2} \;,
\end{align}
where we denote $\bar{E}_i$ is a conjugate event of $E_i$, we use the union bound in the second line and use ${\rm Pr}[\bar{E}_i] \leq \delta_k$ by reversing the final inequality in Eq.~\eqref{eq:appx-proof-coro-pqk1}. Thanks to $N_s \in \OC(\poly(n))$, the probability that each of the kernel values are statistically indistinguishable (and so the Gram matrix is statistically indistinguishable) is $1 - \delta_K$ with $\delta_K := N_s (N_s-1)\delta_k/2 \in \OC(\tilde{c}^{-n})$ for some $\tilde{c} >1$.

The statistical indistinguishability of the optimal parameters directly follows from the above result. This is since the optimal parameters are estimated by simply a post-processing of the Gram matrix.

Finally, the indistinguishability of the model prediction can be proven by further considering $N_s$ additional kernel values for the test input data. Hence, the same procedure via the union bound is repeated which leads to the conclusion that all estimated kernel values (from the Gram matrix and new ones associated with a test input) are indistinguishable from $\widehat{\kappa}^{(\rm rand, PQ)}_N$ with probability exponentially close to $1$.
\end{proof}}

\subsection{Indistinguishability of concentrated quantum states}\label{appendix:indis-states}
So far we have discussed how exponential concentration leads to the statistical indistinguishability of samples obtained from quantum computers. 
Here we extend this by showing that when quantum states concentrate they exhibit a stronger form of indistinguishability. In particular, we show that those quantum states remain indistinguishable even when given coherent access to many copies. Let us begin by stating the standard textbook result of the Helstrom measurement.
\begin{lemma}\label{lemma:state-discrim}
    Suppose that one of either $\rho$ or $\rho'$ is provided to us with equal probability. Then, the probability of \sth{making right decision}  which state is given to us using the optimal POVM measurement is
    \begin{align}
        {\rm Pr}[\sth{{\rm ``right \, decision \, between \,} \rho\, {\rm and} \, \sigma"} ] = \frac{1}{2} + \frac{\| \rho - \sigma \|_1}{4} \;.
    \end{align}
\end{lemma}
Importantly, when the two states become exponentially close to each other in the one-norm distance $\| \rho - \sigma \|_1 \in \OC(1/b^n)$ with $b > 1$, the probability of guessing correctly is exponentially tight to $1/2$. This strong concentration of quantum states can arise due to noise and entanglement. 

One could imagine that having access to multiple copies of quantum states and processing them coherently could help improve their distinguishability. However, we formalize in the following that, when the number of copies is polynomial in the number of qubits, the indistinguishability of the states remains. 
\begin{supplemental_proposition}\label{prop:state-discrim-multiple}
Given that $m$ copies of either $\rho$ or $\sigma$ are provided to us and we are allowed to coherently process all of them at the same time, the probability of \sth{making right decision} which state is given to us is upper bounded by
\begin{align}
    {\rm Pr}[\sth{{\rm ``right \, decision \, between \,} \rho\, {\rm and} \, \sigma"} ] \leq \frac{1}{2} + \frac{m \| \rho - \sigma \|_1}{4}
\end{align}
\end{supplemental_proposition}
\begin{proof}
By invoking Lemma~\ref{lemma:state-discrim}, the probability of guessing correctly given the $m$ copies of the quantum state is
\begin{align}
    {\rm Pr}[\sth{{\rm ``right \, decision \, between \,} \rho\, {\rm and} \, \sigma"} ] = \frac{1}{2} + \frac{\| \rho^{\otimes m} - \sigma^{\otimes m} \|_1}{4} \;.  \label{eq:state-dis-multiple}
\end{align}
We now upper bound the one-norm as
\begin{align}
    \| \rho^{\otimes m} - \sigma^{\otimes m} \|_1 & =  \| \rho^{\otimes m} - \left( \rho^{\otimes m -1} \otimes \sigma \right) +  \left( \rho^{\otimes m -1} \otimes \sigma \right) - \left( \rho^{\otimes m-2} \sigma^{\otimes 2} \right) + ... + \left( \rho \sigma^{\otimes m-1} \right) - \sigma^{\otimes m} \|_1 \\
    & = \| \rho^{\otimes m -1} \otimes (\rho - \sigma) + \rho^{\otimes m-2} \otimes (\rho - \sigma) \otimes \sigma + ... + (\rho - \sigma) \otimes \sigma^{\otimes m-1} \|_1 \\
    & \leq  \|  \rho^{\otimes m -1} \otimes (\rho - \sigma) \|_1 + \| \rho^{\otimes m-2} \otimes (\rho - \sigma) \otimes \sigma  \|_1 + ... + \|  (\rho - \sigma) \otimes \sigma^{\otimes m-1} \|_1 \\
    & \leq \|  \rho^{\otimes m -1} \|_1 \| \rho - \sigma \|_1 + \| \rho^{\otimes m-2}  \|_1 \| \rho - \sigma \|_1 \| \sigma \|_1 + ... +\| \rho - \sigma \|_1 \| \sigma^{\otimes m-1} \|_1 \\
    & \leq  m \| \rho - \sigma \|_1 \; ,
\end{align}
where the first inequality is the triangle inequality, the second inequality is from $\| A \otimes B \|_p \leq \| A \|_p \| B\|_p$ and in the last inequality we use the fact that one-norm of the quantum state is upper bounded by $1$. By substituting this upper bound back in the Eq.~\eqref{eq:state-dis-multiple}, the proof is completed.
\end{proof}
\begin{supplemental_corollary}
Assume quantum states $\rho$ and $\sigma$ are exponentially close in the one-norm i.e., $\| \rho - \sigma \|_1 \in \OC(1/b^n)$ with $b >1$. The quantum states $\rho$ and $\sigma$ are indistinguishable even if a polynomial number of copies can be processed coherently.
\end{supplemental_corollary}
\begin{proof}
    This can be proved by plugging $\| \rho - \sigma \|_1 \in \OC(1/b^n)$ and $m \in \OC(\poly(n))$ in the upper bound of Eq.~\eqref{eq:state-dis-multiple} in Supplemental Proposition~\ref{prop:state-discrim-multiple}. 
\end{proof}

Crucially, this implies that in the presence of noise any error mitigation techniques using multiple copies of the states to reduce the effect of noise cannot be used to avoid the data-independence of the solution which results from exponential concentration. More comprehensive treatment of the error mitigation is shown in Appendix~\ref{appendix:em-fail}.

\subsection{Sufficient condition to resolve kernel values}\label{appendix:further-discuss-impact}

So far we have argued that it is not possible to resolve kernel values with the polynomial number of measurement shots. In this section we formalize a sufficient condition to resolve the kernel values in the presence of exponential concentration. 
When the quantity $X(\vec{\alpha})$ is exponentially concentrated towards a fixed value $\mu$ for all $\vec{\alpha}$ (with high probability or deterministically), we can resolve the statistical estimate of $X(\vec{\alpha})$ from the estimates of some other $X(\vec{\alpha}')$ only if the \textit{relative} precision is sufficiently large. 
More concretely, we can quantify this \textit{relative} precision of the estimated $X(\vec{\alpha})$ by using a relative error
\begin{align}\label{eq:relative-error}
\tilde{\epsilon} = \frac{\epsilon}{\sqrt{\Var_{\vec{\alpha}}[X(\vec{\alpha})]}}
\end{align}
where $\epsilon$ is the statistical uncertainty due to finite measurement shots and $\sqrt{
\Var_{\vec{\alpha}}[X(\vec{\alpha})]}$ characterizes the concentration of $X(\vec{\alpha})$ over different $\vec{\alpha}$.
In particular, 
$\tilde{\epsilon} \lesssim 1$ is needed to resolve $X(\vec{\alpha})$ from some other $X(\vec{\alpha}')$.
The following proposition shows that exponential scaling in measurement shots is indeed required to be in this regime of sufficient resolution.  

\begin{supplemental_proposition}\label{sup-prop:shots-scale}
Assume that $X(\vec{\alpha})$ can be estimated as an expectation value of some observable $O$.  
In order to resolve $X(\vec{\alpha})$ from some other $X(\vec{\alpha}')$ to additive error $\epsilon$ and with probability at least $(1 - p)$ via Hoeffding's inequality it is sufficient to provide a number of measurement shots $N$ satisfying
\begin{align}
    N \geq  \frac{{2}\| O\|_{\infty}^2 \log(2/p)}{\tilde{\epsilon}^2 \Var_{\vec{\alpha}}[X(\vec{\alpha})]}\;,
\end{align}
where $\tilde{\epsilon}\lesssim1$ is the relative error as defined in Eq.~\eqref{eq:relative-error}. 
In addition, if $X(\vec{\alpha})$ is exponentially concentrated to a fixed point $\mu$ according to Definition~\ref{def:exp-concentration} (either deterministically or probabilistically), the number of shots scales as
\begin{align}
    N \in \Omega\left(\frac{b^{2 n}}{\tilde{\epsilon}^2}\right) \;; \;\; b>1 \;,
\end{align}
for arbitrary constant success probability, under the assumption $\| O\|_{\infty} \in \mathcal{O}(1)$. 
\end{supplemental_proposition}

\begin{proof}
Denote $X(\vec{\alpha}) = \langle O \rangle$, i.e. $X(\vec{\alpha})$ is the expectation of some observable $O$. Estimating $X(\vec{\alpha})$ in practice is done by measuring the observable $N$ times, with each outcome associated with one of the eigenvalues of $O$. Then, we can estimate the expectation value as
\begin{align}
    \hat O_{N} = \frac{1}{N} \sum_{i=1}^{N} O_i \;,
\end{align}
where $O_i \in [\lambda_{min}, \lambda_{max}]$ is the outcome of the $i^{\rm th}$ measurement and can be treated as a random variable, with $\lambda_{\rm min}$ and $\lambda_{\rm max}$ being the smallest and the largest eigenvalues of $O$. Invoking Hoeffding's inequality, we have

\begin{align}
    {\rm Pr}\left[ |  \hat O_{N} - \langle O \rangle |  \geq \epsilon \right] & \leq 2 \exp\left( - \frac{2N^2\epsilon^2}{\sum_i (\lambda_{\rm max} - \lambda_{\rm min})^2}\right) \\
    & {\leq 2 \exp\left( - \frac{N\epsilon^2}{2 \| O \|_{\infty}^2}\right)\,,
    }
\end{align}
where in the second line we have used the fact that $\lambda_{\rm max} - \lambda_{\rm min} \leq \|O\|_{\infty}$.

Let $p \geq 2 \exp\left( - \frac{N \epsilon^2}{{2} \| O \|_{\infty}^2}\right)$ be an upper bound on this probability. Upon rearranging, we see that the number of shots scales as
\begin{align}
    N \geq \frac{{2}\| O\|_{\infty}^2 \log(2/p)}{\epsilon^2} = \frac{{2}\| O\|_{\infty}^2 \log(2/p)}{\tilde{\epsilon}^2 \Var_{\vec{\alpha}}[X(\vec{\alpha})]}\;.
\end{align}
to obtain $ | \hat O_{N} - \langle O \rangle |  \leq \epsilon$ with probability at least $(1 - p)$. 
We recall that in order to resolve  $X(\vec{\alpha})$ from $X(\vec{\alpha'})$, we need $\tilde{\epsilon} \lesssim 1$ in general.

For deterministic exponential concentration, i.e. $|X(\vec{\alpha}) - \mu| \leq \beta \in \mathcal{O}(1/b^n)$, we have $\Var_{\vec{\alpha}}[X(\vec{\alpha})] \leq \mathbb{E}_{\vec{\alpha}}[X^2(\vec{\alpha})] \leq \beta^2$. For probabilistic exponential concentration, i.e.~${\rm Pr}_{\vec{\alpha}}\left[ |X(\vec{\alpha}) - \mu| \geq \delta\right] \leq \beta^2 / \delta^2$, we have $\Var_{\vec{\alpha}}[X(\vec{\alpha})] = \beta^2\in \mathcal{O}(1/b^{2n})$. Thus, in both cases, this leads to the number of measurement shots scaling as
\begin{align}
    N \geq \frac{{2}\| O\|_{\infty}^2 \log(2/p) }{\tilde{\epsilon}^2 \beta^2} \in \Omega\left(\frac{b^{2n}}{\tilde{\epsilon}^2}\right) \;,
\end{align}
with probability at least $(1-p)$ for a fixed $p$, where we have assumed that $\| O\|_{\infty} \in \mathcal{O}(1)$.
\end{proof}

Again, in the context of quantum kernel, exponential concentration negatively impacts the performance of kernel-based models in the sense that the Gram matrix $K$ cannot be efficiently estimated in practice. Supplemental Corollary~\ref{coro:scaling-gram} shows that an exponential number of measurement shots is required to distinguish $K$ from a fixed matrix $K_0$ for the fidelity quantum kernel.

\begin{supplemental_corollary}\label{coro:scaling-gram}
Assume that the fidelity quantum kernel $\kappa(\vec{x},\vec{x'})$ is exponentially concentrated towards a fixed value $\mu$ as in Definition \ref{def:exp-concentration}. Further, denote $K_0$ as the fixed matrix whose diagonal elements are all $1$ and off-diagonal elements are all $\mu$. Then, the number of measurement shots $N$ required to resolve each off-diagonal matrix element of the Gram matrix $K$ from $K_0$, 
up to additive error $\epsilon$ and for some arbitrary constant success probability, scales asymptotically as
\begin{align}
    N \in \Omega\left( \frac{N_s^2 b^{2n}}{\tilde{\epsilon}^2}\right) \;,
\end{align}
where $\tilde{\epsilon} \lesssim 1$ is the relative error as defined in Definition \ref{def:exp-concentration}, and we have assumed that statistical fluctuations associated with individual measurement outcomes stay constant.
\end{supplemental_corollary}
\begin{proof}
For a given training dataset of size $N_s$, the number of unique off-diagonal elements in the Gram matrix $K$ is $N_s(N_s-1)/2$. Hence, using Supplemental Proposition~\ref{sup-prop:shots-scale} for each matrix element, this leads to the total number of measurement shots scaling as claimed.
\end{proof}

We remark that one may be able to reduce the quadratic scaling in $N_s$ in the measurement shot scaling in Supplemental Corollary~\ref{coro:scaling-gram} using classical shadow protocols~\cite{huang2020predicting,elben2022randomized}. However, the exponential scaling in the number of qubits $n$ cannot be removed as this already happens at the level of measuring one element of the Gram matrix. Note that we consider the fidelity quantum kernel as an example but the conclusion can be easily extended to the case of the projected quantum kernel.

Thus, the absence of exponential concentration is a necessary condition to enable the potential of quantum kernels. For example, in the case of quantum support vector machine a non-vanishing separation between the two classes obtained from the feature map is essential. In general such embeddings are hard to construct; however, one strategy is to encode the problem structure directly into the embedding. Ref.~\cite{liu2021rigorous} shows that for a specific encryption-inspired learning task one can build a feature map, based on Shor's algorithm, leading to the absence of exponential concentration. In Ref.~\cite{glick2021covariant}, this embedding is shown to be a part of a family of so-called ``covariant quantum kernels'' where the symmetry properties of the target problem are encoded into the embedding. Exploring the extent to which such approaches generalize to other problems is an important direction for future research.

\section{Proof of Theorem~\ref{thm:expressivity-kernel}: Expressivity-induced concentration}

Here we provide a detailed proof of Theorem~\ref{thm:expressivity-kernel} which formally relates the expressivity of data-encoded unitaries and kernel concentration. For convenience, we recall the theorem below. 

\begin{theorem}[Expressivity-induced concentration]
Consider the fidelity quantum kernel as defined in Eq.~\eqref{eq:fidelity-kernel-mt} and the projected quantum kernel as defined in Eq.~\eqref{eq:projected-gaussian-kernel-mt}. Assume that input data $\vec{x}$ and $\vec{x'}$ are drawn from the same distribution, leading to an ensemble of unitaries $\mathbb{U}_{\vec{x}}$ as defined in Eq.~\eqref{eq:unitaryensemble}. We have
\begin{align}
    {\rm Pr}_{\vec{x},\vec{x'}}[|\kappa(\vec{x},\vec{x'}) -  \mathbb{E}_{\vec{x},\vec{x'}} [\kappa(\vec{x},\vec{x'})]| \geq \delta] \leq \frac{G_n(\varepsilon_{\mathbb{U}_{\vec{x}}})}{\delta^2} \;, \label{eq:thm1-concentration-express-mt} 
\end{align}
where $ \varepsilon_{\mathbb{U}_{\vec{x}}} = \|\mathcal{A}_{\mathbb{U}_{\vec{x}}}(\rho_0) \|_1$ is the data-dependent expressivity measure over $\mathbb{U}_{\vec{x}}$ defined in Eq.~\eqref{eq:expressivity-measure-epsilon}, and  $G_n(\varepsilon_{\mathbb{U}_{\vec{x}}})$ is a function of $\varepsilon_{\mathbb{U}_{\vec{x}}}$ defined as below.
\begin{enumerate}
    \item For the fidelity quantum kernel $\kappa(\vec{x},\vec{x'}) = \kappa^{FQ}(\vec{x},\vec{x'})$, we have
\begin{align}
    {G_n(\varepsilon_{\mathbb{U}_{\vec{x}}})} = \beta_{\rm Haar} +  \varepsilon_{\mathbb{U}_{\vec{x}}} ( \varepsilon_{\mathbb{U}_{\vec{x}}} + 2\sqrt{\beta_{\rm Haar}}) \; ,
\end{align}
where $\beta_{\rm Haar} = \frac{1}{2^{n-1}(2^n+1)}$.
    \item For the projected quantum kernel $\kappa(\vec{x},\vec{x'}) = \kappa^{PQ}(\vec{x},\vec{x'})$, we have 
\begin{align}
    {G_n(\varepsilon_{\mathbb{U}_{\vec{x}}})} = 4 \gamma n ( \tilde{\beta}_{\rm Haar} + \varepsilon_{\mathbb{U}_{\vec{x}}}) \;,
\end{align}
where $\tilde{\beta}_{\rm Haar} = \frac{3}{2^{n+1}+2}$.
\end{enumerate}
\end{theorem}

\begin{proof}
We separate the proof into two parts, corresponding to each type of quantum kernel.

\bigskip

\noindent\underline{\noindent Fidelity quantum kernel:} our strategy here is to compute the upper bound of the variance of the fidelity quantum kernel $\kappa^{FQ}(\vec{x},\vec{x'})$ and then use Chebyshev's inequality to show kernel concentration.
Now consider the following
\begin{align}
    \Var_{\vec{x},\vec{x'}}[\kappa^{FQ}(\vec{x},\vec{x'})] 
    & \leq \mathbb{E}_{\vec{x},\vec{x'}} [(\kappa^{FQ}(\vec{x},\vec{x'}))^2] \\
    & = \int dU(\vec{x}) \int dU(\vec{x'}) \Tr[U(\vec{x}) \rho_0 U^\dagger(\vec{x}) U(\vec{x'}) \rho_0 U^\dagger(\vec{x'})] \Tr[U(\vec{x}) \rho_0 U^\dagger(\vec{x}) U(\vec{x'}) \rho_0 U^\dagger(\vec{x'})] \\
    & = \int dU(\vec{x}) \int dU(\vec{x'}) \Tr[(U(\vec{x}))^{\otimes 2} \rho^{\otimes 2}_0 (U^{\dagger}(\vec{x}))^{\otimes 2} (U(\vec{x'}))^{\otimes 2}  \rho^{\otimes 2}_0 (U^{\dagger}(\vec{x'}))^{\otimes 2}] \;  \\
    & = \Tr[\int dU(\vec{x})(U(\vec{x}))^{\otimes 2} \rho^{\otimes 2}_0 (U^{\dagger}(\vec{x}))^{\otimes 2} \int dU(\vec{x'})(U(\vec{x'}))^{\otimes 2} \rho^{\otimes 2}_0 (U^{\dagger}(\vec{x'}))^{\otimes 2} ] \label{eq:proof-thm1-express} \\
    & = \Tr[ (\mathcal{V}_{{\rm Haar}}(\rho_0) -  \mathcal{A}_{\mathbb{U}_{\vec{x}}}(\rho_0) )^2 ] \label{eq:express-line-0}\; , \\
    & = \beta_{\rm Haar} + \Tr[ \mathcal{A}_{\mathbb{U}_{\vec{x}}}(\rho_0)(\mathcal{A}_{\mathbb{U}_{\vec{x}}}(\rho_0) - 2 \mathcal{V}_{{\rm Haar}}(\rho_0))] 
\end{align}
where the second equality comes from the fact that $\Tr[X]\Tr[Y] = \Tr[X\otimes Y] $ and $(AC) \otimes (BD)=(A\otimes B)(C\otimes D) $, in the fourth equality we use the fact that the two integrals are identical due to our starting assumptions and substitute in $\mathcal{A}_{\mathbb{U}_{\vec{x}}}(\rho_0)$ as defined in Eq.~\eqref{eq:expressivity-measure-mt}, and in the last line we introduce $\beta_{\rm Haar} = \Tr[(\mathcal{V}_{\rm Haar}(\rho_0))^2]$.
Additionally, $\beta_{\rm Haar} = \frac{1}{2^{n-1}(2^n+1)}$ which is the result of explicitly performing Haar integration and assuming that the input states $\rho_0$ are pure. We then rearrange the expression to get
\begin{align}
    |\Var_{\vec{x},\vec{x'}}[\kappa^{FQ}(\vec{x},\vec{x'})] - \beta_{\rm Haar}  | & \leq \left|\Tr[\mathcal{A}_{\mathbb{U}_{\vec{x}}}(\rho_0)(\mathcal{A}_{\mathbb{U}_{\vec{x}}}(\rho_0) - 2 \mathcal{V}_{\rm Haar}(\rho_0))] \right| \label{eq:proof-thm1-first-fqk}\\
    & \leq \Tr[\left|\mathcal{A}_{\mathbb{U}_{\vec{x}}}(\rho_0)(\mathcal{A}_{\mathbb{U}_{\vec{x}}}(\rho_0) - 2 \mathcal{V}_{\rm Haar}(\rho_0))\right|] \\
    & \leq \| \mathcal{A}_{\mathbb{U}_{\vec{x}}}(\rho_0)\|_2  \| \mathcal{A}_{\mathbb{U}_{\vec{x}}}(\rho_0) - 2 \mathcal{V}_{{\rm Haar}}\|_2 \label{eq:express-line-1} \\
    & \leq  \| \mathcal{A}_{\mathbb{U}_{\vec{x}}}(\rho_0)\|_2 \left( \| \mathcal{A}_{\mathbb{U}_{\vec{x}}}(\rho_0)\|_2 + 2 \|\mathcal{V}_{\rm Haar}(\rho_0)\|_2  \right) \\
    & \leq \varepsilon_{\mathbb{U}_{\vec{x}}} ( \varepsilon_{\mathbb{U}_{\vec{x}}} + 2\sqrt{\beta_{\rm Harr}}) \; ,\label{eq:proof-thm1-last-fqk}
\end{align}
where the second inequality is due to the triangle inequality (here $|A|=\sqrt{A\ad A}$), the third equality follows from the matrix H\"older's inequality, the fourth inequality is again due to another use of the triangle inequality. Finally, in the last inequality we use the monotonicity of the Schatten $p$-norms, along with the definitions of $\varepsilon_{\mathbb{U}_{\vec{x}}} =  \| \mathcal{A}_{\mathbb{U}_{\vec{x}}}(\rho_0)\|_1$ and $\beta_{\rm Harr}$. Having upper bounded the variance, we can now invoke Chebyshev's inequality to complete the first part of the proof.

\bigskip
\noindent\underline{ Projected quantum kernel:} we first note that as $1 - k^{PQ}(\vec{x},\vec{x'})$ is always non-negative and bounded by 1. Then, we have
\begin{align}
    \Var_{\vec{x},\vec{x'}}[k^{PQ}(\vec{x},\vec{x'})] & = \Var_{\vec{x},\vec{x'}}[1 - k^{PQ}(\vec{x},\vec{x'})]\label{eq:appx-var-pqk-2norm1} \\
    & \leq \mathbb{E}_{\vec{x},\vec{x'}}[(1 - k^{PQ}(\vec{x},\vec{x'}))^2] \\
    & \leq \mathbb{E}_{\vec{x},\vec{x'}}[1 - k^{PQ}(\vec{x},\vec{x'})] \\
    & =  \mathbb{E}_{\vec{x},\vec{x'}} \left[ 1 -  e^{ - \gamma \sum_{k=1}^n \| \rho_k(\vec{x}) - \rho_k(\vec{x'})\|^2_2} \right] \\
    & \leq \mathbb{E}_{\vec{x},\vec{x'}}  \left[ \gamma \sum_{k=1}^n \| \rho_k(\vec{x}) - \rho_k(\vec{x'})\|^2_2 \right]  \\
    & = \gamma \sum_{k=1}^n \mathbb{E}_{\vec{x},\vec{x'}} \| \rho_k(\vec{x}) - \rho_k(\vec{x'})\|^2_2 \;, \label{eq:proof-thm1-pqk-var}
\end{align}
where the second inequality uses $0 \leq (1 - k^{PQ}(\vec{x},\vec{x'})) \leq 1$, the second equality is from substituting in the definition of the projected quantum kernel in Eq.~\eqref{eq:projected-gaussian-kernel-mt}, and finally the last inequality is due to the fact  that $1 - e^{-t} \leq t $. 

Let us focus on one of the expectation values in the sum in Eq.~\eqref{eq:proof-thm1-pqk-var}.
\begin{align}
    \mathbb{E}_{\vec{x},\vec{x'}} \| \rho_k(\vec{x}) - \rho_k(\vec{x'})\|^2_2 & \leq \mathbb{E}_{\vec{x},\vec{x'}} \left( \left\|\rho_k(\vec{x}) - \frac{\mathbb{1}_k}{2} \right\|_2 +  \left\|\rho_k(\vec{x'}) - \frac{\mathbb{1}_k}{2} \right\|_2 \right)^2 \\
    & \leq  \mathbb{E}_{\vec{x},\vec{x'}} \left( 2 \left\|\rho_k(\vec{x}) - \frac{\mathbb{1}_k}{2} \right\|_2^2 + 2 \left\|\rho_k(\vec{x'}) - \frac{\mathbb{1}_k}{2} \right\|_2^2 \right) \\
    & = 2  \mathbb{E}_{\vec{x}} \left\|  \rho_k(\vec{x})  - \frac{\mathbb{1}_k}{2}\right \|_2^2 +  2 \mathbb{E}_{\vec{x'}} \left\|  \rho_k(\vec{x'})  - \frac{\mathbb{1}_k}{2}\right \|_2^2   \;,\label{eq:appx-var-pqk-2norm2}
\end{align}
where $\mathbb{1}_k$ is the identity matrix on the qubit $k$. The first inequality is due to the triangle inequality and the second inequality comes from the fact that $(t+s)^2 \leq 2t^2 + 2s^2$.  Now, consider
\begin{align}
    \mathbb{E}_{\vec{x}} \left\|  \rho_k(\vec{x})  - \frac{\mathbb{1}_k}{2}\right \|_2^2 
    = & \mathbb{E}_{\vec{x}} \Tr_k \left[ \Tr_{\bar k}\left( \rho(\vec{x}) - \mathbb{1}/2^n \right)\Tr_{\bar k}\left( \rho(\vec{x}) - \mathbb{1}/2^n \right)\right] \\
    = &  \mathbb{E}_{\vec{x}} \Tr[ (\rho(\vec{x}) - \mathbb{1}/2^n)\otimes (\rho(\vec{x}) - \mathbb{1}/2^n) ({\rm SWAP_{k_1,k_2}} \otimes \mathbb{1}_{\bar{k}_1, \bar{k}_2})] \\
    = &  \mathbb{E}_{\vec{x}} \Tr[(U(\vec{x})\otimes U(\vec{x}))(\sigma \otimes \sigma)(U^\dagger(\vec{x}) \otimes U^\dagger(\vec{x}))  ({\rm SWAP_{k_1,k_2}} \otimes \mathbb{1}_{\bar{k}_1, \bar{k}_2})] \\ 
    = & \mathbb{E}_{V \sim \rm Haar} \Tr[(V\otimes V)(\sigma \otimes \sigma)(V^\dagger\otimes V^\dagger)  ({\rm SWAP_{k_1,k_2}} \otimes \mathbb{1}_{\bar{k}_1, \bar{k}_2})] -  \Tr[\mathcal{A}_{\mathbb{U}_{\vec{x}}}(\rho_0)  ({\rm SWAP_{k_1,k_2}} \otimes \mathbb{1}_{\bar{k}_1, \bar{k}_2})]  \\
    \leq &  \left| \mathbb{E}_{V \sim \rm Haar} \Tr[(V\otimes V)(\sigma \otimes \sigma)(V^\dagger\otimes V^\dagger)  ({\rm SWAP_{k_1,k_2}} \otimes \mathbb{1}_{\bar{k}_1, \bar{k}_2})] \right| +  \left| \Tr[\mathcal{A}_{\mathbb{U}_{\vec{x}}}(\rho_0) ({\rm SWAP_{k_1,k_2}} \otimes \mathbb{1}_{\bar{k}_1, \bar{k}_2})] \right|  \\
    \leq & \left|\mathbb{E}_{V \sim \rm Haar} \Tr[(V\otimes V)(\sigma \otimes \sigma)(V^\dagger\otimes V^\dagger)  ({\rm SWAP_{k_1,k_2}} \otimes \mathbb{1}_{\bar{k}_1, \bar{k}_2})] \right| + \| \mathcal{A}_{\mathbb{U}_{\vec{x}}}(\rho_0)\|_1 \| {\rm SWAP_{k_1,k_2}} \otimes \mathbb{1}_{\bar{k}_1, \bar{k}_2}  \|_\infty\\
    \leq & \left|\mathbb{E}_{V \sim \rm Haar} \Tr[(V\otimes V)(\sigma \otimes \sigma)(V^\dagger\otimes V^\dagger)  ({\rm SWAP_{k_1,k_2}} \otimes \mathbb{1}_{\bar{k}_1, \bar{k}_2})] \right|  + \| \mathcal{A}_{\mathbb{U}_{\vec{x}}}(\rho_0)\|_1  \\
    = & \mathbb{E}_{V \sim \rm Haar} \| \Tr_{\bar k}[V \sigma V^\dagger] \|_2^2  + \varepsilon_{\mathbb{U}_{\vec{x}}}  \label{eq:proof-thm1-pqk}\;,
\end{align}
where the indices $k$ and $\bar k$ represent the qubit $k$ and the rest of the system excluding $k$ respectively. Further, we introduce $k_1,k_2$ and $\bar{k}_1, \bar{k}_2$ as two copies of such subsystems. 
The second equality comes from using SWAP trick where we denote ${\rm SWAP}_{k_1,k_2}$ as the SWAP operator between $k_1$ and $k_2$, 
in the third equality we denote $\sigma = \rho_0 - \mathbb{1}/2^n$, in the fourth equality we substitute in the expressivity measure $\mathcal{A}_{\mathbb{U}_{\vec{x}}}(\sigma)$ and we note that $\mathcal{A}_{\mathbb{U}_{\vec{x}}}(\sigma) = \mathcal{A}_{\mathbb{U}_{\vec{x}}}(\rho_0)$. In addition, the first inequality is due to $s - t \leq |s| + |t|$,
the second inequality comes from applying the triangle inequality followed by H\"older's inequality to the second term. Finally, in the last inequality we upper bound the second term using the fact that ${\rm SWAP_{k_1,k_2}} \otimes \mathbb{1}_{\bar{k}_1, \bar{k}_2}$ has eigenvalues $\pm 1$, we reverse the SWAP trick on the first term, and we recall that $\varepsilon_{\mathbb{U}_{\vec{x}}} = \| \mathcal{A}_{\mathbb{U}_{\vec{x}}}(\rho_0)\|_1$. 

Next, we evaluate the Haar integration in the first term of~\eqref{eq:proof-thm1-pqk}.
\begin{align}
    \mathbb{E}_{V \sim \rm Haar} \| \Tr_{\bar k}[V \sigma V^\dagger] \|_2^2 & = \mathbb{E}_{V \sim \rm Haar} \left\| \Tr_{\bar k}[V \rho_0 V^\dagger] - \frac{\mathbb{1}_k}{2} \right\|_2^2 \\
    & = \mathbb{E}_{V \sim \rm Haar} \Tr_k \left[ \left( \Tr_{\bar k}[V \rho_0 V^\dagger]  \right)^2 \right] - \frac{1}{2} \\ 
    & = \mathbb{E}_{V \sim \rm Haar} \Tr \left[  \left(V \rho_0  V^\dagger\otimes V \rho_0 V^\dagger \right) ({\rm SWAP}_{k_1,k_2} \otimes \mathbb{1}_{\bar{k}_1,\bar{k}_2})\right]  - \frac{1}{2} \;, \label{eq:express-proof-beta2-1}
\end{align}

where in the first equality we substitute back $\sigma = \rho_0 - \mathbb{1}/2^n$, the second equality is from explicitly expanding the 2-norm and the third equality is due to the $\rm SWAP$ trick. Due to linearity, $\mathbb{E}_{V \sim \rm Haar}$ can be moved inside the trace and we further use the standard Haar integral $\mathbb{E}_{V \sim \rm Haar} [(V\rho_0 V^\dagger)\otimes (V\rho_0 V^\dagger)] = \frac{\mathbb{1} \otimes \mathbb{1} + {\rm SWAP}}{2^n(2^n+1)}$ and the fact that $\rho_0$ is pure (see for example Eq.~(2.26) in~\cite{roberts2017chaos}), leading to
\begin{align}
    \mathbb{E}_{V \sim \rm Haar} \Tr \left[  \left(V \rho_0  V^\dagger\otimes V \rho_0 V^\dagger \right) ({\rm SWAP}_{k_1,k_2} \otimes \mathbb{1}_{\bar{k}_1,\bar{k}_2})\right] & = 
    \Tr\left[ \left(\frac{ \mathbb{1} \otimes \mathbb{1} + {\rm SWAP}}{2^n(2^n+1)}\right){\rm SWAP_{k_1,k_2}} \otimes \mathbb{1}_{\bar{k}_1,\bar{k}_2}\right] \\
    & = \frac{2^{2(n-1)}\Tr[(\mathbb{1}_{k_1} \otimes \mathbb{1}_{k_2}){\rm SWAP}_{k_1,k_2}] + 2^{n-1} \Tr[{\rm SWAP}_{k_1,k_2} ^2] }{2^n(2^n + 1)} \\
    & = \frac{2^{n-1}+2}{2^n+1} \label{eq:express-proof-beta2-2},
\end{align}
where in the last line we have used the fact that $\mathrm{SWAP}^2=\mathbb{1}$. By substituting Eq.~\eqref{eq:express-proof-beta2-2} back into Eq.~\eqref{eq:express-proof-beta2-1}, we have $ \mathbb{E}_{V \sim \rm Haar} \| \Tr_{\bar k}[V \sigma V^\dagger] \|_2^2 = \tilde{\beta}_{\rm Haar} = \frac{3}{2^{n+1}+2}$
Altogether, we can now upper bound the variance in~\eqref{eq:proof-thm1-pqk-var} as 
\begin{align}
    \Var_{\vec{x},\vec{x'}}[k^{PQ}(\vec{x},\vec{x'})] \leq 4 \gamma n (\tilde{\beta}_{\rm Haar} + \varepsilon_{\mathbb{U}_{\vec{x}}}) \;.
\end{align}
Upon using Chebyshev's inequality, we complete the proof. 

\end{proof}

\subsection{Extensions of Theorem~\ref{thm:expressivity-kernel} to different input distributions}\label{appendix:extension-expressivity}
In Theorem~\ref{thm:expressivity-kernel}, we assume that both $\vec{x}$ and $\vec{x'}$ are averaged over all possible input data, implying that they are drawn from the same distribution. In this section, we
relax this assumption and consider a scenario where $\vec{x}$ and $\vec{x'}$ are drawn from different distributions leading to different data-embedded unitary ensembles $\mathbb{U}_{\vec{x}}$ and $\mathbb{U}_{\vec{x'}}$. We still observe kernel concentration in the same form as in~\eqref{eq:thm1-concentration-express-mt} of Theorem~\ref{thm:expressivity-kernel} but with modified values of $G_n(\varepsilon_{\mathbb{U}_{\vec{x}}}, \varepsilon_{\mathbb{U}_{\vec{x'}}})$ where $\varepsilon_{\mathbb{U}_{\vec{x}}}$ and $\varepsilon_{\mathbb{U}_{\vec{x'}}}$ are expressivity measures averaging over $\vec{x}$ and $\vec{x'}$.
\begin{enumerate}
    \item For the fidelity quantum kernel, $G_n(\varepsilon_{\mathbb{U}_{\vec{x}}}, \varepsilon_{\mathbb{U}_{\vec{x'}}}) = \beta_{\rm Haar} + \varepsilon_{\mathbb{U}_{\vec{x}}} \varepsilon_{\mathbb{U}_{\vec{x'}}} + \sqrt{\beta_{\rm Haar}}(\varepsilon_{\mathbb{U}_{\vec{x}}} + \varepsilon_{\mathbb{U}_{\vec{x'}}})$. 
    \item For the projected quantum kernel, $G_n(\varepsilon_{\mathbb{U}_{\vec{x}}}, \varepsilon_{\mathbb{U}_{\vec{x'}}}) = 2\gamma n (2 \tilde{\beta}_{\rm Haar} +  \varepsilon_{\mathbb{U}_{\vec{x}}} + \varepsilon_{\mathbb{U}_{\vec{x'}}})$. 
\end{enumerate}

\begin{proof}
First, consider the fidelity quantum kernel. We revisit~\eqref{eq:proof-thm1-express} in the proof of Theorem~\ref{thm:expressivity-kernel}. 
\begin{align}
    \Var_{\vec{x},\vec{x'}}[\kappa^{FQ}(\vec{x},\vec{x'})] & \leq \Tr[\int dU(\vec{x})(U(\vec{x}))^{\otimes 2} \rho^{\otimes 2}_0 (U^{\dagger}(\vec{x}))^{\otimes 2} \int dU(\vec{x'})(U(\vec{x'}))^{\otimes 2} \rho^{\otimes 2}_0 (U^{\dagger}(\vec{x'}))^{\otimes 2} ] \\
    & = \Tr[ (\mathcal{V}_{{\rm Haar}}(\rho_0) -  \mathcal{A}_{\mathbb{U}_{\vec{x}}}(\rho_0) ) (\mathcal{V}_{{\rm Haar}}(\rho_0) -  \mathcal{A}_{\mathbb{U}_{\vec{x'}}}(\rho_0) ) ] \\
    & = \beta_{\rm Haar} -  \Tr[ \mathcal{V}_{{\rm Haar}}(\rho_0)\mathcal{A}_{\mathbb{U}_{\vec{x}}}(\rho_0) ] -  \Tr[ \mathcal{V}_{{\rm Haar}}(\rho_0)\mathcal{A}_{\mathbb{U}_{\vec{x'}}}(\rho_0) ] + \Tr[\mathcal{A}_{\mathbb{U}_{\vec{x}}}(\rho_0)\mathcal{A}_{\mathbb{U}_{\vec{x'}}}(\rho_0) ] \; .
\end{align}
Similar to before, we rearrange the terms leading to
\begin{align}
     |  \Var_{\vec{x},\vec{x'}}[\kappa^{FQ}(\vec{x},\vec{x'})] - \beta_{\rm Haar} ] | & =  \left|\Tr[\mathcal{A}_{\mathbb{U}_{\vec{x}}}(\rho_0)\mathcal{A}_{\mathbb{U}_{\vec{x'}}}(\rho_0) ] -  \Tr[ \mathcal{V}_{{\rm Haar}}\mathcal{A}_{\mathbb{U}_{\vec{x}}}(\rho_0) ] -  \Tr[ \mathcal{V}_{{\rm Haar}}\mathcal{A}_{\mathbb{U}_{\vec{x'}}}(\rho_0) ] \right|  \\
     & \leq \left|\Tr[\mathcal{A}_{\mathbb{U}_{\vec{x}}}(\rho_0)\mathcal{A}_{\mathbb{U}_{\vec{x'}}}(\rho_0) ]\right| + \left|\Tr[ \mathcal{V}_{{\rm Haar}}\mathcal{A}_{\mathbb{U}_{\vec{x}}}(\rho_0) ]\right| + \left|\Tr[ \mathcal{V}_{{\rm Haar}}\mathcal{A}_{\mathbb{U}_{\vec{x'}}}(\rho_0) ] \right| \label{eq:express-different-class-1}\\
     & \leq \varepsilon_{\mathbb{U}_{\vec{x}}} \varepsilon_{\mathbb{U}_{\vec{x'}}} +\sqrt{\beta_{\rm Haar}} (\varepsilon_{\mathbb{U}_{\vec{x}}} + \varepsilon_{\mathbb{U}_{\vec{x'}}}) \label{eq:express-different-class-2}\;,
\end{align}
where the first inequality is from the triangle inequality and the second inequality due to H\"older's inequality and the monotonicity of the Schatten $p$-norms.
Hence, we have a bound for the variance as
\begin{align}
    \Var_{\vec{x},\vec{x'}}[\kappa^{FQ}(\vec{x},\vec{x'})] \leq \beta_{\rm Haar} + \varepsilon_{\mathbb{U}_{\vec{x}}} \varepsilon_{\mathbb{U}_{\vec{x'}}} + \sqrt{\beta_{\rm Haar}}(\varepsilon_{\mathbb{U}_{\vec{x}}} + \varepsilon_{\mathbb{U}_{\vec{x'}}}) \;.
\end{align}

For the projected quantum kernel, the bound of $ 2 \mathbb{E}_x \left\|  \rho_k(\vec{x})  - \mathbb{1}_k/2\right \|_2^2 $ remains unchanged as in~\eqref{eq:proof-thm1-pqk}.
However, when assembling terms together in the last step, we need to treat expressivity measures over $\vec{x}$ and $\vec{x'}$ to be different. 
The modification leads to
\begin{align}
    \Var_{\vec{x},\vec{x'}}[k^{PQ}(\vec{x},\vec{x'})] \leq 2 \gamma n (2\tilde{\beta}_{\rm Haar} + \varepsilon_{\mathbb{U}_{\vec{x}}} + \varepsilon_{\mathbb{U}_{\vec{x'}}}) \;,
\end{align}
which completes the proof.
\end{proof}

\section{Proof of Theorem~\ref{thm:entanglement-kernel}: Entanglement-induced concentration}\label{appendix:proof-entangled}
In this section, we provide a proof of Theorem~\ref{thm:entanglement-kernel}, describing the concentration of the kernel in terms of concentration of reduced states. The theorem is restated below for convenience.
\begin{theorem} [Entanglement-induced concentration]
Consider the projected quantum kernel as defined in Eq.~\eqref{eq:projected-gaussian-kernel-mt}. For a given pair of data-encoded states associated with $\vec{x}$ and $\vec{x'}$, we have
\begin{align}
    \left| 1 - \kappa^{PQ}(\vec{x},\vec{x'})\right|   \leq (2\ln2) \gamma \Gamma_s(\vec{x},\vec{x'}) \; ,
\end{align}
where
\begin{align}
    \Gamma_s(\vec{x},\vec{x'}) =  \sum_{k=1}^n \left[ \sqrt{S\left(\rho_k(\vec{x})\Big\|  \frac{\mathbb{1}_k}{2}\right)} +  \sqrt{S\left(\rho_k(\vec{x'})\Big\|  \frac{\mathbb{1}_k}{2}\right)} \right]^2  \;,
\end{align}
where we denote $S\left(\cdot\|  \cdot \right)$ as the quantum relative entropy, $\rho_k$ as a reduced state on qubit $k$, and $\mathbb{1}_k$ as the maximally mixed state on qubit $k$.
\end{theorem}
\begin{proof}
We consider the reduced state on a sub-system of $n_s$ qubits, which we denote as $\rho_s(\vec{x}) = \Tr_{\bar s}[\rho(\vec{x})]$ where $\Tr_{\bar s}[\cdot]$ is the partial trace over the rest of the system $\bar s$. We first remark that the trace distance and relative quantum entropy are related via Pinsker’s inequality as
\begin{align}
    \left\| \rho_s(\vec{x}) - \frac{\mathbb{1}_s}{2^{n_s}}\right\|_1^2 \leq 2 {\rm ln}2 \cdot S\left(\rho_s(\vec{x})\Big\|  \frac{\mathbb{1}_s}{2^{n_s}}\right) \;, \label{eq:proof-thm2-pqk-concentration-reduced-state1}
\end{align}
where $ S(\cdot \| \cdot)$ is the relative von Neumann entropy between two quantum states.

For a given pair of quantum data states $\rho(\vec{x}),  \rho(\vec{x'})$, we now look at a quantity $ \| \rho_s(\vec{x}) - \rho_s(\vec{x'}) \|_2^2$ which is a crucial ingredient to construct the projected quantum kernels (see Eq.~\eqref{eq:projected-gaussian-kernel-mt}). Consider the following bound:
\begin{align}
    \| \rho_s(\vec{x}) - \rho_s(\vec{x'}) \|_2 & \leq \| \rho_s(\vec{x}) - \rho_s(\vec{x'}) \|_1  \\ 
    & =\left \| \left( \rho_s(\vec{x})  - \frac{\mathbb{1}_s}{2^{n_s}}\right) - \left(\rho_s(\vec{x'})  - \frac{\mathbb{1}_s}{2^{n_s}}\right) \right\|_1  \\
    & \leq  \left\|  \rho_s(\vec{x})  - \frac{\mathbb{1}_s}{2^{n_s}}\right \|_1 + \left\|  \rho_s(\vec{x'})  - \frac{\mathbb{1}_s}{2^{n_s}}\right \|_1  \\
    & \leq \sqrt{2 {\rm ln}2} \left( \sqrt{S\left(\rho_s(\vec{x})\Big\|  \frac{\mathbb{1}_s}{2^{n_s}}\right)} +  \sqrt{S\left(\rho_s(\vec{x'})\Big\|  \frac{\mathbb{1}_s}{2^{n_s}}\right)} \right) \label{eq:proof-thm2-pqk-concentration-reduced-state2}\;,
\end{align}
where the first inequality comes from the monotonicity of Schatten $p$ norms, the second inequality is due to the triangle inequality and the last inequality is from the inequality in Eq.~\eqref{eq:proof-thm2-pqk-concentration-reduced-state1}.

For $n_s=1$ as in the projected quantum kernel, we can upper bound $\left| 1 - k^{PQ}(\vec{x},\vec{x'})\right| $ as 
\begin{align}
    \left| 1 - k^{PQ}(\vec{x},\vec{x'})\right|  & =  \left| 1 -e^{-\gamma \sum_{k=1}^n \|\rho_k(\vec{x}) - \rho_k(\vec{x'}) \|_2^2}\right| \\
    & \leq \left| \gamma \sum_{k=1}^n \|\rho_k(\vec{x}) - \rho_k(\vec{x'})\|_2^2\right|  \\
    & \leq (2 \ln 2) \gamma  \sum_{k=1}^n \left[ \sqrt{S\left(\rho_k(\vec{x})\Big\|  \frac{\mathbb{1}_k}{2}\right)} +  \sqrt{S\left(\rho_k(\vec{x'})\Big\|  \frac{\mathbb{1}_k}{2}\right)} \right]^2 
\end{align}
where we use $1 - e^{-t} \leq t $ in the first inequality and the second inequality follows from using the inequality in Eq.~\eqref{eq:proof-thm2-pqk-concentration-reduced-state2}.
\end{proof}

\section{Proof of Proposition~\ref{prop:global-measurement}: Global-measurement-induced concentration}

We restate Proposition~\ref{prop:global-measurement} here for convenience, which describes a model of concentration due to global measurement.

\begin{proposition}[Global-measurement-induced concentration]
Consider the fidelity quantum kernel as defined in Eq.~\eqref{eq:fidelity-kernel-mt} where the data embedding is of the form $U(\vec{x}) = \bigotimes_{k=1}^n U_k(x_k)$ with $x_k$ being an input component encoded in the qubit $k$, and $U_k$ being a single-qubit rotation about the $y$-axis on the $k$-th qubit. For an input data point $\vec{x}$, assume that all components of $\vec{x}$ are independent and uniformly sampled in $[-\pi,\pi]$. Given a product initial state  $\rho_0 = \bigotimes_{k=1}^n \dya{0}$, we have,
\begin{align}
    {\rm Pr}_{\vec{x},\vec{x'}}[|\kappa^{FQ}(\vec{x},\vec{x'}) - \sth{1/2^n}  | \geq \delta] \leq \left(\frac{3}{8}\right)^n \cdot \frac{1}{\delta^2} \;. 
\end{align}
\end{proposition}

\begin{proof}
Similar to the proof of Theorem~\ref{thm:expressivity-kernel}, we upper bound the variance of the kernel over the input data and then use Chebyshev's inequality to obtain the concentration bound. The difference here is that we specify the form of the data-embedding as $U(\vec{x}) = \bigotimes_{k=1}^n U_k(x_k)$ and the initial state as $\rho_0 = \bigotimes_{k=1}^n \rho_0^{(k)}$. Now consider the following:
\begin{align}
    \Var_{\vec{x},\vec{x'}}[\kappa^{FQ}(\vec{x},\vec{x'})] 
    & \leq \mathbb{E}_{\vec{x},\vec{x'}} [(\kappa^{FQ}(\vec{x},\vec{x'}))^2] \\
    & = \int dU(\vec{x}) \int dU(\vec{x'}) \Tr[U(\vec{x}) \rho_0 U^\dagger(\vec{x}) U(\vec{x'}) \rho_0 U^\dagger(\vec{x'})] \Tr[U(\vec{x}) \rho_0 U^\dagger(\vec{x}) U(\vec{x'}) \rho_0 U^\dagger(\vec{x'})] \\
    & = \int dU(\vec{x}) \int dU(\vec{x'})\prod_{k=1}^n \Tr[U_k(x_k) \rho_0^{(k)} U^\dagger_k(x_k)U_k(x'_k) \rho_0^{(k)} U^\dagger_k(x'_k)] \nonumber \\ 
    & \qquad\qquad\qquad\qquad\qquad\qquad\qquad\qquad \times \prod_{k=1}^n \Tr[U_k(x_k) \rho_0^{(k)} U^\dagger_k(x_k)U_k(x'_k) \rho_0^{(k)} U^\dagger_k(x'_k)] \\
    & = \int dU(\vec{x}) \int dU(\vec{x'}) \prod_{k=1}^n \Tr[(U_k(x_k))^{\otimes 2} (\rho_0^{(k)})^{\otimes 2}(U_k^\dagger(x_k))^{\otimes 2}(U_k(x'_k))^{\otimes 2} (\rho_0^{(k)})^{\otimes 2}(U_k^\dagger(x'_k))^{\otimes 2}] \\
    & =   \prod_{k=1}^n \Tr[\int dU_k(x_k)(U_k(x_k))^{\otimes 2} (\rho_0^{(k)})^{\otimes 2}(U_k^\dagger(x_k))^{\otimes 2}\int dU_k(x'_k)(U_k(x'_k))^{\otimes 2} (\rho_0^{(k)})^{\otimes 2}(U_k^\dagger(x'_k))^{\otimes 2}] \\
    & =   \left(\frac{3}{8}\right)^n\,.
\end{align}
In the third equality we have used the fact that since the components of $\vec{x}$ are independently sampled, then $\int dU(\vec{x})=\prod \int dU_k(x_k)$. Then, in the fourth inequality we have used the following result, which can be verified to  hold for $U_k(x_k)=e^{-i x_k Y}$, and for $\rho_0^{(k)}=\dya{0}$ via a direct computation:
\begin{equation}
    \Tr[\int dU_k(x_k)(U_k(x_k))^{\otimes 2} (\rho_0^{(k)})^{\otimes 2}(U_k^\dagger(x_k))^{\otimes 2}\int dU_k(x'_k)(U_k(x'_k))^{\otimes 2} (\rho_0^{(k)})^{\otimes 2}(U_k^\dagger(x'_k))^{\otimes 2}]=\frac{3}{8}\,.
\end{equation}
\sth{In addition, we can show that the concentration point becomes exponentially small with the number of qubits.
\begin{align}
    \Ebb_{\vec{x},\vec{x'}}[\kappa^{\rm FQ}(\vec{x},\vec{x'})] & =  \prod_{k=1}^n \int dU_k(x_k) \int dU_k(x'_k)U_k(x'_k) \Tr[U_k(x_k) \rho_0^{(k)}U_k^\dagger(x_k)  \rho_0^{(k)}U_k^\dagger(x'_k)] \\
    & = \frac{1}{2^n} \;.
\end{align}
where each term in the product is evaluated to be $1/2$ with $U_k(x_k)=e^{-i x_k Y}$ and $\rho_0^{(k)}=\dya{0}$.}
\end{proof}

\subsection{Extension to arbitrary local unitaries}

We can further generalize the previous proposition to the case $U_k$ is a general unitary. Now, the following result holds.

\begin{supplemental_proposition}[Generalized global-measurement-induced concentration]
Consider the fidelity quantum kernel as defined in Eq.~\eqref{eq:fidelity-kernel-mt} where the data embedding is of the form $U(\vec{x}) = \bigotimes_{k=1}^n U_k(x_k)$ with $x_k$ is the $k^{th}$ element of $\vec{x}$ encoded in the qubit $k$. For an input data point $\vec{x}$, assume that all components of $\vec{x}$ are independent. Given that the initial state is a product state $\rho_0 = \bigotimes_{k=1}^n \rho_0^{(k)}$, we have
\begin{align}
    {\rm Pr}_{\vec{x},\vec{x'}}[|\kappa^{FQ}(\vec{x},\vec{x'}) - \mu  | \geq \delta] \leq \frac{\prod_{k=1}^n G^{(k)}_1(\varepsilon_{\mathbb{U}_{x_k}})}{\delta^2} \;, 
\end{align}
where $\mu = \mathbb{E}_{\vec{x},\vec{x'}}[\kappa^{FQ}(\vec{x},\vec{x'})]$ and
\begin{align}
   G^{(k)}_1(\varepsilon_{\mathbb{U}_{x_k}}) = \frac{1}{3} + \varepsilon_{\mathbb{U}_{x_k}}\left(\varepsilon_{\mathbb{U}_{x_k}} + \sqrt{\frac{4}{3}}\right) \;. 
\end{align}
Here, $\varepsilon_{\mathbb{U}_{x_k}} = \left\|\mathcal{A}_{\mathbb{U}_{x_k}}\left(\rho_0^{(k)}\right) \right\|_1$ is a data-dependent local expressivity measure of the local unitary $U_k(x_k)$ over all possible values of $x_k$ encoded in qubit $k$, where $\mathcal{A}_{\mathbb{U}}(\cdot)$ is defined in Eq.~\eqref{eq:expressivity-measure-mt}. 
\end{supplemental_proposition}

\begin{proof}
Similar to the proof of Theorem~\ref{thm:expressivity-kernel}, we upper bound the variance of the kernel over the input data and then use Chebyshev's inequality to obtain the concentration bound. The difference here is that we specify the form of the data-embedding as $U(\vec{x}) = \bigotimes_{k=1}^n U_k(x_k)$ and the initial state as $\rho_0 = \bigotimes_{k=1}^n \rho_0^{(k)}$. Now consider the following:
\begin{align}
    \Var_{\vec{x},\vec{x'}}[\kappa^{FQ}(\vec{x},\vec{x'})] 
    & \leq \mathbb{E}_{\vec{x},\vec{x'}} [(\kappa^{FQ}(\vec{x},\vec{x'}))^2] \\
    & = \int dU(\vec{x}) \int dU(\vec{x'}) \Tr[U(\vec{x}) \rho_0 U^\dagger(\vec{x}) U(\vec{x'}) \rho_0 U^\dagger(\vec{x'})] \Tr[U(\vec{x}) \rho_0 U^\dagger(\vec{x}) U(\vec{x'}) \rho_0 U^\dagger(\vec{x'})] \\
    & = \int dU(\vec{x}) \int dU(\vec{x'})\prod_{k=1}^n \Tr[U_k(x_k) \rho_0^{(k)} U^\dagger_k(x_k)U_k(x'_k) \rho_0^{(k)} U^\dagger_k(x'_k)] \nonumber \\ 
    & \qquad\qquad\qquad\qquad\qquad\qquad\qquad\qquad \times \prod_{k=1}^n \Tr[U_k(x_k) \rho_0^{(k)} U^\dagger_k(x_k)U_k(x'_k) \rho_0^{(k)} U^\dagger_k(x'_k)] \\
    & = \int dU(\vec{x}) \int dU(\vec{x'}) \prod_{k=1}^n \Tr[(U_k(x_k))^{\otimes 2} (\rho_0^{(k)})^{\otimes 2}(U_k^\dagger(x_k))^{\otimes 2}(U_k(x'_k))^{\otimes 2} (\rho_0^{(k)})^{\otimes 2}(U_k^\dagger(x'_k))^{\otimes 2}] \\
    & =  \prod_{k=1}^n \Tr[\int dU_k(x_k)(U_k(x_k))^{\otimes 2} (\rho_0^{(k)})^{\otimes 2}(U_k^\dagger(x_k))^{\otimes 2}\int dU_k(x'_k)(U_k(x'_k))^{\otimes 2} (\rho_0^{(k)})^{\otimes 2}(U_k^\dagger(x'_k))^{\otimes 2}] \\
    & = \prod_{k=1}^n \Tr[ \left(\mathcal{V}_{{\rm Haar}}(\rho_0^{(i)}) - \AC_{\Ubb_{x_k}}(\rho_0^{(i)}) \right)^2 ] \\
    & = \prod_{k=1}^n \left( \Tr[\left(\VC_{\rm Haar} (\rho_0^{(i)})\right)^2] + \Tr[  \AC_{\Ubb_{x_k}}(\rho_0^{(i)})\left( \AC_{\Ubb_{x_k}}(\rho_0^{(i)}) - 2 \VC_{{\rm Haar}}(\rho_0^{(i)})\right)] \right) \\
    & \leq \prod_{k=1}^n \left( \frac{1}{3} + \left| \Tr[  \AC_{\Ubb_{x_k}}(\rho_0^{(i)})\left( \AC_{\Ubb_{x_k}}(\rho_0^{(i)}) - 2 \VC_{{\rm Haar}}(\rho_0^{(i)})\right)]\right| \right)  \\
    & \leq  \prod_{k=1}^n \left( \frac{1}{3}  + \left\|\AC_{\Ubb_{x_k}}(\rho_0^{(i)}) \right\|_2 \left\|\AC_{\Ubb_{x_k}}(\rho_0^{(i)})  - 2 \VC_{{\rm Haar}}(\rho_0^{(i)})\right\|_2 \right) \label{eq:proof-global-1}\\
    & \leq \prod_{k=1}^n\left[\frac{1}{3} + \varepsilon_{\mathbb{U}_{x_k}}\left(\varepsilon_{\mathbb{U}_{x_k}} + \sqrt{\frac{4}{3}}\right) \right]\label{eq:proof-global-2}
\end{align}
where the second equality comes from substituting $U(\vec{x}) = \bigotimes_{k=1}^n U_k(x_k)$ and $\rho_0 = \bigotimes_{k=1}^n \rho_0^{(k)}$ followed by using the trace property $\Tr[X\otimes Y] =\Tr[X]\Tr[Y]$, the fourth equality is due to the assumption that all components of $x$ and $x'$ are independent, the fifth equality is due to the assumption that $\vec{x}$ and $\vec{x'}$ are drawn from the same distribution and the use of the definition of the local superoperator $\AC_{\Ubb_{x_k}}\left(\rho_0^{(i)}\right)$. In addition, we note that $\Tr[\left(\VC_{\rm Haar} (\rho_0^{(i)})\right)^2] = \frac{1}{3}$. The inequalities~\eqref{eq:proof-global-1} and \eqref{eq:proof-global-2} follow the same steps as \eqref{eq:proof-thm1-first-fqk} to \eqref{eq:proof-thm1-last-fqk} in the proof of Theorem~\ref{thm:expressivity-kernel}. That is, we apply the triangle inequality followed by H\"older's inequality in~\eqref{eq:proof-global-1} and we use the monotonicity of Schatten $p$-norm in~\eqref{eq:proof-global-2}. In the last step, we also substitute in $\varepsilon_{\Ubb_{x_k}} = \left\| \AC_{\Ubb_{x_k}}(\rho_0^{(i)})\right\|_1$. With this upper bound of the variance, we invoke Chebyshev's inequality, leading to our desired result.

\end{proof}
In the limit where all single-qubit unitaries are random (i.e. $\varepsilon_{U_k} = 0 \; \forall k$), the upper bound in~\eqref{eq:prop1} takes the value $1/3^n$ and therefore the kernel exponentially concentrates probabilistically. 

\section{Proof of Theorem~\ref{thm:noise-kernel}: Noise-induced concentration}\label{appendix:proof-noise}
In this section, we prove Theorem~\ref{thm:noise-kernel} which formally establishes how noise leads to the concentration of quantum kernels. 
We first note that a quantum state of $n$ qubits can be expressed in the Pauli basis as
\begin{align}\label{eq:noise-initial-state}
    \rho & = \frac{1}{2^n} (\mathbb{1} + \sum_i a_i \sigma_i) \\
    & = \frac{1}{2^n} (\mathbb{1}  + \vec{a}\cdot\vec{\sigma}) \;,
\end{align}
where $a_i$ is a coefficient associated with a Pauli operator $\sigma_i \in \{\mathbb{1}, X, Y, Z \}^{\otimes n}/\{\mathbb{1}^{\otimes n}\}$. Correspondingly, $\vec{a}$ is a vector of such coefficients and $\vec{\sigma}$ is a vector of such Pauli operators. We now provide three lemmas describing the evolution of quantum states under unitary transformations and noise channels.

\begin{lemma} [Pauli coefficients under unitary transformations] \label{lemma:noise-unitary}
Consider the Pauli decomposition of a state $\rho$ that takes the form in Eq.~\eqref{eq:noise-initial-state}. $\| \vec{a} \cdot \vec{\sigma} \|_p$ is invariant under the unitary transformation $\rho \rightarrow U\rho U^\dagger$.
\end{lemma}

\begin{proof}
The invariance under the transformation is a direct consequence of the linearity of unitary transformations and the unitary invariance of Schatten norms.
\end{proof}

\begin{lemma}[Pauli coefficients under noise channels]\label{lemma:noise-noise}
Consider the Pauli coefficients of a state $\rho$ that takes the form in Eq.~\eqref{eq:noise-initial-state} under the action of the local Pauli noise channel $\mathcal{N} = \mathcal{N}_1 \otimes ... \otimes \mathcal{N}_n $ where each $\mathcal{N}_j$ acts on qubit $j$ according to Eq.~\eqref{eq:noise-noise-model}. Then, we have
\begin{align}
    \| \vec{a'}\cdot\vec{\sigma} \|_2 \leq q  \| \vec{a}\cdot\vec{\sigma} \|_2 \;,
\end{align}
where $\vec{a'}$ are the new Pauli coefficients after the action of noise.
\end{lemma}
\begin{proof}
We have
\begin{align}
    \left\| \vec{a'}\cdot \vec{\sigma} \right\|_2 & = \left\| \mathcal{N}(\vec{a}\cdot\vec{\sigma}) \right\|_2  \\
    & = \left\|\mathcal{N}\left(  \sum_i a_i \sigma_i \right) \right\|_2 \\
    & = \left\|\sum_i a_i q_X^{x(i)}q_Y^{y(i)}q_Z^{z(i)} \sigma_i \right\|_2  \label{eq:noise-proof-lemma1} \\
    &  \leq \left\|\sum_i a_i q^{x(i) + y(i) + z(i)} \sigma_i \right\|_2  \\
    & \leq q \left\|\vec{a}\cdot\vec{\sigma} \right\|_2 \; ,
\end{align}
where, in Eq.~\eqref{eq:noise-proof-lemma1} we use the fact that $\mathcal{N}(\sigma_i) = q_X^{x(i)}q_Y^{y(i)}q_Z^{z(i)} \sigma_i$ with $x(i), y(i), z(i)$ being the number of respective single-qubit $X, Y, Z$ Pauli operators that appear in the Pauli string $\sigma_i$, the first inequality comes from replacing the coefficients with the noise parameter as defined in Eq.~\eqref{eq:noise-charac}, and in the final inequality we use the fact that there is at least one non-identity single-qubit Pauli term in $\sigma_i$ i.e. $x(i)+y(i)+z(i) \geq 1$ (recall that $\sigma_i \in \{\mathbb{1}, X, Y, Z \}^{\otimes n}/\{\mathbb{1}^{\otimes n}\}$). 
\end{proof}

\begin{lemma}[Supplemental Lemma 6 from Ref.~\cite{wang2020noise}, adapted]\label{lemma:noise-pauli-entropy} Consider a quantum state $\rho$ under the action of the local Pauli noise channel $\mathcal{N} = \mathcal{N}_1 \otimes ... \otimes \mathcal{N}_n $ where each $\mathcal{N}_j$ acts on qubit $j$ according to Eq.~\eqref{eq:noise-noise-model}. Then, we have
\begin{align}
    S_2\left(\mathcal{N}(\rho)\Big\|  \frac{\mathbb{1}}{2^{n}}\right) \leq q^{b} S_2\left(\rho \Big\|  \frac{\mathbb{1}}{2^{n}}\right) \;, 
\end{align}
where $S_2(\cdot\|\cdot)$ is the sandwiched 2-R\'enyi relative entropy and $b = 1/(2\ln(2)) \approx 0.72$.
\end{lemma}

\begin{proof}
This result is a direct consequence of Corollary 5.6 of Ref.~\cite{hirche2020contraction}. Let us first restate the general result for convenience: For some density operator $\gamma$ and probability $p>0$ consider the channel $\AC_{p,\gamma}(\cdot)=p (\cdot)+(1-p) \gamma$. Suppose that some other channel $\BC$ satisfies
\begin{equation}
    \left\|\Gamma_{\BC(\gamma)}^{-\frac{1}{2}} \circ \BC \circ \AC_{p,\gamma}^{-1} \circ \Gamma_{\gamma}^{\frac{1}{2}}\right\|_{2 \rightarrow 2} \leq 1\,
\end{equation}
where $\AC_{p,\gamma}^{-1}$ denotes the inverse map of $\AC_{p,\gamma}$ and $\Gamma_{\gamma}^{p}$ denotes the map $\Gamma_{\gamma}^{p}(\cdot)=\gamma^{\frac{p}{2}} (\cdot) \gamma^{\frac{p}{2}}$. Then, for all states $\rho$, 
\begin{equation}\label{eq:2renyi}
    S_{2}\!\left(\BC^{\otimes n}(\rho) \| \BC^{\otimes n}\left(\gamma^{\otimes n}\right)\right) \leq \alpha(p, \gamma) S_{2}\!\left(\rho \| \gamma^{\otimes n}\right)
\end{equation}
where $\alpha(p, \gamma)=\exp \left(\left(1-\left\|\gamma^{-1}\right\|_\infty^{-1}\right) \frac{\ln (p)}{\ln \left(\left\|\gamma^{-1}\right\|_\infty\right)}\right)$. Now, in our case, we consider $\AC_{p,\gamma}$ and $\BC$ to act on a single qubit. Then, if one chooses $\AC_{p,\gamma}$ to be the single qubit depolarizing channel $\DC_{p_d}$ with depolarizing probability $p_d$ and maximally mixed fixed point $\gamma = \frac{\mathbb{1}}{2}$, then \eqref{eq:2renyi} implies that if some unital qubit channel $\BC$ (which acts trivially on the identity) satisfies
\begin{equation}\label{eq:ball}
    \left\|\BC \circ \DC_{p_d}^{-1} \right\|_{2 \rightarrow 2} \leq 1\,.
\end{equation}
From the previous, we have  for any $n$-qubit state $\rho$ 
\begin{align}
    S_{2}\!\left(\BC^{\otimes n}(\rho) \Big\| \frac{{\mathbb{1}}^{\otimes n}}{2^n}\right) & = \alpha\big((1-p_d),\mathbb{1}/2\big) S_{2}\!\left(\rho \Big\| \frac{{\mathbb{1}}^{\otimes n}}{2^n}\right) \\
    & \leq (1-p_d)^b S_{2}\!\left(\rho \Big\| \frac{{\mathbb{1}}^{\otimes n}}{2^n}\right)\,,\label{eq:2renyi2}
\end{align}
where we denote $b = 1/(2\ln(2)) \approx 0.72$.

Now suppose that $\BC$ is the single-qubit Pauli noise channel $\NC_i$ as defined in \eqref{eq:noise-noise-model}. We can explicitly write the condition \eqref{eq:ball} as
\begin{equation}\label{eq:ball2}
    \sup_{X\neq 0} \frac{\|\NC_i\circ\DC_{p_d}^{-1}(X) \|_2}{\|X \|_2} \leq 1.
\end{equation}
We note that the superoperator (Pauli transfer matrix) of the concatenated channel $\NC_i\circ\DC_{p_d}^{-1}$ is diagonal with diagonal entries $(1,\frac{q_x}{1-p_d},\frac{q_y}{1-p_d},\frac{q_z}{1-p_d})$. Consider an arbitrary complex matrix $X$ decomposed in the Pauli basis as $X = a\mathbb{1} + \Vec{b}\cdot\Vec{\sigma}$, where $a$ is a complex number and $\Vec{b}$ is a vector of complex coefficients. Then one can verify 
\begin{align}
    \|X \|_2 &= \sqrt{2}\sqrt{|a|^2 + \textstyle\sum_i |b_i|^2}\,, \\
    \|\NC_i\circ\DC_p^{-1} (&X) \|_2 = \sqrt{2}\sqrt{|a|^2 + \textstyle\sum_i\left(\frac{q_i}{1-p_d}\right)^2|b_i|^2} \,,
\end{align}
where the second expression is obtained by reading off the diagonal entries of the superoperator of $\NC_i\circ\DC_{p_d}^{-1}$. In order to satisfy condition \eqref{eq:ball2}, one can pick
\begin{equation}
    1-p_d=\max_{\sigma \in \{X,Y,Z\}} |q_{\sigma}|\,.
\end{equation}
Thus, by denoting $q = {\max_{\sigma \in \{X,Y,Z\}} |q_\sigma|}$ and inspecting \eqref{eq:2renyi2} we obtain the result as required.
\end{proof}

\bigskip

Now, we are ready to prove Theorem~\ref{thm:noise-kernel}, which is restated below for convenience. 

\begin{theorem}[Noise-induced concentration]
Consider the $L$-layered data embedding circuit defined in Eq.~\eqref{eq:noise-embedding-mt} with input state $\rho_0$ and the layer-wise Pauli noise model defined in Eq.~\eqref{eq:noise-noise-evolution-mt} with characteristic noise parameter $q < 1$. The concentration of quantum kernel values may be bounded as follows
\begin{align}
    \left| \tilde{\kappa}(\vec{x},\vec{x'}) - \mu \right| \leq F(q,L) \; .
\end{align}
\begin{enumerate}
    \item For the fidelity quantum kernel $\tilde{\kappa}(\vec{x},\vec{x'}) =\tilde{\kappa}^{FQ}(\vec{x},\vec{x'})$, we have $ \mu = 1/2^n$, and
    \begin{align}
        F(q,L) = q^{2L+1}  \left\| \rho_0 - \frac{\mathbb{1}}{2^n} \right\|_2\;.
    \end{align}
    \item For the projected quantum kernel $\tilde{\kappa}(\vec{x},\vec{x'}) =\tilde{\kappa}^{PQ}(\vec{x},\vec{x'})$, we have $ \mu = 1$, and
    \begin{align}
        F(q,L) = (8 \ln 2) \gamma n  q^{b(L+1)}S_2\left(\rho_0 \Big\|  \frac{\mathbb{1}}{2^{n}}\right) \;,
    \end{align}
    where $S_2( \cdot \| \cdot)$ denotes the sandwiched 2-R\'enyi relative entropy and $b = 1/(2\ln(2)) \approx 0.72 $.
\end{enumerate}
Additionally, the noisy data-encoded quantum state $\tilde{\rho}(\vec{x})$ concentrates towards the maximally mixed state as
\begin{align} 
     \left\| \tilde{\rho}(\vec{x}) - \frac{\mathbb{1}}{2^n} \right\|_2 \leq q^{L+1} \left\| \rho_0 - \frac{\mathbb{1}}{2^n} \right\|_2 \; .
\end{align}
\end{theorem}

\begin{proof}
First we prove the concentration of noisy quantum states toward the maximally mixed state, following Ref.~\cite{wang2020noise}. We express $\tilde{\rho}(\vec{x})$ explicitly in terms of its Pauli decomposition as $\tilde{\rho}(\vec{x}) =  \frac{1}{2^n} (\mathbb{1} + \vec{\tilde{a}}\cdot\vec{\sigma})$ where $\vec{\tilde{a}}$ are the coefficients after the noisy embedding in Eq.~\eqref{eq:noise-noise-evolution-mt}. Hence, we have
\begin{align}
    \left\|\tilde{\rho}(\vec{x}) -  \frac{\mathbb{1}}{2^n} \right\|_2 & = \left\| \frac{1}{2^n}  \vec{\tilde{a}}\cdot\vec{\sigma}\right\|_2\\
    & \leq q^{L+1} \left\| \frac{1}{2^n} \vec{a}\cdot\vec{\sigma}\right\|_2 \\
    & = q^{L+1} \left\| \rho_0 - \frac{\mathbb{1}}{2^n} \right\|_2 \;, \label{eq:noise-state-concentration}
\end{align}
where the inequality comes from repeatedly applying Lemma~\ref{lemma:noise-unitary} and Lemma~\ref{lemma:noise-noise} $L+1$ times. This completes the proof of the quantum state concentration.
Now we prove the concentration of quantum kernels. Similar to the proof of Theorem~\ref{thm:noise-kernel}, we separate the proof into two sub-sections for the fidelity and projected quantum kernels.

\bigskip\noindent
\underline{Fidelity quantum kernels:} Consider a noisy fidelity quantum kernel 
which can be expressed as 
\begin{align} 
    \tilde{\kappa}^{FQ}(\vec{x},\vec{x'}) & = \Tr[\tilde{\rho}(\vec{x})\tilde{\rho}(\vec{x'})] \\
    & = \Tr[\WC_{\vec{x}}(\rho_0) \WC_{\vec{x'}}(\rho_0)] \\
    & =  \Tr[\rho_0 \WC_{\vec{x},\vec{x'}}(\rho_0)]\;,\label{eq:noise-fqk}
\end{align}
where we have denoted the channel $\WC_{\vec{x},\vec{x'}} = \WC_{\vec{x}}^{\dag}\circ\WC_{\vec{x'}}$ which is composed of
\begin{align} \label{eq:noise-fqk-channel}
    \WC_{\vec{x},\vec{x'}} = \NC \circ \UC_1^\dagger(\vec{x}_1) \circ \NC \cdots \NC\circ  \UC_L^\dagger(\vec{x}_L) \circ \NC \circ \UC_L(\vec{x'}_L)\circ \NC \cdots \NC \circ \UC_1(\vec{x'}_1) \;,
\end{align}
where we have used the fact that the Pauli noise channel in Eq.~\eqref{eq:noise-noise-model} is self adjoint. Now, we show the concentration of the fidelity kernel. 
\begin{align}
    \left|\tilde{\kappa}^{FQ}(\vec{x},\vec{x'}) - \frac{1}{2^n}\right| & = \left| \Tr[\rho_0  \WC_{\vec{x},\vec{x'}}(\rho_0)]  - \frac{1}{2^n} \Tr\left[ \rho_0 \right]\right| \;  \\
    & = \left| \Tr[\rho_0 \left( \WC_{\vec{x},\vec{x'}}(\rho_0) - \frac{\mathbb{1}}{2^n}\right)] \right| \; \\
    & \leq \|\rho_0 \|_2 \left\| \WC_{\vec{x},\vec{x'}}(\rho_0) - \frac{\mathbb{1}}{2^n}  \right\|_2 \; \\
    & \leq q^{2L+1} \left\|\rho_0 - \frac{\mathbb{1}}{2^n}  \right\|_2 \; ,
\end{align}
where in the first line we express the noisy quantum kernel as in Eq.~\eqref{eq:noise-fqk} and also use $\Tr[\rho_0] = 1$, the first inequality is due to H\"older's inequality, and lastly the second inequality comes from using the fact that $\|\rho_0\|_2 \leq 1$ together with repeatedly applying Lemma~\ref{lemma:noise-unitary} and Lemma~\ref{lemma:noise-noise} for the noisy quantum channel $\WC_{\vec{x},\vec{x'}}$ in Eq.~\eqref{eq:noise-fqk-channel}. We note that this is similar to the proof of quantum state concentration but the number of instances of noise $\mathcal{N}$ is now $2L + 1$ as we want to implement $U^\dagger(\vec{x})U(\vec{x'})$ instead of $U(\vec{x})$.  

\bigskip\noindent
\underline{Projected quantum kernels:} Here we have
\begin{align}
   \left|1 - \tilde{\kappa}^{PQ}(\vec{x},\vec{x'})\right|  & =   \left| 1 -  e^{ - \gamma \sum_{k=1}^n \| \tilde{\rho}_k(\vec{x}) - \tilde{\rho}_k(\vec{x'})\|^2_2} \right| \\
    & \leq   \gamma \sum_{k=1}^n  \| \tilde{\rho}_k(\vec{x}) - \tilde{\rho}_k(\vec{x'})\|^2_2  \\
    & \leq \gamma \sum_{k=1}^n \left( \left\|\tilde{\rho}_k(\vec{x}) - \frac{\mathbb{1}_k}{2} \right\|_2 +  \left\|\tilde{\rho}_k(\vec{x'}) - \frac{\mathbb{1}_k}{2} \right\|_2 \right)^2 \\
    & \leq \gamma \sum_{k=1}^n \left( 2 \left\|\tilde{\rho}_k(\vec{x}) - \frac{\mathbb{1}_k}{2} \right\|_2^2 + 2 \left\|\tilde{\rho}_k(\vec{x'}) - \frac{\mathbb{1}_k}{2} \right\|_2^2 \right) \;,
\end{align}
where the first inequality is due to the standard inequality $1 - e^{-t} \leq t$, the second inequality is due to the triangle inequality, the third inequality is due to the fact that $(s+t)^2 \leq 2s^2 + 2t^2$. 
Note that the concentration of the reduced state $\tilde{\rho}_k(\vec{x})$ can be bounded as
\begin{align}
    \left\|  \tilde{\rho}_k(\vec{x})  - \frac{\mathbb{1}_k}{2}\right \|_2^2 
    &=  \Tr_k \left[ \Tr_{\bar k}\left( \tilde{\rho}(\vec{x}) - \mathbb{1}/2^n \right)\Tr_{\bar k}\left( \tilde{\rho}(\vec{x}) - \mathbb{1}/2^n \right)\right] \\
    &= \Tr[ (\tilde{\rho}(\vec{x}) - \mathbb{1}/2^n)\otimes (\tilde{\rho}(\vec{x}) - \mathbb{1}/2^n) ({\rm SWAP_{k_1,k_2}} \otimes \mathbb{1}_{\bar{k}_1, \bar{k}_2})] \\
    &\leq  \|(\tilde{\rho}(\vec{x}) - \mathbb{1}/2^n)\otimes (\tilde{\rho}(\vec{x}) - \mathbb{1}/2^n) \|_1 \|{\rm SWAP_{k_1,k_2}} \otimes \mathbb{1}_{\bar{k}_1, \bar{k}_2} \|_{\infty} \\
    &\leq \|\tilde{\rho}(\vec{x}) - \mathbb{1}/2^n \|_1^2 \\
    &\leq 2 \ln 2 \cdot S\left(\tilde{\rho}(\vec{x})\Big\|  \frac{\mathbb{1}}{2^{n}}\right) \\
    &\leq  2 \ln 2 \cdot S_2\left(\tilde{\rho}(\vec{x})\Big\|  \frac{\mathbb{1}}{2^{n}}\right) \\
    &\leq   (2 \ln 2 )q^{b(L+1)}  S_2\left(\rho_0 \Big\|  \frac{\mathbb{1}}{2^{n}}\right) 
\end{align}
where we use the SWAP trick in the second line with $\rm SWAP_{k_1,k_2}$ being the SWAP operator between two reduced subsystems,  
the first inequality comes from H\"older's inequality, in the second line we use the fact that ${\rm SWAP_{k_1,k_2}} \otimes \mathbb{1}_{\bar k_1, \bar k_2}$ has eigenvalues in $\{1,-1 \}$ and that $\| X \otimes Y \|_1 = \|X\|_1 \|Y\|_1$, the third inequality is due to Pinsker's inequality, the fourth inequality is due to the monotonicity of the sandwiched 2-R\'enyi relative entropy, and finally the last inequality is from repeatedly applying the data-processing inequality and Lemma~\ref{lemma:noise-pauli-entropy} for each layer of unitaries and noise.
Hence, we have the concentration bound of the projected quantum kernel as
\begin{align}
    \left|1 - \tilde{\kappa}^{PQ}(\vec{x},\vec{x'})\right| \leq (8 \ln 2) \gamma n  q^{b(L+1)}S_2\left(\rho_0 \Big\|  \frac{\mathbb{1}}{2^{n}}\right)  \;,
\end{align}
which completes the proof. 

\end{proof}

\section{Error Mitigation}\label{appendix:em-fail}

Error mitigation (EM) strategies have been widely implemented to reduce the effect of noise in variational quantum algorithms (VQAs) and QML.
Despite tremendous success to significantly suppress errors of expectation values, it has been recently shown that current common EM strategies cannot resolve the issue of noise-induced barren plateaus for QNNs~\cite{wang2021can, takagi2021fundamental, quek2022exponentially}. In particular, even after applying EM protocols, the cost landscape could remain exponentially flat, or otherwise exponential resources are required to reach a sufficiently high resolution of the expectation values. 
As estimating quantum kernels in practice requires us to measure expectation values of some operators, the results derived in Ref.~\cite{wang2021can} can be directly applied to the kernel framework, which we explain in more detail below. Consequently, EM strategies also fails to remove the exponential decay in the noise-induced kernel concentration.

Given that we are interested in a noise-free expectation value  $C = \Tr[ \rho O ]$ of some operator $O$ and $n$-qubit quantum state of interest $\rho$, the main purpose of EM strategies is to approximate $C$ under the effect of noise by implementing some protocol which gives us a noise-mitigated quantity $C_m$. Usually, an EM strategy includes one or more of the following protocols: running some modification of the initial circuit of interest, modifying the observable, utilizing multiple copies of the state of interest, performing classical post-processing. 
Most of well-known EM strategies can be grouped under a unified framework, which includes Zero-Noise Extrapolation~\cite{li2017efficient,temme2017error,endo2018practical,kandala2018error}, Clifford Data Regression~\cite{czarnik2020error}, Virtual Distillation~\cite{huggins2020virtual,koczor2020exponential} and Probabilistic Error Cancellation~\cite{temme2017error,endo2018practical}. Within this unified framework, we prepare expectation values of the form
\begin{align}\label{eq:appendix-em-term}
    E_{\sigma, A, M, k} = \Tr\left[ A\left( \sigma^{\otimes M}\otimes \dya{0}^{\otimes k} \right)\right] \;,
\end{align}
where we allow modifications to the original circuit leading to the state $\sigma$ instead of $\rho$, $M$ copies of $\sigma$ and $k$ ancillary qubits are allowed, and we measure some operator $A$ that is allowed to act on up to the entire composite system. 
Then, the noise-mitigated value $C_m$ can be expressed as a linear combination of $E_{\sigma, A, M, k}$ over different $\sigma, A, M, k$, as
\begin{align} \label{eq:appendix-em-general}
    C_m = \sum_{\sigma, A, M, k \in \mathcal{T}_{\rm EM}} a_{\sigma, A,M,k} E_{\sigma, A, M, k} \; ,
\end{align}
where $a_{\sigma, A,M,k}$ are chosen coefficients and $\mathcal{T}_{\rm EM}$ is a set containing all relevant indices for the considered EM strategy. As an example, consider Zero Noise Extrapolation (ZNE). In this strategy, the noise strength in the circuit is augmented leading to noisier expectation values, and then the error mitigated value is estimated via extrapolating back to the noiseless regime. Given two noisy expectation values $\tilde{C}(\epsilon_q), \tilde{C}(a\epsilon_q)$ with two different noise strengths $\epsilon_q$ and $a \epsilon_q$ for $a>1$, we can express $C_m$ using the first level of Richard extrapolation as
\begin{align}
    C_m = \frac{a \tilde{C}(\epsilon_q) - \tilde{C}(a\epsilon_q)}{a - 1}  \;,
\end{align}
which can serve as a better approximation of $C$ than $\tilde{C}(\epsilon_q)$. Note that this takes the form of the general expression in Eq.~\eqref{eq:appendix-em-general}. 
For more details about error mitigation, we refer the reader to Ref.~\cite{cerezo2020variationalreview,endo2021hybrid}.
We now quote one of the main results in Ref.~\cite{wang2021can} which is relevant to our work.

\begin{supplemental_theorem}[Theorem 1 and Corollary 1 in Ref.~\cite{wang2021can}]\label{sup-thm:em-noise}
Consider an error mitigation strategy that as a step in its protocol, estimates $E_{\sigma, A, M, k} $ as defined in  Eq.~\eqref{eq:appendix-em-term}. Suppose that $\sigma$ is prepared with a circuit consisting of $L_{\sigma}$ layers of gates with local depolarizing noise with depolarizing probability $p$ occurring before and after each layer. Under these conditions, $E_{\sigma, A, M, k} $ exponentially concentrates on a state-independent fixed point with increasing circuit depth as
\begin{align}
    \left|E_{\sigma, A, M, k} - \Tr\left[ A \left( \frac{\mathbb{1}^{\otimes M}}{2^{ Mn}}\otimes \ket{0}\bra{0}^{\otimes k}\right)\right]  \right| \leq \sqrt{\ln 4} \|A\|_\infty M n^{1/2} (1-p)^{L_{\sigma} + 1} \;.
\end{align}
In addition, consider a noise-mitigated expectation value $C_m$ constructed via Eq.~\eqref{eq:appendix-em-general}. 
Assume $ \|A\|_\infty  \in \mathcal{O}({\rm poly}(n))$ and denote $a_{\rm max}, M_{\rm max}$ as the maximum values $a_{\sigma,A,M,k}$ and $M$ in $\mathcal{T}_{\rm EM}$. Then $C_m$ concentrates towards some fixed point ${F}_0$ as
\begin{align}
    |C_m - {F}_0| \in \mathcal{O}(2^{-b n} a_{\rm max} |\mathcal{T}_{\rm EM}| M_{\rm max}) \;,
\end{align}
for some constant $b>0$ if the circuit depths satisfies $L_{\sigma} \in \Omega(n)$ for all $\sigma \in \mathcal{T}_{\rm EM}$. 
\end{supplemental_theorem}

In the context of quantum kernels, the fidelity quantum kernel $\kappa^{FQ}(\vec{x},\vec{x'})$ can be estimated by executing a circuit $U^\dagger(\vec{x'})U(\vec{x})$ and then measuring the expectation value of the projection $O = \dya{\psi_0}$. 
Alternatively, we can perform a SWAP test to measure the fidelity kernel via $\kappa^{FQ}(\vec{x},\vec{x'}) = \Tr[(\rho(\vec{x})\otimes\rho(\vec{x'})){\rm SWAP}]$. Here, in the context of error mitigation we can regard $\rho(\vec{x})\otimes\rho(\vec{x'})$ as the state of interest and the SWAP operator as the measurement observable.  

On the other hand, for the projected quantum kernel, a kernel value $\kappa^{PQ}(\vec{x},\vec{x'})$ can be obtained by first estimating $n$ individual terms $\| \rho_k(\vec{x}) - \rho_k(\vec{x'})\|_2^2$ on the quantum computer. Since we consider the reduced states on a single-qubit subsystem, it is efficient to directly estimate $\rho_k(\vec{x})$ and $\rho_k(\vec{x'})$
by measuring the expectation values of Pauli operators $X, Y$ and $Z$ on qubit $k$. Thus, $6$ expectation values are required for each pair of states, leading to $6n$ expectation values in total. Alternatively, $\| \rho_k(\vec{x}) - \rho_k(\vec{x'})\|_2^2$ can be expressed as
\begin{align}\label{eq:appendix-em-pqk}
    \| \rho_k(\vec{x}) - \rho_k(\vec{x'})\|_2^2 = \Tr[\rho_k^2(\vec{x})] + \Tr[\rho_k^2(\vec{x'})] - 2 \Tr[\rho_k(\vec{x})\rho_k(\vec{x'})] \;.
\end{align}
We can then measure the purities and the overlap in~\eqref{eq:appendix-em-pqk} using the SWAP test, leading to 3 expectation values for each pair of states and $3n$ expectation values in total.
After all individual terms are estimated, we can sum them and exponentiate them classically to obtain the kernel value.

In all cases, we can see that measuring quantum kernels in practice requires estimating the expectation values of some observable. Therefore, Supplemental Theorem~\ref{sup-thm:em-noise} can be directly applied. Consequently, EM strategies which prepare the noise-mitigated expectation value according to Eq.~\eqref{eq:appendix-em-general} cannot mitigate the exponential concentration of kernel values due to the effect of noise. We note that, for small $L$ we do not rule out that error mitigation can indeed offer improvements. However, as found in Ref.~\cite{wang2021can}, even when considering fixed system size, error mitigation can often impair resolvability compared to applying no error mitigation at all.

\section{Proof of Proposition~\ref{prop:target-kernel}: Concentration of kernel target alignment}

In this section, we provide a proof of Proposition~\ref{prop:target-kernel}, showing that the concentration of the kernel target alignment in Eq.~\eqref{eq:kernel-ta-mt} can be upper bounded by the concentration of parametrized quantum kernels. We first present some useful lemmas.

\begin{lemma}[Variance of sum of correlated random variables]\label{lemma:var-sum}
For a collection of $N_s$ correlated random variables $\{ R_i \}_{i=1}^{N_s}$, we have
\begin{align}
    \Var\left[\sum_i R_i\right] \leq N_s \sum_i \Var[R_i] \; .
\end{align}
\end{lemma}
\begin{proof}
The variance of the sum of two correlated random variables is given by
\begin{align} \label{bound-var-sum}
    {\rm Var}[R_1 + R_2] & =  {\rm Var}[R_1] +  {\rm Var}[R_2] +  2{\rm Cov}[R_1,R_2] \; ,\\
    & \leq  {\rm Var}[R_1] +  {\rm Var}[R_2]  + 2\sqrt{\Var[R_1]\Var[R_2]}  \; ,\\
    &\leq \Var[R_1] + \Var[R_2] + \sqrt{\Var[R_1]\Var[R_1]} + \sqrt{\Var[R_2]\Var[R_2]} \; , \\
    &= 2\Var[R_1] + 2\Var[R_2] \; ,
\end{align}
where in the first inequality we have used Cauchy-Schwarz, and the second inequality comes from the rearrangement inequality.
Using induction along with the fact that ${\rm Cov}(R_1+R_2,R_3) = {\rm Cov}(R_1,R_3) + {\rm Cov}(R_2,R_3) $, the variance of the full sum can be bounded as presented.
\end{proof}

\begin{lemma}[Variance of product] \label{lemma:var-prod}
Given two correlated random variables $X$ and $Y$, we have
\begin{align}
    \Var[XY] \leq 2\Var[X] |Y^2|_{max} + 2(\mathbb{E}[X])^2\Var[Y]\,, \label{eq:}
\end{align}
where $|Y^2|_{max}$ is the maximum possible value of $Y^2$ i.e. $|Z| _{max} =  \max\{ |Z| : \rm{Pr}(Z) > 0 \}$. 
\end{lemma}
\begin{proof}
We have
\begin{align}
    \Var[X + Y] &= \Var[X] + \Var[Y] + 2\Cov[X,Y] \\
    &\leq \Var[X] + \Var[Y] + 2\sqrt{\Var[X]\Var[Y]} \\
    &\leq \Var[X] + \Var[Y] + \sqrt{\Var[X]\Var[X]} + \sqrt{\Var[Y]\Var[Y]} \\
    &= 2\Var[X] + 2\Var[Y]\,, \label{eq:var-sum}
\end{align}
where in the first inequality we have used Cauchy-Schwarz, and the second inequality comes from the rearrangement inequality. Now consider 
\begin{align}
    \Var[XY] &= \Var\big[(X-\mathbb{E}[X])Y + \mathbb{E}[X]Y\big]\\
    &\leq 2\Var\big[(X-\mathbb{E}[X])Y\big] + 2\Var\big[\mathbb{E}[X]Y\big]\\
    &\leq 2\mathbb{E}\big[ (X-\mathbb{E}[X])^2Y^2 \big] + 2(\mathbb{E}[X])^2\Var[Y]\\
    &\leq 2\mathbb{E}\big[ (X-\mathbb{E}[X])^2 \big] |Y^2 |_{max} + 2(\mathbb{E}[X])^2\Var[Y] \\
    &= 2\Var[X] |Y^2|_{max} + 2(\mathbb{E}[X])^2\Var[Y]\,,
\end{align}
where in the first inequality we have used Eq.~\eqref{eq:var-sum}, in the second inequality we have used the definition of the variance, and in the third inequality we have simply taken the maximum value for $Y^2$. 
\end{proof}

\begin{lemma}\label{lemma:var-taylor}
Given a positive bounded random variable $X$, whose minimum value $|X|_{\min}$ is strictly non-zero, and whose maximum value is $|X|_{\max}$. Then, the following inequality holds
\begin{align}
    \Var\left[1/\sqrt{X}\right] \leq \left (\frac{1}{2 |X|_{\rm min}^3} + \frac{9 (|X|_{\rm max} - |X|_{\rm min} )^2}{32 |X|_{\rm min}^5} \right)\Var[X] \;.
\end{align}
\end{lemma}
\begin{proof}
Let us denote $f(X) = 1/\sqrt{X}$. The truncated Taylor expansion of $f(X)$ around $X_0 = \Ebb[X]$ up to order $p$ can be expressed as
\begin{align}
    f_p(X) & = \sum_{m=0}^{p} \frac{f^{(m)}(X_0) }{m!}  (X-X_0)^m \\
    & = \sum_{m=0}^{p} \left(\frac{(-1)^m(2m)!}{2^{2m} (m!)^2}\right)\cdot \frac{(X-X_0)^m}{X_0^{m+1/2}}\,, \label{eq:taylor-sqrt} 
\end{align}
where $f^{(m)}(X_0)$ is the $m$-order derivative evaluated at $X_0$. The second equality is the result of explicitly computing the derivatives. Truncating the series to the first order ($p=1$) gives $f_1(X) = \frac{X - X_0}{2 X_0^{3/2}}$. 
The difference $R_p(X)$ between $f(X)$ and $f_p(X)$ can be bounded using Taylor’s remainder theorem (see, for example, Chapter 20 of Ref.~\cite{kline1998calculus})
\begin{align}
    R_p(X) & \leq \frac{\max_{Z\in[X_0,X]}|f^{(p+1)}(Z)|}{(p+1)!} |X - X_0|^{p+1} \,.
\end{align}
For $p=1$, we have 
\begin{equation}\label{eq:taylor-remainder}
    R_1(X) \leq \frac{3(X - X_0)^{2}}{8|X|_{\min}^{5/2}}\,.
\end{equation}

Then, the variance of $f(X)$ can then be upper bounded as
\begin{align}
    \Var[f(X)] &= \Var[f_1(X) + R_1(X)] \\
    & \leq 2\Var[f_1(x)] + 2\Var[R_1(X)] \\
    & = \frac{\Var[X]}{2 X_0^3} + 2\Var[R_1(X)] \\
    & \leq \frac{\Var[X]}{2 X_0^3} + 2|R_1|_{\rm max}\Ebb[R_1(X)]\\
    & \leq \frac{\Var[X]}{2 X_0^3} + 2 \left( \frac{3 |X-X_0|^2_{\rm max}}{8|X|^{5/2}_{\rm min}}\right) \cdot \Ebb\left[\frac{3(X - X_0)^{2}}{8|X|_{\min}^{5/2}}\right] \\
    & \leq  \frac{\Var[X]}{2 |X|_{\rm min}^3} + \frac{9 (|X|_{\rm max} - |X|_{\rm min} )^2}{32 |X|_{\rm min}^5} \Ebb\left[ (X-X_0)^2\right] \\
    & = \left(\frac{1}{2 |X|_{\rm min}^3} + \frac{9 (|X|_{\rm max} - |X|_{\rm min} )^2}{32 |X|_{\rm min}^5} \right)\Var[X]
\end{align}
where the first inequality is due to the variance of the sum in Lemma~\ref{lemma:var-sum}, the second equality is due explicitly evaluating $\Var[f_1(X)]$, the second inequality is from $\Var[R_1(X)] \leq \Ebb[(R_1(X))^2] \leq |R_1(X)|_{\rm max}\Ebb[R_1(X)]$, in the third inequality we have used Eq.~\eqref{eq:taylor-remainder},  in the fourth inequality we have used $ |X|_{\rm min}\leq X_0$ for the denominator of the term in the brackets and taken the minimum value of $|X|_{\min}^{5/2}$ in the expectation together with $|X-X_0|^2_{\rm max} \leq (|X|_{\rm max} - |X|_{\rm min})^2 $.
In the last line, we recall that $X_0 = \Ebb[X]$ and hence $\Ebb[(X-X_0)^2] = \Var[X]$.

\end{proof}

We are now ready to prove our proposition relating concentration of the kernel target alignment with the concentration of the kernel, which is recalled below for convenience.
\begin{proposition}[Concentration of kernel target alignment]
Consider an arbitrary parameterized kernel $\kappa_{\vec{\theta}}(\vec{x},\vec{x'})$ and a training dataset $\{ \vec{x}_i, y_i \}_{i=1}^{N_s}$ for binary classification with $y_i = \pm 1$. The probability that the kernel target alignment ${\rm TA}(\vec{\theta})$ (defined in Eq.~\eqref{eq:kernel-ta-mt}) deviates from its mean value is approximately bounded as 
\begin{align}
    {\rm Pr}_{\vec{\theta}}[|{\rm TA}(\vec{\theta}) - \mathbb{E}_{\vec{\theta}}[{\rm TA}(\vec{\theta})]| \geq \delta]  \leq \frac{ M \sum_{i,j}  \Var_{\vec{\theta}}[\kappa_{\vec{\theta}}(\vec{x}_i,\vec{x}_j)]}{\delta^2} \;,
\end{align}
with $M =\frac{8+N_s^3\left(9(N_s-1)^2+16\right)}{4N_s}$. 
\end{proposition}

\begin{proof}
We recall the kernel target alignment ${\rm TA}(\vec{\theta})$ in Eq.~\eqref{eq:kernel-ta-mt}
\begin{align}
    {\rm TA}(\vec{\theta}) & = \frac{\sum_{i,j}y_i y_j \kappa_{\vec{\theta}}(x_i,x_j)}{\sqrt{\left( \sum_{i,j} (\kappa_{\vec{\theta}}(x_i,x_j))^2\right)\left( \sum_{i,j} (y_i y_j)^2\right)}} \\
    & = \frac{\sum_{i,j}y_i y_j \kappa_{\vec{\theta}}(x_i,x_j)}{ \sqrt{D_A(\vec{\theta})\sum_{i,j} (y_i y_j)^2}} \;,
\end{align}
where we define $D_A(\vec{\theta}) = \sum_{i,j} (\kappa_{\vec{\theta}}(x_i,x_j))^2$. We remark that, as the kernel is normalized and the sum is over all the training data, the minimum value of $D_A(\vec{{\theta}})$ (over all possible kernel-based models) happens when $\kappa_{\vec{\theta}}(x_i,x_j) = 0$ for all $i\neq j$, leading to
\begin{align}
    |D_A|_{min} = N_s \;. \label{eq:proof-ta-da-min}
\end{align}
Similarly, the maximum value of $D_A(\vec{{\theta}})$ (over all possible kernel-based models) is upper bounded with the scenario where $\kappa_{\vec{\theta}}(x_i,x_j) = 1$ for all $i$ and $j$, leading to
\begin{align}
    |D_A|_{max} = N^2_s \;.  \label{eq:proof-ta-da-max}
\end{align}

We now consider the variance of the kernel target alignment~\eqref{eq:kernel-ta-mt} and, again, the concentration bound can be shown via Chebyshev's inequality.
\begin{align}
    \Var_{\vec{\theta}}[{\rm TA}(\vec{\theta})] & = \Var_{\vec{\theta}} \left[ \frac{\sum_{i,j} y_iy_j \kappa_{\vec{\theta}}(x_i,x_j)}{\sqrt{D_A(\vec{\theta})\sum_{i',j'} (y_{i'} y_{j'})^2}} \right]  \\
    & \leq N_s^2 \sum_{i,j} \Var_{\vec{\theta}}\left[\frac{y_i y_j\kappa_{\vec{\theta}}(\vec{x}_i,\vec{x}_j)}{\sqrt{D_A(\vec{\theta})\sum_{i',j'} (y_{i'} y_{j'})^2}} \right] \\
    & =  \sum_{i,j} \Var_{\vec{\theta}}\left[\frac{\kappa_{\vec{\theta}}(\vec{x}_i,\vec{x}_j)}{\sqrt{D_A(\vec{\theta})}} \right] \\
    & \leq \sum_{i,j} \left( 2 \Var_{\vec{\theta}}[\kappa_{\vec{\theta}}(\vec{x}_i,\vec{x}_j)] \cdot \left( \frac{1}{|D_A|_{\rm min}}\right) + 2 (\Ebb_{\vec{\theta}}[\kappa_{\vec{\theta}}(\vec{x}_i,\vec{x}_j)])^2 \Var_{\vec{\theta}}\left[ \frac{1}{\sqrt{D_A(\vec{\theta})}}\right]\right) \\
    & \leq   \sum_{i,j} \left( \frac{ 2\Var_{\vec{\theta}}[\kappa_{\vec{\theta}}(\vec{x}_i,\vec{x}_j)]}{N_s} + 2 \Var_{\vec{\theta}}\left[ \frac{1}{\sqrt{D_A(\vec{\theta})}}\right] \right) \label{eq:proof-ta-var-1}
\end{align}
where the first inequality is due to Lemma~\ref{lemma:var-sum}, the second equality is from $\Var[c X] = c^2 \Var[X]$ for a constant $c$ and evaluating $ \frac{(y_i y_j)^2}{\sum_{i',j'}(y_{i'} y_{j'})^2 }= \frac{1}{N_s^2}$ thanks to $(y_iy_j)^2 = 1$ $\forall i,j$ , the second inequality comes from using Lemma~\ref{lemma:var-prod} and the last inequality is due to the fact that $\Ebb_{\vec{\theta}}[\kappa_{\vec{\theta}}(\vec{x}_i,\vec{x}_j)] \leq 1$.

We now focus on the variance of $1/\sqrt{D_A(\vec{\theta})}$. We have
\begin{align}
    \Var_{\vec{\theta}}\left[ \frac{1}{\sqrt{D_A(\vec{\theta})}}\right] & \leq \left (\frac{1}{2 |D_A|_{\rm min}^3} + \frac{9 (|D_A|_{\rm max} - |D_A|_{\rm min} )^2}{32 |D_A|_{\rm min}^5} \right)\Var[D_A(\vec{\theta})] \\
    & = \left( \frac{16+ 9(N_s -1)^2}{32N_s^3}\right) \Var_{\vec{\theta}}\left[\sum_{ij} \kappa^2_{\vec{\theta}}(\vec{x}_i,\vec{x}_j)\right] \\
    & \leq \left( \frac{16+ 9(N_s -1)^2}{32N_s^3}\right) N_s^2 \sum_{i,j}  \Var_{\vec{\theta}}\left[\kappa^2_{\vec{\theta}}(\vec{x}_i,\vec{x}_j) \right] \\
    & \leq \left( \frac{16+ 9(N_s -1)^2}{8N_s}\right)  \sum_{i,j}  \Var_{\vec{\theta}}[\kappa_{\vec{\theta}}(\vec{x}_i,\vec{x}_j)] \label{eq:proof-ta-var-2}
\end{align}
where the first inequality is from using Lemma~\ref{lemma:var-taylor}, the first equality is from substituting $|D_A|_{\rm min}$ in Eq.~\eqref{eq:proof-ta-da-min} and $|D_A|_{\rm max}$ in Eq.~\eqref{eq:proof-ta-da-max}, the second inequality is due to Lemma~\ref{lemma:var-sum}, and finally the third inequality is from using Lemma~\ref{lemma:var-prod} followed by $\Ebb_{\vec{\theta}}[\kappa_{\vec{\theta}}(\vec{x}_i,\vec{x}_j)] \leq 1$.  Substituting Eq.~\eqref{eq:proof-ta-var-2} back into Eq.~\eqref{eq:proof-ta-var-1} leads to 
\begin{align}
    \Var_{\vec{\theta}}[{\rm TA}(\vec{\theta})] \leq \left( \frac{8+N_s^2\left(9(N_s-1)^2+16\right)}{4N_s}\right) \sum_{i,j} \Var_{\vec{\theta}}[\kappa_{\vec{\theta}}(\vec{x}_i,\vec{x}_j)] \;.
\end{align}
Using Chebyshev's inequality leads us to the desired concentration result.

\end{proof}

\section{Sources that lead to exponentially flat landscape of parameterized quantum kernels}\label{appendix:features-ta}
Proposition~\ref{prop:target-kernel} establishes that the training landscape of the kernel target alignment ${\rm TA}(\vec{\theta})$ can be analyzed at the level of the parameterized quantum kernels $\kappa_{\vec{\theta}}(\vec{x},\vec{x'})$. Namely, 
if the training landscape of $\kappa_{\vec{\theta}}(\vec{x},\vec{x'})$ with respect to the variational parameters $\vec{\theta}$ is exponentially flat in the number of qubits $n$, then the training landscape of ${\rm TA}(\vec{\theta})$ also suffers the same fate. In this section, we investigate features of the parameterized data embedding $U(\vec{x},\vec{\theta})$ that lead to an exponentially flat training landscape of the parameterized quantum kernels. In particular, when designing the parameterized quantum kernels,  features that induce barren plateaus in QNNs should be avoided. These include the \textit{expressivity of the training block}, \textit{entanglement}, \textit{global measurements} and \textit{noise}.
We note that, although the proofs of the following results are similar to those in the previous sections, the implication of the results is different. While the kernel concentration in the previous sections happens due to the input data, here the training flat landscape is due to the variational part of the parameterized data embedding.

\subsection{Expressivity}
Similar to the ensemble of data-encoded unitaries over the possible input data, we can define an ensemble of parametrized unitaries $U(\vec{x},\vec{\theta})$ for a given input data $\vec{x}$ over variational parameters $\vec{\theta}$ sampled from a domain $\Theta$. That is, for $\vec{\theta} \in \Theta$, we have the ensemble $\Ubb_{\vec{\theta}}(\vec{x})$ for a given $\vec{x}$
\begin{align}
    \Ubb_{\vec{\theta}}(\vec{x}) = \{ U(\vec{x},\vec{\theta}) | \vec{\theta} \in \Theta \}\,.\label{eq:ensemble-params}
\end{align}
Then, the expressivity can be measured using the superoperator~\eqref{eq:expressivity-measure-mt} with $\Ubb =  \Ubb_{\vec{\theta}}(\vec{x})$.

\begin{supplemental_theorem} 
Consider the parametrized fidelity quantum kernel and the parametrized projected quantum kernel defined in Eq.~\eqref{eq:projected-gaussian-kernel-mt} associated with the parametrized data embedding $U(\vec{x},\vec{\theta})$. For a given pair of input data $\vec{x}$ and $\vec{x'}$, we have
\begin{align}
    {\rm Pr}_{\vec{\theta}}[|\kappa_{\vec{\theta}}(\vec{x},\vec{x'}) - \mu| \geq \delta] \leq \frac{\tilde{G}_n(\varepsilon_{\Ubb_{\vec{\theta}}(\vec{x})},\varepsilon_{\Ubb_{\vec{\theta}}(\vec{x'})})}{\delta^2};, 
\end{align}
where $\mu = \Ebb_{\vec{\theta}} [\kappa_{\vec{\theta}}(\vec{x},\vec{x'})]$ and
$\varepsilon_{\Ubb_{\vec{\theta}}(\vec{x})} = \| \AC_{\Ubb_{\vec{\theta}}(\vec{x})}(\rho_0)\|_1$ is the expressivity measure over $\Ubb_{\vec{\theta}}(\vec{x})$
as defined in Eq.~\eqref{eq:ensemble-params} 
with $\AC_{\Ubb_{\vec{\theta}}(\vec{x})}(\cdot)$ defined in Eq.~\eqref{eq:expressivity-measure-mt}.
\begin{enumerate}
    \item For the fidelity quantum kernel $\kappa_{\vec{\theta}}(\vec{x},\vec{x'}) = \kappa^{FQ}_{\vec{\theta}}(\vec{x},\vec{x'})$, we have
\begin{align}
    \tilde{G}_n(\varepsilon_{\Ubb_{\vec{\theta}}(\vec{x})},\varepsilon_{\Ubb_{\vec{\theta}}(\vec{x'})}) = \beta_{\rm Haar} +  \varepsilon_{\Ubb_{\vec{\theta}}(\vec{x})}\varepsilon_{\Ubb_{\vec{\theta}}(\vec{x'})}+ \sqrt{\beta_{\rm Haar}}\left(\varepsilon_{\Ubb_{\vec{\theta}}(\vec{x})} + \varepsilon_{\Ubb_{\vec{\theta}}(\vec{x'})}\right) \; ,
\end{align}
where $\beta_{\rm Haar} = \frac{1}{2^{n-1}(2^n+1)}$. 
    \item For the projected quantum kernel $\kappa_{\vec{\theta}}(\vec{x},\vec{x'}) = \kappa^{PQ}_{\vec{\theta}}(\vec{x},\vec{x'})$, we have
\begin{align}
    \tilde{G}_n(\varepsilon_{\Ubb_{\vec{\theta}}(\vec{x})},\varepsilon_{\Ubb_{\vec{\theta}}(\vec{x'})}) = 2 \gamma n ( 2 \tilde{\beta}_{\rm Haar} +\varepsilon_{\Ubb_{\vec{\theta}}(\vec{x})} + \varepsilon_{\Ubb_{\vec{\theta}}(\vec{x'})}) \;.
\end{align}
where $\tilde{\beta}_{\rm Haar} = \frac{3}{2^{n+1}+2}$.
\end{enumerate}
\end{supplemental_theorem}

\begin{proof}

The proof follows the same steps as the proof of the extension of Theorem~\ref{thm:expressivity-kernel} in Appendix~\ref{appendix:extension-expressivity} with the integration over $\vec{x}$ and $\vec{x'}$ replaced with the integration over $\vec{\theta}$.
\end{proof}

\subsection{Entanglement}

We show that the entanglement generated via the parametrized data embedding can have a negative impact on the projected quantum kernels. Particularly, the following theorem generalizes Theorem~\ref{thm:entanglement-kernel} for the parametrized projected quantum kernel.

\begin{supplemental_theorem}\label{sup-thm:entanglement-kernel}
Consider the parametrized projected quantum kernel $\kappa^{PQ}_{\vec{\theta}}(\vec{x},\vec{x'})$. Consider a pair of parametrized data-encoded states $\rho(\vec{x},\vec{\theta})$ and  $\rho(\vec{x'},\vec{\theta})$, associated with $\vec{x}$, $\vec{x'}$ and $\vec{\theta}$. Then, we have
\begin{align}
    \left| 1 - \kappa^{PQ}_{\vec{\theta}}(\vec{x},\vec{x'})\right|   \leq (2\ln2) \gamma \Gamma_s(\vec{x},\vec{x'},\vec{\theta}) \; ,
\end{align}
where
\begin{align}
    \Gamma_s(\vec{x},\vec{x'},\vec{\theta}) =  \sum_{k=1}^n \left[ \sqrt{S\left(\rho_k(\vec{x},\vec{\theta})\Big\|  \frac{\mathbb{1}_k}{2}\right)} +  \sqrt{S\left(\rho_k(\vec{x'},\vec{\theta})\Big\|  \frac{\mathbb{1}_k}{2}\right)} \right]^2  \;,
\end{align}
where we denote $S\left(\cdot\|  \cdot \right)$ as the quantum relative entropy, $\rho_k$ as a reduced state on qubit $k$, and $\mathbb{1}_k$ as the maximally mixed state on qubit $k$.
\end{supplemental_theorem}
\begin{proof}
The proof is the same as the proof of Theorem~\ref{thm:entanglement-kernel} in Appendix~\ref{appendix:proof-entangled} with $\rho(\vec{x})$ replaced with $\rho(\vec{x},\vec{\theta})$.
\end{proof}

\subsection{Global measurements}
We argue that the variational part of $U(\vec{x},\vec{\theta})$ should not contain global measurements. This is only relevant to the fidelity quantum kernel since its associated observable is global. On the other hand, global measurements have no impact on projected kernels due to their local construction. 

To illustrate this, we consider the parametrized embedding of the form $U(\vec{x},\vec{\theta}) = U_d(\vec{x})U_p(\vec{\theta})$ where $U_d(\vec{x})$ and $U_p(\vec{\theta})$ can be arbitrary. Supplemental Proposition~\ref{sup-prop:global} then shows that the variance of the parametrized kernel with respect to $\vec{\theta}$ is upper bounded by the variance of an expectation of some global observable.
\begin{supplemental_proposition}\label{sup-prop:global}
Consider the parametrized fidelity quantum kernels $\kappa^{FQ}_{\vec{\theta}}(\vec{x},\vec{x'})$ with the parametrized embedding $U(\vec{x},\vec{\theta}) = U_d(\vec{x})U_p(\vec{\theta})$ where $U_d(\vec{x})$ and $U_p(\vec{\theta})$ are the embedding and parametrized unitary blocks. In addition, consider the decomposition of $U_d^\dagger(\vec{x})U_d(\vec{x'})$ as $\mathcal{M}_R(\vec{x},\vec{x'}) + i \mathcal{M}_I(\vec{x},\vec{x'})$ with $\mathcal{M}_R(\vec{x},\vec{x'})$ and $\mathcal{M}_I(\vec{x},\vec{x'})$ being some Hermitian operators.
Then, we have
\begin{align}
    \Var_{\vec{\theta}}[\kappa^{FQ}_{\vec{\theta}}(\vec{x},\vec{x'})] \leq 4 {\rm max}( \Var_{\vec{\theta}}[a_R],  \Var_{\vec{\theta}}[a_I]) \;,
\end{align}
where $a_R = \Tr[\mathcal{M}_R(\vec{x},\vec{x'})U_p(\vec{\theta}) \rho_0 U_p^\dagger(\vec{\theta})] $ and $a_I = \Tr[\mathcal{M}_I(\vec{x},\vec{x'})U_p(\vec{\theta}) \rho_0 U_p^\dagger(\vec{\theta})] $ with an initial state $\rho_0$.
\end{supplemental_proposition}
\begin{proof}

We are now ready to prove the proposition. Consider the decomposition of $\kappa^{FQ}_{\vec{\theta}}(\vec{x},\vec{x'})$ with an initial state $\rho_0 = \ket{\psi_0}\bra{\psi_0}$
\begin{align}
    \kappa^{FQ}_{\vec{\theta}}(\vec{x},\vec{x'}) = & |\bra{\psi_0}U_p^\dagger(\vec{\theta}) U_d^\dagger(\vec{x})U_d(\vec{x'})U_p(\vec{\theta}) \ket{\psi_0}|^2 \;,\\
    = & |\bra{\psi_0}U_p^\dagger(\vec{\theta}) (\mathcal{M}_R(\vec{x},\vec{x'}) + i\mathcal{M}_I(\vec{x},\vec{x'}) )U_p(\vec{\theta}) \ket{\psi_0}|^2 \;,\\
    = & (\bra{\psi_0}U_p^\dagger(\vec{\theta}) \mathcal{M}_R(\vec{x},\vec{x'}) U_p(\vec{\theta}) \ket{\psi_0})^2 + (\bra{\psi_0}U_p^\dagger(\vec{\theta}) \mathcal{M}_I(\vec{x},\vec{x'}) U_p(\vec{\theta}) \ket{\psi_0})^2  \label{eq:param-kernel-line0}\\
    = & a_R^2 + a_I^2 \;,
\end{align}
where we express $U_d^\dagger(\vec{x})U_d(\vec{x'})$ as $\mathcal{M}_R(\vec{x},\vec{x'}) + i\mathcal{M}_I(\vec{x},\vec{x'})$ with $\mathcal{M}_R(\vec{x},\vec{x'})$ and $\mathcal{M}_I(\vec{x},\vec{x'})$ being some Hermitian matrices. We now upper bound the variance of $\kappa^{FQ}_{\vec{\theta}}(\vec{x},\vec{x'})$.
\begin{align}
    \Var_{\vec{\theta}}[\kappa^{FQ}_{\vec{\theta}}(\vec{x},\vec{x'})] & = \Var_{\vec{\theta}}[a_R^2 + a_I^2]  \\
    & \leq 2 \Var_\theta [a_R^2] + 2 \Var_\theta[a_I^2] \label{eq:param-kernel-line1} \\
    & \leq 8|a_R|_{\rm max}^2 \Var_{\vec{\theta}}[a_R] + 8|a_I|_{\rm max}^2 \Var_{\vec{\theta}}[a_I] \\
    & \leq 8\Var_{\vec{\theta}}[a_R] + 8\Var_{\vec{\theta}}[a_I]\\
    & \leq 16 {\rm max}(\Var_{\vec{\theta}}[a_R], \Var_{\vec{\theta}}[a_I]) \;, 
\end{align}
where the first inequality is due to Lemma~\ref{lemma:var-sum}, the second inequality is due to Lemma~\ref{lemma:var-prod} followed by $\Ebb[X] \leq |X|_{\rm max} $, the third inequality comes from the fact that $a_R$ and $a_I$ are upper bounded by 1 (since $\|\mathcal{M}_R(\vec{x},\vec{x'})\|_{\infty}, \|\mathcal{M}_I(\vec{x},\vec{x'})\|_{\infty} \leq 1$), the last inequality is from choosing the maximum of the two terms.
\end{proof}

It follows from Supplemental Proposition \ref{sup-prop:global} that if $a_R$ and $a_I$ exhibit barren plateaus (with respect to their implicit $\vec{\theta}$ dependence), then $\kappa^{FQ}_{\vec{\theta}}(\vec{x},\vec{x'})$ will also exhibit a barren plateau. 
Since $a_I$ and $a_R$ are linear expectation values of Hermitian operators $ \mathcal{M}_R(\vec{x},\vec{x'})$ and $ \mathcal{M}_I(\vec{x},\vec{x'})$, this allows us to apply barren plateaus results from QNNs to $\kappa^{FQ}_{\vec{\theta}}(\vec{x},\vec{x'})$. In particular, if $U^\dagger_d(\vec{x})U_d(\vec{x'})$ is global and $U_p(\vec{\theta}) $ is a layer hardware efficient ansatz, the results in Ref.~\cite{cerezo2020cost} for global costs imply $\kappa^{FQ}_{\vec{\theta}}(\vec{x},\vec{x'})$ exponentially concentrates around its mean.

\subsection{Noise}

Noise negatively affects the trainability of the parametrized quantum kernels, exponentially flattening the training landscape (with respect to $\vec{\theta}$) of $\tilde{\kappa}_{\vec{\theta}}(\vec{x},\vec{x'})$ and at the same time leading to the exponential concentration (with respect to $\vec{x},\vec{x'}$). The following theorem generalizes Theorem~\ref{thm:noise-kernel} to the noisy parametrized quantum kernels $ \tilde{\kappa}_{\vec{\theta}}(\vec{x},\vec{x'})$

First, we specify the parametrized data embedding to be in the following form
\begin{align}\label{eq:noise-params-embedding}
U(\vec{x},\vec{\theta}) =  \prod_{l=1}^L U_l (\vec{x}_l,\vec{\theta}_l)
\end{align}
We consider the same local Pauli noise model as described in where the noise acts before and after each layer~\eqref{eq:noise-noise-evolution-mt} with the noise characteristic $q$.

\begin{supplemental_theorem}\label{sup-thm:noise-kernel}
Consider the $L$-layered parametrized data embedding defined in Eq.~\eqref{eq:noise-params-embedding} with input state $\rho_0$ and the layer-wise Pauli noise model defined Eq.~\eqref{eq:noise-noise-evolution-mt} with characteristic noise parameter $q\leq 1$. The concentration of quantum kernel values may be bounded as follows
\begin{align}
    \left| \tilde{\kappa}_{\vec{\theta}}(\vec{x},\vec{x'}) - \mu \right| \leq F(q,L) \; . 
\end{align}
\begin{enumerate}
    \item For the fidelity quantum kernel $\tilde{\kappa}_{\vec{\theta}}(\vec{x},\vec{x'}) =\tilde{\kappa}_{\vec{\theta}}^{FQ}(\vec{x},\vec{x'})$, we have $ \mu = 1/2^n$, and
    \begin{align}
        F(q,L) = q^{2L+1}  \left\| \rho_0 - \frac{\mathbb{1}}{2^n} \right\|_2\;.
    \end{align}
    \item For the projected quantum kernel $\tilde{\kappa}_{\vec{\theta}}(\vec{x},\vec{x'}) =\tilde{\kappa}_{\vec{\theta}}^{PQ}(\vec{x},\vec{x'})$, we have $ \mu = 1$, and
    \begin{align}
        F(q,L) = (8 \ln 2) \gamma n  q^{b(L+1)}S_2\left(\rho_0 \Big\|  \frac{\mathbb{1}}{2^{n}}\right) \;,
    \end{align}
    where $S_2( \cdot \| \cdot)$ denotes the sandwiched 2-R\"enyi relative entropy and $b = 1/(2\ln(2)) \approx 0.72 $.
\end{enumerate}
Additionally,  the noisy data-encoded quantum state $\tilde{\rho}(\vec{x},\vec{\theta})$ concentrates towards the maximally mixed state as
\begin{align} 
     \left\| \tilde{\rho}(\vec{x},\vec{\theta}) - \frac{\mathbb{1}}{2^n} \right\|_2 \leq q^{L+1} \left\| \rho_0 - \frac{\mathbb{1}}{2^n} \right\|_2 \; .
\end{align}
\end{supplemental_theorem}
\begin{proof}
The proof is the same as the proof of Theorem~\ref{thm:noise-kernel} in Appendix~\ref{appendix:proof-noise} with $\tilde{\rho}(\vec{x})$ replaced with $\tilde{\rho}(\vec{x},\vec{\theta})$.
\end{proof}

\end{document}